\documentclass[12pt]{article}
\usepackage{authblk}
\usepackage{geometry}
\usepackage{amsmath}
\usepackage{amssymb}
\usepackage{amsthm}
\usepackage{mathrsfs}
\usepackage{subcaption}
\usepackage{customcommands}
\usepackage{bm}
\usepackage{cancel}
\usepackage{style}

\usepackage{physics} 
\usepackage{tensor} 
\usepackage{tcolorbox}
\usepackage{fancyvrb}
\usepackage{cancel}

\usepackage[normalem]{ulem}

\preprint{MIT-CTP/5495}

\title{Maximal Entangling Rates from Holography}

\author[1,2]{\AA{}smund Folkestad}
\author[3,4]{and Aditya Dhumuntarao}
\emailAdd{afolkest@mit.edu}
\emailAdd{aditya@brandeis.edu}
\affiliation[1]{Center for Theoretical Physics, Massachusetts Institute of Technology, Cambridge, MA 02139, USA}
\affiliation[2]{Kavli Institute for Theoretical Physics, University of California, Santa Barbara, CA 93106, USA}
\affiliation[3]{School of Physics and Astronomy, University of Minnesota, Minneapolis, MN, 55455, USA}
\affiliation[4]{Martin A. Fisher School of Physics, Brandeis University, Waltham, MA, 02453, USA}

\abstract{
    We prove novel speed limits on the growth of entanglement, equal-time
    correlators, and spacelike Wilson loops in spatially uniform time-evolving
    states in strongly coupled CFTs with holographic duals.  These bounds can
    also be viewed as quantum weak energy conditions.  Several of the speed
    limits are valid for regions of arbitrary size and with multiple connected
    components, and our findings imply new bounds on the effective entanglement
    velocity of small subregions.  
    In 2d CFT,  our results prove a conjecture by Liu and Suh for a large class of states. 
    We also bound spatial derivatives of entanglement and correlators. 
    Key to our findings is a momentum-entanglement correspondence, showing that entanglement growth
    is computed by the momentum crossing the HRT surface.
    In our setup, we prove a
    number of general features of boundary-anchored extremal surfaces, such as a
    sharp bound on the smallest radius that a surface can probe, and that the
    tips of extremal surfaces cannot lie in trapped regions. Our methods
    rely on novel global GR techniques, including a delicate interplay between
    Lorentzian and Riemannian Hawking masses. While our proofs assume the
    dominant energy condition in the bulk, we provide numerical evidence that
    our bounds are true under less restrictive assumptions. 
}

\begin{document}

\maketitle

\section{\label{sec:intro}Introduction}
Entanglement is one of the key features unique to quantum mechanics, and its
effects are ubiquitous in modern physics. It is now clear that 
entanglement and entanglement entropy is a central quantity across a diverse range of fields, 
such as quantum many-body physics \cite{KitPre05,LevWen05,Ham05}, quantum
information theory \cite{BenWie92,BenBra93,BenBra14}, quantum gravity
\cite{BomKou86,Bek74,
Pag93,RyuTak06a,HubRan07,Van10,Wal11b,Wal12,AlmMar12,MalSus13,FauLew13,EngWal14,BouFis15,AlmEng19,Pen19},
and  quantum field theories and their RG flows
\cite{CasHue04,CalCar05,CasHue06,KleKut07,Bui08,MyeSin10,MyeSin10b,CasHue12,Sol13}.

A central question relevant to all of the above subjects is how entanglement behaves dynamically. 
In this paper, we address the following questions:
are there general bounds on the entanglement entropy in time-dependent states?
Does there exist speed limits on how fast it can grow? The latter question is
relevant to understanding how rapidly quantum information can propagate, how
long it takes a many-body system to thermalize, or, in quantum gravity,
for constraining the dynamics of spacetime itself. 

While calculating entanglement entropies is notoriously hard, many lessons
have been learned over the last two decades. Quantum quenches in particular have received
considerable interest. In a quantum quench, the Hamiltonian is abruptly
changed, or a source is turned on over a small time interval $\delta t$. 
In either case, there is an abrupt injection of energy into the system, kicking the state out of equilibrium. 
The subsequent approach to equilibrium can then be
computed in various setups.  In the seminal paper by Calabrese and Cardy
\cite{CalCar05}, the entanglement entropy $S_R$ of an interval $R$ 
of length $\ell$ in a $(1+1)$-dimensional conformal field theory (CFT) after a uniform quench was
computed, and for large times and interval lengths, it was found to behave as
\begin{equation}\label{eq:introCardy}
\begin{aligned}
    S_R(t)-S_R(t=0) = \begin{cases}
    2 s_{\rm th} t & t<\ell/2 \\
    s_{\rm th}\ell  & t\geq \ell/2
\end{cases},
\end{aligned}
\end{equation}
where $s_{\rm th}$ the thermal entropy density of the final state.  
Linear growth of entanglement for large regions $R$ after uniform quenches has also been found in higher dimensional
holographic CFTs \cite{BalBer10, BalBer10b,LiuSuh13a, LiuSuh13b}. In particular,
after local equilibration and before late time saturation, the entanglement
entropy of a region $R$ after a quench was found to behave as \cite{LiuSuh13a, LiuSuh13b}
\begin{equation}
    \begin{aligned}\label{eq:introve}
    S_R(t)-S_R(t=0) = v_E s_{\rm th} \text{Area}[\partial R] t + \ldots
\end{aligned}
\end{equation}
with $v_E$ the so-called entanglement velocity, which satisfies $v_E\leq 1$. 

While quenches provide useful insights on entanglement dynamics, they do not cover all
kinds of states, and it would be useful to have more general constraints. Some such
results do exist. Consider two quantum systems $A\cup a$ and $B \cup b$ coupled by an
interaction Hamiltonian $H$ acting only on $A$ and $B$. In \cite{VanAco13}
(building on \cite{Bra07}) it was proven that
\begin{equation}\label{eq:SIEC}
\begin{aligned}
    \Big|\frac{ \dd S_{A\cup a} }{ \dd t }\Big| \leq \eta \norm{H} \log d,
\end{aligned}
\end{equation}
where $d=\min\{\dim A, \dim B\}$, and where $\eta$ is an order $1$ constant.
While this bound has broad generality for finite-dimensional systems, it is not
useful in QFT, where $d$ is infinite. 
Even if we UV-regulate to make $d$ finite, $\norm{H}$ is infeasible to compute.
Furthermore, the bound is state-independent, and
it is natural to suspect there exists stronger bounds that depend on the conserved
charges of the state.

A bound more useful in QFT was conjectured \cite{LiuSuh13a,
LiuSuh13b}, based on the findings in holographic quenches. It was proposed that a normalized instantaneous entanglement growth 
$\mathfrak{R}$ in relativistic QFT satisfies the bound
\begin{equation}\label{eq:Rbound}
\begin{aligned}
    \mathfrak{R} \equiv \frac{ 1 }{ \text{Area}[\partial R]s_{\rm th} }\Big|\frac{ \dd
    S_R}{ \dd t } \Big| \leq 1.
\end{aligned}
\end{equation}
In \cite{HarJed15} relativistic QFT was 
used to prove that $\mathfrak{R} \leq 1$ for large
convex regions $R$ in spatially uniform states, neglecting contributions to
$S_R$ not scaling with volume.\footnote{A proof was also given in \cite{CasLiu15} for
half-planes, taking linear growth of entanglement as an assumption. For
quenches, it was proven for large regions holographically in \cite{Mez16},
together with many other properties of quenches.}
However, it was found in \cite{LiuSuh13b} that the largest values for $\mathfrak{R}$ were
obtained for intermediate sized regions, where it could exceed $v_E$
(for $d>2$), and where existing proofs of $\mathfrak{R}\leq 1$ do not apply.
Thus, the validity of \eqref{eq:Rbound} 
for general regions is still an open question.\footnote{In fact, since generic QFTs can
have state-dependent divergences in $S_R$ \cite{MarWal16}, $\partial_t S_R$ can be divergent in
some theories, and so \eqref{eq:Rbound} cannot be true in all relativistic QFTs. This means that a generalization of the proofs of
\cite{HarJed15, CasLiu15} to include contributions not scaling with volume
impossible without more input on the theories under consideration.}

In this work, for holographic CFTs with large coupling and large$-N$ (large effective
central charge), we prove novel bounds that imply $\mathfrak{R}\leq 1$ for a large class of situations not covered by
\cite{HarJed15, CasLiu15}. We also prove several bounds that to our knowledge
have not been previously discussed, including growth bounds on correlators and Wilson
loops.  We will see that our growth bounds can be seen as new types of quantum energy
conditions, valid for uniform states.
We also derive absolute bounds on entanglement entropy and equal-time
correlators.

Let us now summarize our results. Consider first a 2d CFT on $S^1 \times
\mathbb{R}$ or Minkowski space in a homogeneous and isotropic state
undergoing time-evolution. Let $t$ label the timeslices on which the state
is uniform. Let $R$ be a union of $n$ finite intervals of any size.
Assuming an energy condition and certain falloff conditions on the matter fields
in the bulk, which we assume for all bounds presented in the following, we prove that
\begin{equation}\label{eq:introthm1}
\begin{aligned}
    \Big|\frac{ \dd S_R }{ \dd t } \Big| \leq n \sqrt{\frac{ 8\pi c }{ 3
    }\big(\left<T_{tt}\right>-\left<T_{tt}\right>_{\rm vacuum}\big)},
\end{aligned}
\end{equation}
where $c$ is the central charge and $\left<T_{tt}\right>$ the CFT energy density
one-point function, which is the same everywhere in a uniform state. 
If we work with uncharged states,
\eqref{eq:introthm1} implies that $\mathfrak{R}\leq
1$. 
Thus, for the $2d$ theories under
consideration, we have given a proof of $\mathfrak{R}\leq 1$ to regions
of arbitrary finite size and with any number of connected components.\footnote{
For holographic CFTs, \eqref{eq:introthm1}
also improves on a bound proven for single intervals of any size in all $2d$ CFTs by
\cite{HarJed15}, which can be written as 
    $\mathfrak{R} \leq \coth\left(\pi \ell \sqrt{\frac{ \pi c }{ 6
    \left<T_{tt}\right> }} \right)$.
} 

Next, consider $d\geq 2$-dimensional holographic CFTs on Minkowski space, again in
a time-evolving uniform state.
Taking $R$ to be either a single ball or strip of characteristic size $\ell$, we
prove that
\begin{equation}\label{eq:introthm3}
\begin{aligned}
    \Big|\frac{ \dd S_R }{ \dd t } \Big| \leq
    \kappa\text{Vol}[R]\left<T_{tt}\right>\left[1 + \mathcal{O}\left(\frac{
        \ell^{d} \left<T_{tt}\right> }{ c_{\rm eff}  } \right)\right],
\end{aligned}
\end{equation}
where $\kappa$ is an $O(1)$ numerical constant given in \eqref{eq:kappaddef}, and which depends on $d$ and the shape of $R$.
$c_{\rm eff}$ is the effective central charge, to be defined in the following. 
For small regions this bound is much stronger than $\mathfrak{R}\leq 1$.\footnote{See \cite{KunPed16} for a discussion of a different definition of
$\mathfrak{R}$, where $s_{\rm th}$ is replaced by the
vacuum-subtracted entanglement entropy per volume in the final state. With this
definition, $\mathfrak{R}$ is $O(1)$ for small subregions, but it can exceed
$1$.} If 
$\beta$ is the effective inverse temperature at which the thermal energy density
equals $\left<T_{tt}\right>$, we get
\begin{equation}
\begin{aligned}
    \mathfrak{R} \leq \mathcal{O}(\ell/\beta) \ll 1.
\end{aligned}
\end{equation}

We also prove a higher-dimensional analogue of \eqref{eq:introthm1}, although the proof is more limited.
We prove for states that are somewhat more general than quench states that
\begin{equation}\label{eq:introthm2}
\begin{aligned}
    \Big|\frac{ \dd S_R }{ \dd t } \Big| \leq 
    \frac{1}{4}\text{Area}[\partial R] c_{\rm eff}\left[\frac{ 16\pi }{
        (d-1)c_{\rm eff} }\left<T_{tt}\right> \right]^{\frac{ d-1 }{ d }},
\end{aligned}
\end{equation}
where $R$ either is a single ball, or the union of any number of strips. 
Considering a neutral state, \eqref{eq:introthm2} translates into $\mathfrak{R}\leq 1$.
While our proof of \eqref{eq:introthm2} applies to a smaller class of
states, we give substantial numerical evidence that \eqref{eq:introthm2}
holds more generally for all uniform states. 

For strips, we also prove bounds on the entanglement entropy itself. For
$R_{\ell}$ a strip of width
$\ell$ at fixed time $t$, we prove that the vacuum subtracted entropy $\Delta
S(\ell)$ satisfies
\begin{equation}\label{eq:introthmdx1}
\begin{aligned}
    \partial_{\ell} \Delta
    S(\ell) \geq 0,
\end{aligned}
\end{equation}
which in particular implies $\Delta S\geq 0$.

In the special dimensions of $d=2,3,4$, we prove additional bounds.
Assuming the geodesic approximation for correlators \cite{BalRos99}, we prove
that the equal-time two-point function of a scalar operator
$O$ of large scaling dimension $\Delta$ in $d=2$ satisfies
\begin{equation}\label{eq:introthm4}
\begin{aligned}
\Big|\frac{ \dd }{ \dd t }\log \left<O(x)
    O(0)\right>_{\rho(t)}\Big| &\leq \sqrt{\frac{96 \pi \Delta^2 }{ c
    }\big(\left<T_{tt}\right>-\left<T_{tt}\right>_{\rm vac}\big)},
\end{aligned}
\end{equation}
where $\rho(t)$ is the state under consideration. This bound is saturated in
the global CFT$_2$ quenches studied in \cite{CalCar06,CalCar07} (for any
$\Delta$ and $c$). We also prove a tighter bound
when $x$ is small:
\begin{equation}\label{eq:introthm5}
\begin{aligned}
\Big|\frac{ \dd }{ \dd t }\log \left<O(x)
    O(0)\right>_{\rho(t)}\Big| &\leq \frac{ 12\pi \Delta |x|}{ c
    } \big(\left<T_{tt}\right>-\left<T_{tt}\right>_{\rm vac}\big)\left[1 +
    \ldots \right],
\end{aligned}
\end{equation}
where dots indicate $\mathcal{O}\left(x^2\left<T_{tt}\right>/c\right)$
corrections. We furthermore prove a bound on the correlator itself. Letting
$x>0$, we have
\begin{equation}\label{eq:introthmdx2}
\begin{aligned}
    \frac{ \dd }{ \dd x }\ln\left<O(x)O(0) \right>_{\rho(t)}
    \leq  \frac{ \dd }{ \dd x }\ln\left<O(x)O(0)
    \right>_{\text{vacuum}} = -\frac{ 2\Delta }{ x }, 
\end{aligned}
\end{equation}
which shows that for the states covered by our assumptions, correlations between
heavy scalars must die off faster than in the vacuum. 

When $d=3, 4$, we prove bounds on Wilson loops $\mathcal{W}(C)$ of spacelike
circles $C$, assuming we 
can compute these using classical worldsheets in the bulk. Assuming $\mathcal{N}=4$ SYM with gauge group $SU(N)$
and 't Hooft coupling $\lambda$ on the boundary, we show that\footnote{For other
potential $d=4$ holographic CFTs, our result can be written in terms of the effective central
charge and effective coupling -- see main text.}
\begin{equation}\label{eq:Wsym}
\begin{aligned}
    \Big|\frac{ \dd }{ \dd t  }\log\left<\mathcal{W}(C)\right>_{\rho(t)}\Big| &\leq
    \text{Length}[C]\sqrt{\frac{ 2\lambda }{ 3N^2
    }\left< T_{tt} \right>}, \qquad d=4.
\end{aligned}
\end{equation}
In $d=3$, we prove a similar result, but 
for the more restricted set of states which includes global quenches (see \eqref{eq:d3Wilson}).
For small Wilson loops, we also have stricter bounds, which we give in the main
text (see \eqref{eq:Wsmallr}). 

How are these bounds proven? Let us give the broad picture, restricting
to the time-derivative of the
entanglement entropy of a strip for concreteness. For CFTs dual to classical Einstein gravity, the von Neumann entropy of the reduced state $\rho_R$ on a
subregion $R$ is given by the HRT formula \cite{RyuTak06a, RyuTak06b,HubRan07}, which says that
\begin{equation}
\begin{aligned}
    S_R = \frac{ \text{Area}[X] }{ 4G_N },
\end{aligned}
\end{equation}
where $X$ is the HRT surface in the gravitational bulk, which roughly means a codimension 2 spacelike
surface that has stationary area under perturbations of $X$ in the bulk interior.
Bounding $\partial_t S_R$  in uniform states now corresponds to bounding
$\partial_t \text{Area}[X_t]$, where $X_t$ is a one-parameter family of HRT surfaces
 living in general time-dependent spacetimes with
planar symmetry. Key to our proofs then is carrying out the analysis locally on a
planar symmetric spatial slice $\Sigma$ that contains $X_{t}$. 
We then show that the change in entanglement entropy is given by
\begin{equation}\label{eq:dSdtintro}
\begin{aligned}
    \frac{ \dd S_R}{ \dd t } = \int_{X} G P ,
\end{aligned}
\end{equation}
where $P$ is the matter momentum density in a direction
orthogonal to HRT surface, and $G$ essentially a propagator that only depends on the smallest radius probed
by $X$, and not any other details of the spacetime. We thus see that the flux of matter
falling out of the entanglement wedge is directly responsible for the increase of entanglement entropy. The formula
\eqref{eq:dSdtintro} can be seen as momentum-entanglement correspondence, 
analogue to the momentum-complexity correspondence proposed in \cite{Sus18} and
given a precise form in \cite{BarJos19,BarMar20,BarMar21}. 
To further leverage this formula to get our proofs, we study two quasilocal
masses and find in certain dimensions the integral in \eqref{eq:dSdtintro}
is exactly encoded in the difference between these two quasilocal masses at infinity.
A detailed analysis of the monotonicity properties of these masses
under various flows then lets us prove our final bounds, essentially
using a combination of Lorentzian and Riemannian inverse mean curvature flows.
We emphasize that beyond our assumed symmetries, we do not need to assume a particular form of the spacetimes
we are considering, and we are certainly not restricted to quenches for our most
general bounds.

Along the way we derive various general properties of the HRT surfaces of strips
and spheres in planar symmetric spacetimes. 
For example, for $d=2$ we prove that the radius $r_0$ of the tip of the HRT
surface of a strip of width $\ell$ satisfies
\begin{equation}
\begin{aligned}
    r_0 \geq \frac{ 2L^2 }{ \ell },
\end{aligned}
\end{equation}
where $L$ the AdS radius.
We prove similar bounds in higher dimensions. We also prove that the
tip of an HRT surface of a sphere or a strip can never lie in a trapped region
in spacetime. The same is shown for boundary anchored extremal surfaces of
dimension $q+1$ anchored at $q$-spheres.

This paper is organized as follows. In Sec.~\ref{sec:strip} we set up our
assumptions and prove all our entanglement growth bounds for strip subregions
$R$. In Sec.~\ref{sec:generaldim} we prove the entanglement growth bounds for ball
shaped regions $R$ and furthermore derive general properties $(q+1)$-dimensional extremal surfaces
anchored at $q$-dimensional spheres on the boundary, leading to our results
for correlators and Wilson loops. In Sec.~\ref{sec:spatialbnds} we prove bounds
on spatial derivatives of the entanglement entropy of strips and equal-time
two-point correlators in $d=2$. In Sec.~\ref{sec:numerics}, for a subset of
our bounds, we give significant numerical evidence that the dominant energy condition, which
was assumed for our proofs, can be replaced by less restrictive
assumptions. Finally, in Sec.~\ref{sec:discussion}, we conclude with a discussion of the
implications of our findings, together with future directions.
For a reader only wanting to understand the results without getting into the details
of the proofs, it is possible to only read sections
\ref{sec:stripsummary}, \ref{sec:qsummary}, 
\ref{sec:spatialbnds}, \ref{sec:numerics}, and \ref{sec:discussion}.

\textbf{Note added in v3:} In previous versions of this paper, stronger growth
bounds were given for strips in the case when $S - S_{\rm vac}\leq 0$. However, these
bounds were vacuous, since for strips, we now have a proof that $S - S_{\rm vac}\geq0$ under our assumptions. 
Furthermore, the bounds \eqref{eq:introthmdx1},
and \eqref{eq:introthmdx2} were added in v3.

\section{Maximal Entanglement Rates for Strips}\label{sec:strip}
\subsection{Setup and summary of results}\label{sec:stripsummary}
Consider a $d$-dimensional holographic CFT in Minkowski space dual to classical Einstein
gravity. Consider now some general time-evolving state $\rho(t)$ possessing
a geometric dual, and having spatially homogeneous and isotropic one-point functions for
local operators dual to bulk fields, such as the CFT stress tensor $T_{ij}$. 
Homogeneity and isotropy ensures that the dual asymptotically AdS$_{d+1}$ spacetime
$(\mathcal{M}, g)$ has planar symmetry. We allow
$\rho(t)$ to live on either one or two copies of Minkowski space, so that
the dual spacetime can have either one or two asymptotic
boundaries. For a single system, we allow $\rho(t)$ to be mixed.\footnote{Allowing
two-sided spacetimes means that automatically allow
mixed states on a single CFT, since we can always find a purification dual to a
wormhole, simply by gluing a second CPT-conjugate copy of the spacetime to itself along the
HRT surface \cite{EngWal19}.}

Our goal in this section is to use the HRT entropy formula in this setup to
derive a speed limit on the growth of the entanglement for a strip, and in some
cases the union of any number of strips, provided they all live on the same
connected component of the conformal boundary. In Sec.~\ref{sec:generaldim} we
will generalize to spherical subregions, and to Wilson loops and two-point
correlators. However, we will present the results on entanglement growth for
spherical regions in this section, since they naturally are presented together
with the results for strips.

Before presenting our results, let us set up our assumptions.
We will assume that our spacetimes are AdS-hyperbolic, meaning that we can
foliate $(\mathcal{M}, g_{ab})$ with spacelike hypersurfaces $\Sigma_t$ that all have the same topology and are geodesically complete as Riemannian
manifolds. These represent moments of time. Next, letting $L$ be the asymptotic
AdS radius, we assume that $(\mathcal{M}, g_{ab})$ satisfies the Einstein
equations 
\begin{equation}
\begin{aligned}
    R_{ab} - \frac{ 1 }{ 2 }g_{ab}R - \frac{ d(d-1) }{ 2L^2 }g_{ab} = 8\pi G_N
    \mathcal{T}_{ab},
\end{aligned}
\end{equation}
and that the dominant energy condition (DEC) holds for the bulk stress tensor
$\mathcal{T}_{ab}$, meaning that
\begin{equation}
\begin{aligned}
    \mathcal{T}_{ab}u^a v^b \geq 0 \qquad \text{ for all timelike } u^a, v^b.
\end{aligned}
\end{equation}

Next, we assume that the 
Balasubramanian-Kraus \cite{BalKra99} boundary stress tensor $\expval{T_{ij}}$ is
finite. 
When it is finite, it corresponds to the one-point function of the CFT stress tensor. 
To specify falloff assumptions more explicitly, let $\Omega$ be any defining
function, meaning any function on the conformal compactification of
$\mathcal{M}$ such that the pullback of $\Omega^2 g_{ab}|_{\partial \mathcal{M}}$ to the conformal boundary
is a Lorentzian metric. We then require that the bulk stress tensor satisfies
\begin{equation}\label{eq:falloffs}
\begin{aligned}
    \mathcal{T}_{ab}u^a v^b \sim o\left(\Omega^{d}\right), \qquad \forall \text{ unit vectors }
    v^a, u^a,
\end{aligned}
\end{equation}
near the conformal boundary $\partial \mathcal{M}$. In the radial coordinate $r$
introduced below, this means the
stress tensor in an orthonormal basis falls off as $o(r^{-d})$.
Matter fields with falloffs sufficiently slow to require modifications of the
definition of the spacetime mass are not covered by our results.\footnote{In this case, depending
on how slow the falloffs are, subleading divergences in the entropy might become
state dependent \cite{MarWal16}, in which case $\mathfrak{R}\leq 1$ cannot
remain true. See discussion in Sec.~\ref{sec:discussion}.} 
To avoid having to repeat the same assumptions in every theorem, let us define
the following:
\begin{defn}
    We say that an AAdS$_{d+1}$ spacetime $(\mathcal{M}, g_{ab})$ is regular if it is
    AdS-hyperbolic, has falloffs \eqref{eq:falloffs}, and $g_{ab}$ is $C^2$.
\end{defn}

For index conventions, we will take $a,b,\ldots$ to be abstract spacetime indices,
and $\alpha,\beta,\ldots$ to be abstract indices on spacelike hypersurfaces
$\Sigma$. We take $\mu, \nu, \ldots$ to be coordinate indices on $\Sigma$. Other
indices should be clear in the context. Furthermore, whenever intrinsic tensors on
submanifolds are written with spacetime indices, we mean the pushforward/pullback to
spacetime using the embedding map.

To describe the boundary regions covered by our results, we select a Minkowski conformal frame on the conformal boundary with coordinates
\begin{equation}
\begin{aligned}
    \dd s^2|_{\partial \mathcal{M}} = -\dd t^2 + L^2(\dd \phi^2 + \dd \bm{x}^2), \qquad
    \phi \in \mathbb{R},\quad \bm{x} \in \mathbb{R}^{d-2},\label{eq:frame1}
\end{aligned}
\end{equation}
where the constant $t$-slices are the ones on which we have uniform one-point
functions for local operators. For $d=2$ we can allow $\phi$ to be periodically identified, in which
case we say that $\mathcal{M}$ has spherical symmetry.
If $\partial \mathcal{M}$ has two connected components, we focus on a particular one.
We define $R_{t'}$ to be the one-parameter family of boundary regions given by 
\begin{equation}\label{eq:stripregiondef}
\begin{aligned}
    -\frac{ \ell }{ 2L } \leq \phi \leq  \frac{ \ell }{ 2L }, \qquad t=t',
\end{aligned}
\end{equation}
which just corresponds to a strip or interval of length $\ell$ at time $t'$.
In this section, when we talk about strips or refer to a one-parameter family, we always mean the
family \eqref{eq:stripregiondef}. 
We will abbreviate $R_{t=0}\equiv R$, and define 
\begin{equation}
\begin{aligned}
    \text{Area}[\partial R_t]  = \text{Area}[\partial R]  =
    L^{d-2}\int_{\mathbb{R}^{d-2}}
    \dd^{d-2}\bm{x},\qquad d>2,
\end{aligned}
\end{equation}
while for $d=2$, we have $\text{Area}[\partial R]=2$. For $d>2$ this is of course divergent, 
but since it always appears as an overall prefactor it causes no difficulties.

Next, the HRT formula \cite{RyuTak06a, RyuTak06b,HubRan07} states that the von Neumann entropy of the
reduced CFT state on $R_t$, $\rho_R(t) \equiv \tr_{R^{\text{c}}}\rho(t)$, is given by
\begin{equation}\label{eq:HRTdef}
\begin{aligned}
    S_R(t) = -\tr\left[\rho_R(t) \ln \rho_R(t)\right] = \frac{ \text{Area}[X_t] }{ 4G_N },
\end{aligned}
\end{equation}
where $X_t$ is the minimal codimension-2 spacelike surface in $(\mathcal{M},
g_{ab})$ that is
(1) a stationary point of the area functional (i.e. extremal), (2) anchored at $\partial R_t$ on the
conformal boundary ($\partial X_t = \partial R_t$), and (3) homologous to $R_t$. The latter means that there
exists spacelike hypersurface $\Sigma$ with $\partial \Sigma = X_t \cup
R_t$, where we here mean the boundary in the conformal completion.
We will use the gravitational description to derive an upper bound on
\begin{equation}\label{eq:whatwebound}
\begin{aligned}
    \qty|\frac{ \dd }{ \dd t }\left(\frac{ \text{Area}[X_t] }{ 4G_N
    }\right)|
\end{aligned}
\end{equation}
purely in terms of quantities that have a known interpretation in the CFT. 
While $\text{Area}[X_t]$ is formally divergent, since we (1) work with
spacetimes with falloffs \eqref{eq:falloffs} and (2) $\text{Area}[\partial
R_t]$ is time-independent, \eqref{eq:whatwebound} is in fact finite up to
the $\text{Area}[\partial R]$ prefactor.

Let us now summarize our main results, which are broadly divided into two
categories. The first class of bounds scales like $\text{Area}[\partial R]$, and
they are strongest when $R$ is large. The second class of bounds scales like
$\text{Vol}[R]$, and they are consequently the strongest for small subregions.
For intermediate sized regions, where the entanglement entropy is about the
enter the volume-scaling regime, we expect the two types of upper bounds to be
roughly comparable. 

First, for a three-dimensional bulk, we obtain the following
\begin{thm}\label{thm:mainthm1}
    Let $(\mathcal{M},g_{ab})$ be a regular asymptotically AdS$_{3}$ spacetime with planar
    or spherical symmetry satisfying the DEC. Assume that $X_t$ is the HRT surface of
    a finite interval $R_t$. Then
    \begin{equation}\label{eq:mainthm1}
    \begin{aligned}
        \qty|\frac{ \dd }{ \dd t }\left(\frac{{\rm Area}[X_t] }{ 4G_N
        }\right)| \leq \sqrt{\frac{ 8\pi c }{ 3 }\big(\left< T_{tt}
        \right>-\left< T_{tt} \right>_{\rm vac}\big)},
    \end{aligned}
    \end{equation}
    where $c=\frac{ 3L }{ 2G_N }$.
\end{thm}
\noindent Since the HRT surface of a union of strips is just the union of HRT
surfaces of a collection of individual strips, this bound immediately implies that if $R$ is a union of $n$
intervals contained in a single moment of time on one of the connected
components of $\partial \mathcal{M}$, then
\begin{equation}
\begin{aligned}
    \qty|\frac{ \dd }{ \dd t }S_R| \leq n\sqrt{\frac{ 8\pi c }{ 3
    }\big(\left< T_{tt}
        \right>-\left< T_{tt} \right>_{\rm vac}\big)}.
\end{aligned}
\end{equation}
While we are not able to give a general proof of the analogue of Theorem~\ref{thm:mainthm1} in
higher dimensions, we prove a generalization in thin-shell spacetimes: 
\begin{thm}\label{thm:mainthm2}
    Let $(\mathcal{M},g_{ab})$ be an asymptotically AdS$_{d+1 \geq 3}$ spacetime with planar
    symmetry satisfying the DEC. Assume that $X_t$ is the HRT surface of a
    region $R_t$ corresponding to either a finite width strip or a ball.
    Assume that the  bulk matter consists of $U(1)$ gauge fields and a thin shell of matter:
    \begin{equation}
    \begin{aligned}
    \mathcal{T}_{ab} = \mathcal{T}_{ab}^{\rm shell} +
        \mathcal{T}_{ab}^{\rm Maxwell},
    \end{aligned}
    \end{equation}
    where $\mathcal{T}_{ab}^{\rm shell}$ has delta function
    support on a codimension$-1$ worldvolume that is timelike or null, and with
    $\mathcal{T}_{ab}^{\rm shell}$ separately
    satisfying the DEC. Assume $(\mathcal{M},g_{ab})$ is regular, except we do
    not require $g_{ab}$ to be $C^2$ at the shell.
    Then
    \begin{equation}\label{eq:shellnonlinearT}
    \begin{aligned}
    \qty|\frac{ \dd }{ \dd t }\left(\frac{ {\rm Area}[X_t] }{ 4G_N
        }\right)|  \leq \frac{ 1 }{ 4 }{\rm Area}[\partial 
    R]c_{\rm eff} \left[\frac{ 16\pi }{ (d-1)c_{\rm eff}
    }\left<T_{tt}\right> \right]^{\frac{ d-1 }{ d }},
    \end{aligned}
    \end{equation}
    where $c_{\rm eff} = L^{d-1}/G_N$.
\end{thm}
\noindent 
This theorem applies to thin-shell Vaidya spacetimes and charged
generalizations. These spacetimes (and related setups) have been studied
extensively
\cite{AbaApa10,BalBer10,BalBer10b,BarGal12,CacKun12,GalSch12,SheSta13,SheSta13b,LiuSuh13a,LiuSuh13b,LiWu13,AreAnd13,CapMan13,Uga13,AliMoh14,FonFra14,KerNis2014,AliFar14,CacKun14,BucMye14,StaSus14,BabMoh20,LeiMoo15,EckGru15,Zio15,Soh15,CacSan15,RanRoz15,KunPed16,Roy16,Eck16,Mez16,ObaPro16,LokOli17,MyeRoz17,AreKhr17,Jah17,FloErd17,Avi18,GhaYar18,FisJah18,Mez18,KudMac19,CouEcc19,LinLiu19,MezWil20,GotNoz21} as holographic models CFT quenches.  However, more general cases than Vaidya are allowed, where the shell might correspond to some brane in the
bulk, propagating in a timelike direction. Using the duality between radius and
scale in the CFT, thin shell spacetimes correspond to CFT states
where all dynamics is happening at a single scale (that evolves with time).
We also should note that
 \eqref{eq:shellnonlinearT} holds if $R$ is a union of any number of strips on
 the same conformal boundary,
 due to the fact that the HRT surface of $n$ strips is just equal to $n$
 HRT surfaces of $n$ (generally different) strips.
One can hope that this might also be true for multiple spheres, but this does not follow from our current analysis.
Also, while we do not have a proof, we conjecture that
\eqref{eq:shellnonlinearT} is valid in all DEC respecting regular planar
symmetric AAdS$_{d+1}$ spacetimes, and we provide strong numerical evidence for this in Sec.~\ref{sec:numerics}.

Also, note that $c_{\rm eff}$ can be defined purely in CFT in terms of a
universal prefactor of the sphere vacuum entanglement entropy \cite{RyuTak06a,RyuTak06b}, or in terms of the
renormalized entanglement entropy \cite{LiuMez12,LiuMez13}. So our final bounds
on $|\partial_t S|$ make no reference to the bulk.

The previous two results give upper bounds scaling like $\text{Area}[\partial
R]$.  Now let us turn to bounds scaling like $\text{Vol}[R]$. We prove the
following bound on small regions, valid for all $d\geq2$: 
\begin{thm}\label{thm:mainthm3}
    Let $(\mathcal{M},g_{ab})$ be a regular asymptotically AdS$_{d+1 \geq 3}$ spacetime with planar
    symmetry satisfying the DEC. Assume that $X_t$ is the HRT surface of a
    region $R_t$ corresponding to either a strip or a ball.
    Let $\ell$ be either the strip width or ball radius, and assume that
    \begin{equation}
    \begin{aligned}
        \frac{ \ell^{d} \left<T_{tt}\right> }{ c_{\rm eff} } \ll 1.
    \end{aligned}
    \end{equation}
    Then
    \begin{equation}
    \begin{aligned}\label{eq:entgrowthbound}
        \qty|\frac{ \dd }{ \dd t }\left(\frac{ {\rm Area}[X_t] }{ 4G_N
        }\right)| \le
        \kappa_d  
        {\rm Vol}[R]\left<T_{tt} \right>\left[1 +
        \mathcal{O}\left(\frac{ \left<T_{tt}\right>\ell^{d} }{ c_{\rm eff}
        }\right)\right].
    \end{aligned}
    \end{equation}
    where
    \begin{equation}\label{eq:kappaddef}
    \begin{aligned}
    \kappa_d = \frac{\Gamma\qty(\frac{1}{2(d-1)})}{\Gamma\qty(\frac{d}{2(d-1)})} 
    \begin{cases}
        2\pi & R {\rm \ is\ an\ interval\ and\ } d=2, \\
        \frac{\sqrt{\pi}}{d-1} & R {\rm\ is\ a\ strip\ and\ } d>2,  \\
        2\sqrt{\pi} &R {\rm\ is\ a\ ball\ and\ } d>2. \\
    \end{cases}
    \end{aligned}
    \end{equation}
\end{thm}
\noindent Next, for thin shell spacetimes, volume-type bounds can be proven exactly for subregions
of any size, at the cost of a slightly larger prefactor:
\begin{thm}\label{thm:mainthm4}
    Consider the same setup as in Theorem~\ref{thm:mainthm2}.
    Then
    \begin{equation}\label{eq:entgrowthbound2}
    \begin{aligned}
        \qty|\frac{ \dd }{ \dd t }\left(\frac{ {\rm Area}[X_t] }{ 4G_N
        }\right)| \le \kappa_{d}' {\rm Vol}[R]\left<T_{tt} \right>,
    \end{aligned}
    \end{equation}
    with
    \begin{equation}
    \begin{aligned}
    \kappa'_{d} = d^{-\frac{ d }{ 2(d-1) }}\frac{\Gamma\qty(\frac{1}{2(d-1)})}{\Gamma\qty(\frac{d}{2(d-1)})}
     \begin{cases}
         2\pi & R {\rm \ is\ an\ interval\ and\ } d=2, \\
         \sqrt{\frac{4\pi}{(d-1)}} & R {\rm\ is\ a\ strip\ and\ } d>2,  \\
         \sqrt{16\pi(d-1) } & R {\rm\ is\ a\ ball\ and\ }  d>2.
    \end{cases}
    \end{aligned}
    \end{equation}
\end{thm}
\noindent For a strip, a few values of the prefactors are
\begin{equation}
\begin{aligned}
\kappa_d = \begin{cases}
    2\pi & d=2 \\
    2.62\ldots & d=3 \\
    2.43\ldots & d=4 \\
    2 & d=\infty \\
\end{cases}, \qquad
\kappa_d' = \begin{cases}
    2\pi & d=2 \\
    3.25\ldots & d=3 \\
    3.34\ldots & d=4 \\
    4 & d=\infty. \\
\end{cases}
\end{aligned}
\end{equation}

We will now outline the strategy used to obtain these bounds. First, we observe
that there exists exactly one homology hypersurface $\Sigma_t$ that both
contains $X_t$, and which respects the planar symmetry of $(\mathcal{M}, g_{ab})$.
Then we show that the location of $X_t$ on $\Sigma_t$ can be solved for exactly
in terms of the intrinsic geometry on $\Sigma_t$. Together with the DEC, this
fact allows us to lower bound the radius of the tip
of the HRT surface. Next, we use the fact that since $X_t$ is extremal,
the first order variation of its area is a pure boundary term located at $\partial
X_t$ \cite{BaoCao19}, and we show that this boundary term is simply given by a particular
component of the extrinsic curvature of $\Sigma_t$ as $r\rightarrow \infty$.
Then we work out the form of Einstein constraint equations on $\Sigma_t$, and
show that the relevant extrinsic curvature component can be written as an
integral of the matter flux over the HRT surface. Finally, essentially relying
on inverse mean curvature flow of Lorentzian and Riemannian Hawking
masses, and their monotonicity properties under these flows, we bound the integrated matter 
flux across the HRT surface from above in terms of the mass of the spacetime.

Now, before we dive in, we should clarify the meaning of radii in planar symmetric spacetimes. Since we
have planar symmetry, spacetime has a two-parameter foliation where each leaf is
a codimension$-2$ spacelike plane that has the usual flat intrinsic metric. When we talk
about a plane, we always mean one of these leafs. These planes can
all be assigned an ``area radius'' $r$, and it is possible to view $r$ as a
scalar function on spacetime which is not tied to any coordinate. 
Nevertheless, unlike in spherical symmetry, there is an overall scaling ambiguity 
in this function, since the non-compactness of the planes means we cannot
normalize $r$ to some area -- there is no ``unit plane''.
However, if we choose some Minkowski conformal frame on the boundary, we can
fix the overall normalization of $r$ by demanding that the defining function
$\Omega$ that takes us to the chosen conformal frame is $\Omega = r/L$. We will
implicitly assume such a choice, and refer to \textit{the} radius of a
plane.

\subsection{An explicit solution for the HRT surface location}\label{sec:Xlocsol}

\begin{figure}
    \centering
    \includegraphics[width=.5\textwidth]{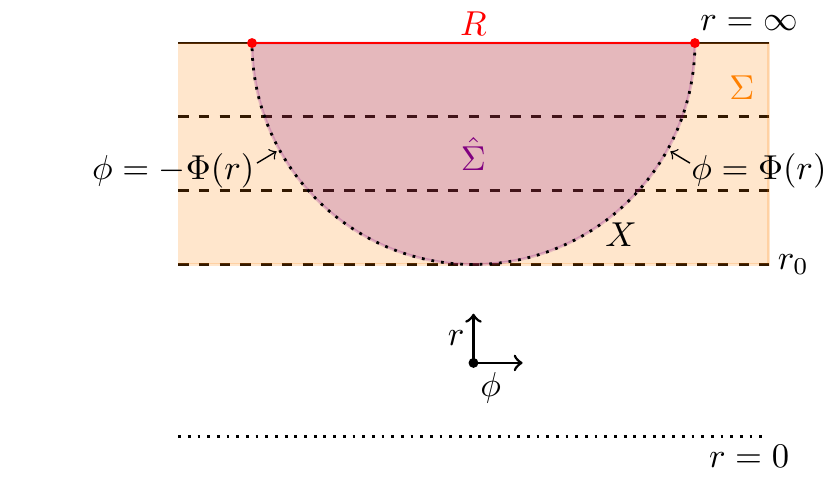}
    \includegraphics[width=.32\textwidth]{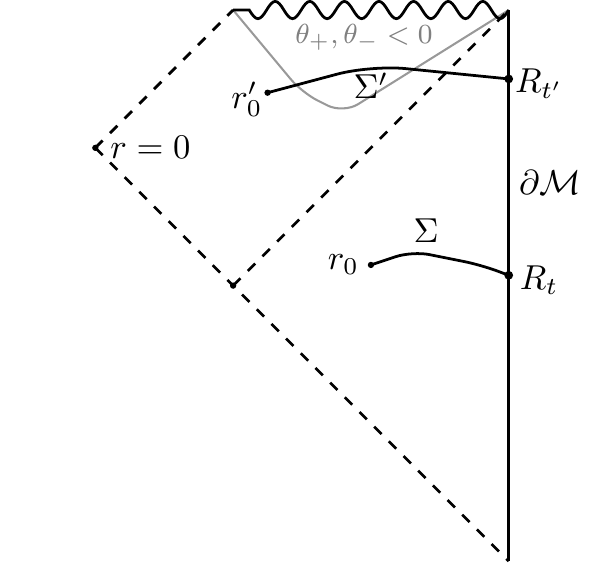}
    \caption{Left: the planar symmetric homology hypersurface $\hat{\Sigma}$
    with respect to the HRT surface $X$. $\Sigma$ is the
    extended homology hypersurface, whose boundary is the plane at $r=r_0$.
    Dashed lines are planes -- i.e. constant $r$ surfaces. 
    Right: example conformal diagram indicating
    possible embeddings of two extended homology hypersurfaces $\Sigma$ and
    $\Sigma'$. The grey line is an apparent horizon, with vanishing outwards
    null expansion, $\theta_{+}=0$.}
    \label{fig:planarHRTs}
\end{figure}

Without loss of generality, we will bound the time-derivative at $t=0$ and use
the shorthands $X_{t=0}=X$ and $R_{t=0}=R$. 
Since $R$ is a strip contained in a canonical time slice of
Minkowski, and since the ambient spacetime has planar symmetry,
there exists a homology hypersurface $\hat{\Sigma}$ of $X$ respecting the planar symmetry
-- see Figure \ref{fig:planarHRTs}. We can pick coordinates on $\hat{\Sigma}$ so that its
induced metric
$H_{\alpha\beta}$ reads
\begin{equation}
\begin{aligned}
   H_{\mu\nu}\dd y^{\mu}\dd y^{\nu} = B(r) \dd r^2 + r^2(\dd \phi^2 + \dd \bm{x}^2), \quad
    r \in [r_0, \infty),\quad \phi \in [-\Phi(r), \Phi(r)], \label{eq:cancoords}
\end{aligned}
\end{equation}
where $\phi=\Phi(r)$ is the coordinate embedding function of (half of) the HRT surface in
$\hat{\Sigma}$, as illustrated in Figure \ref{fig:planarHRTs}. $r_0$ is the smallest value of $r$ probed
by the HRT surface, corresponding to its tip. $\hat{\Sigma}$ can naturally
be extended to include all $\phi \in \mathbb{R}$ by planar symmetry, and this choice turns out to
be convenient for us. We denote the corresponding hypersurface as $\Sigma$,
and refer to it as the extended homology hypersurface.\footnote{
$\hat{\Sigma}$ and $\Sigma$ are unique. 
Planar symmetry means that $(\mathcal{M}, g_{ab})$ can be foliated by planes, and every point $p \in X$ lies in
some plane in this foliation. Demanding planar symmetry of $\Sigma$ requires that the
full leaf intersected by $p$ is included in $\Sigma$, and so we have a one-parameter
family of codimension$-2$ surfaces picked out by $X$, which thus fully specifies
$\Sigma$, and similarly for $\hat{\Sigma}$.}  
See Figure ~\ref{fig:planarHRTs}. The boundary of $\Sigma$ (in the bulk proper) is a plane of radius $r_0$.

Relying on the formulas derived in the remainder of this section, we prove the
following Lemma in appendix~\ref{app:nothroat}:
\begin{lem}\label{lem:cancoords}
    Let $\Sigma$ be the extended homology hypersurface of an HRT surface $X$
    anchored at a strip region given by \eqref{eq:stripregiondef}. Then a single coordinate system of the form
    \begin{equation}
    \begin{aligned}
        \dd s^2 = B(r)\dd r^2 + r^2 \dd \bm{x}^2\label{eq:cancoords2}
    \end{aligned}
    \end{equation}
    is enough to cover all of $\Sigma$. Furthermore, $X$ has only
    one turning point, meaning that embedding function $r(\phi)$ is monotonically increasing for $\phi \geq 0$.
\end{lem}
\noindent This means one function $\Phi(r)$ contains all the information about the
embedding of $X$ in $\Sigma$ -- we do not need multiple branches. It also means that $\Sigma$ cannot have any locally
stationary planes -- that is -- no planes of vanishing mean curvature, where
$B(r)$ would blow up.  The means we never need to worry about patching across coordinate systems when
working on $\Sigma$. Geometrically, it implies that $\Sigma$ has no
``throats''.

Now, taking $(r, \bm{x})$ to be coordinates on $X$, the induced metric on $X$ reads
\begin{equation}\label{eq:coordsHRT}
\begin{aligned}
    \dd s^2|_{X}  = \left[B(r)+r^2 \Phi'(r)^2\right]\dd r^2 + r^2 \dd \bm{x}^2.
\end{aligned}
\end{equation}
Since $X$ is an extremal surface, its area is stationary under all variations,
including under variations within $\Sigma$. Enforcing this gives an ODE for $\Phi(r)$ in terms
of $B(r)$. To find it, we compute the mean curvature $\mathcal{K}$ of $X$ viewed
as a submanifold of $\Sigma$
and demand it to be zero. This gives the equation (see appendix~\ref{app:meancurv} for a
computation)
\begin{equation}\label{eq:planarAeq}
    r B \Phi'' + (d-1) r^2 \left(\Phi'\right)^3 + \Phi' \left(d B-\frac{r}{2}
    B'\right)= 0.
\end{equation}
The relevant boundary conditions are
\begin{equation}
\begin{aligned}
    \Phi'(r_0) = \infty, \quad \Phi(r_0) = 0,
\end{aligned}
\end{equation}
where the former says that $r_0$ is the radius of the plane tangent to the tip of the HRT surface (i.e. where
$\frac{ \dd r }{ \dd \phi }=0$), while the latter implements that $\phi=0$
corresponds to the center of the strip.
It turns out that equation \eqref{eq:planarAeq} can be integrated, and the solution with the
correct boundary condition is 
\begin{equation}
\begin{aligned}
    \Phi(r) = \int_{r_0}^{r}\dd \rho \frac{ \sqrt{B(\rho)} }{ \rho
    \sqrt{(\rho/r_0)^{2d-2} - 1} }.\label{eq:phisol}
\end{aligned}
\end{equation}
This gives the location of the HRT surface within $\Sigma$ explicitly in terms
of the geometry of $\Sigma$. We now use this solution to determine the
Einstein constraint equations on $\Sigma$, and to derive a formula for the
rate of change of the entanglement growth.

\subsection{A momentum-entanglement correspondence}
Since $X_t$ is extremal, its first order variation reduces to a pure boundary
term given by (see for example the appendix of \cite{FisWis16, BaoCao19}):
\begin{equation}\label{eq:Neta}
\begin{aligned}
    \frac{ \dd \text{Area}[X_t] }{ \dd t }\Big|_{t=0} = \int_{\partial X} N^a\eta_a,
\end{aligned}
\end{equation}
where $\eta^a$ is the translation vector generating
the flow of $\partial X_t$ at conformal infinity at $t=0$, while $N^a$ is the normal
to $\partial X$ that is also tangent to $X$, and that points towards the
conformal boundary. In writing this formula, we implicitly assume that it is
evaluated with some near-boundary cutoff that is
subsequently removed. As is well known, given some choice of boundary conformal frame, a
canonical choice of cutoff exists \cite{FefGra85,GraLee91,GraWit99}, which in our case reduces to a cutoff
in the radial coordinate $r$. With a cutoff adapted to the Minkowski conformal
frame and the falloffs \eqref{eq:falloffs}, \eqref{eq:Neta} is finite, even
though $\text{Area}[X_t]$ diverges.

Now we write \eqref{eq:Neta} in a more useful form. We will give all the main
steps, but relegate tedious but straight forward computations to the appendix.

Using the planar symmetry of $\Sigma$, the extrinsic curvature
$K_{\alpha\beta}$ of $\Sigma$ is given by
\begin{equation}
\begin{aligned}
    K_{\mu\nu}\dd y^{\mu}\dd y^{\nu} = K_{rr}(r)\dd r^2 +
    K_{\phi\phi}(r)\left[\dd \phi^2 + \dd\bm{x}^2\right],
\end{aligned}
\end{equation}
where we take the extrinsic curvature to be defined with respect to the future
directed normal. Using this, we show in appendix~\eqref{app:dSKphi}, retracing
the steps of \cite{EngFol21}, that
\begin{equation}\label{eq:dAdt}
\begin{aligned}
    \frac{ \dd \text{Area}[X_t] }{ \dd t }\Big|_{t=0} = -\frac{ \text{Area}[\partial R] }{
    L^{d-2}}\lim_{r \rightarrow
    \infty}r^{d-3}K_{\phi\phi}.
\end{aligned}
\end{equation}
Physically, $\lim_{r \rightarrow \infty}r^{d-3}K_{\phi\phi}$ measures the
boost angle
at which $X$ hits the conformal boundary, or rather, the subleading part of the
angle, since extremality implies that $X$ hits $\partial \mathcal{M}$
orthogonally. This can be seen by studying extremal surfaces in a near-boundary
expansion. 
Thus, we see that the entanglement growth is, up to a factor,
identically given by the (subleading) boost angle at which the HRT surface hits the
boundary. The same was found for maximal volume slices in \cite{EngFol21}.

Next we want to find a more explicit expression for $\lim_{r \rightarrow \infty
}r^{d-3}K_{\phi\phi}$. To do this, we need to use the Einstein constraint
equations, which read
\begin{equation}
\begin{aligned}
    \mathcal{R} + K^2 -K^{\alpha \beta}K_{\alpha \beta} + \frac{ d(d-1) }{ L^2
    }&= 16 \pi G_N \mathcal{T}_{ab}t^a t^b, \\
    D_{\alpha}K\indices{^{\alpha}_\beta} - D_{\beta}K &= 8\pi G_N
    \mathcal{T}_{ab}t^a e^{b}_{\beta},
\end{aligned}
\end{equation}
where $\mathcal{R}$ is the Ricci scalar of the metric on
$\Sigma$, $t^a$ the future unit normal to $\Sigma$,
$K=H^{\alpha\beta}K_{\alpha\beta}$, and $e^{a}_{\alpha}$ a set of tangent
vectors to $\Sigma$. To write these equations in coordinate form, it is convenient to
introduce the function $\omega(r)$ as
\begin{equation}
\begin{aligned}
    B(r)  = \frac{ 1 }{ \frac{ r^2 }{ L^2  } - \frac{\omega(r)}{r^{d-2}} }.
\end{aligned}
\end{equation}
We will call $\omega(r)$ the Riemannian Hawking mass.\footnote{For $d=3$ it is
also known as the Geroch-Hawking mass \cite{Ger73,JanWal77,Wan01,HuiIlm01}, and
its was used to prove the Riemannian Penrose inequality \cite{HuiIlm01}.} 
It will play a central role in our work. Whether or not $\omega(\infty)$ is
proportional to the spacetime mass for some general spacelike hypersurface
$\Sigma$ depends on the behavior of the extrinsic
curvature $\Sigma$ at large $r$. It turns out that for $d\geq 3$, and with
$\Sigma$ being the extended homology hypersurface of an HRT surface, 
it has the property that it is proportional to the CFT
energy density: 
\begin{equation}\label{eq:TCFT}
\begin{aligned}
\left<T_{tt}\right> = \frac{ d-1 }{ 16 \pi G_N L^{d-1}  }\omega(\infty), \qquad
d\geq 3.
\end{aligned}
\end{equation}
For $d=2$, the right hand side is a lower bound on
$\left<T_{tt}\right>-\left<T_{tt}\right>_{\rm vac}$, where the vacuum energy
must be subtracted when we allow $\phi$ to be periodic. 
We will explain these facts in Sec.~\ref{eq:d2proof}.

It is also convenient to redefine $K_{rr}(r)$ in terms of a function $F(r)$
which is the $rr$-component of the extrinsic curvature in an orthonormal basis
\begin{equation}
    \begin{aligned}
        K_{rr}(r) \equiv B(r)F(r).
    \end{aligned}
\end{equation}
In terms of these functions, the constraint equations in coordinate form read
\begin{align}
    (d-1) \frac{\omega'(r)}{r^{d-1}} &= 2 \mathcal{E}(r) - \frac{ (d^2 - 3d + 2) }{ r^4 }K_{\phi\phi}(r)^2 - \frac{
        2(d-1) }{ r^2 }F(r) K_{\phi\phi}(r), \label{eq:preConstraint1}\\
    K_{\phi\phi}'(r) - \frac{ K_{\phi\phi}(r) }{ r } &= r F(r) - \frac{ r^2 }{
        d-1}\mathcal{J}(r), \label{eq:preConstraint2}
\end{align}
where we introduced the notation
\begin{equation}
\begin{aligned}
    \mathcal{E} &=  8 \pi G_N \mathcal{T}_{ab}t^a t^b, \label{eq:EJdef}\\
    \mathcal{J} &= 8\pi G_N \mathcal{T}_{ab} (\partial_r)^a t^{b}.
\end{aligned}
\end{equation}
These are (proportional to) the energy density and radial momentum density of the matter
with respect to the frame $t^a$. $\mathcal{J}>0$ corresponds to matter falling
into the bulk towards smaller $r$. From \eqref{eq:falloffs} and the fact that
$B(r)\sim \mathcal{O}(r^{-1})$, we find that
\begin{equation}
\begin{aligned}
    \mathcal{E}\sim o\left(1/r^{d}\right), \qquad \mathcal{J}\sim
    o\left(1/r^{d+1}\right),
\end{aligned}
\end{equation}
where we use that $\frac{ 1 }{ \sqrt{B} }(\partial_r)^a$ is a unit vector.

To turn \eqref{eq:preConstraint1} and \eqref{eq:preConstraint2} into a
closed system, we will eliminate $F(r)$. We do this by imposing extremality of
$X$ in the direction of $t^a$. To do this, note that the inwards
(outwards) null
expansion $\theta_+$ ($\theta_-$) of $X$ can be written as (see for example
the appendix of \cite{EngFol21})
\begin{equation}\label{eq:thetaonslice}
\begin{aligned}
    \sqrt{2}\theta_{\pm}[X] =\pm \mathcal{K}[X]+ K -
    n^{\alpha}n^{\beta}K_{\alpha\beta},
\end{aligned}
\end{equation}
where $n^{\alpha}$ is the outwards normal to $X$ within $\Sigma$,\footnote{
We have here taken the outwards and inwards null vectors, $k_+^a$ and $k_-^a$,
respectively, to be $k_{\pm}^a = 2^{-1/2}(t^a \pm n^a)$.
}
and where we remind that $\mathcal{K}[X]$ is the mean curvature of $X$ within
$\Sigma$. Extremality means $\theta_{+}=\theta_{-}=0$, which implies that $\mathcal{K}=0$ and
\begin{equation}
\begin{aligned}
    K|_{X} = n^{\alpha}n^{\beta}K_{\alpha\beta}|_{X}.
\end{aligned}
\end{equation}
This equation holds at $\phi=\Phi(r)$, which by planar symmetry means it holds
everywhere on $\Sigma$. Writing out this equation in coordinates, carried out in
appendix~\ref{app:tExtremality}, we find
\begin{equation}\label{eq:extremalitycond}
\begin{aligned}
    F+\frac{({d}-2) K_{\phi\phi}}{r^2}+({d}-1) K_{\phi\phi} \frac{(\Phi
    ')^2}{B}=0.
\end{aligned}
\end{equation}
Plugging in the solution for $\Phi(r)$, given in \eqref{eq:phisol}, we get that
\begin{equation}
\begin{aligned}
    F(r) = -\frac{  K_{\phi\phi}(r) }{ r^2 }\left(\frac{
    (d-2)r^{2d-2} + r_0^{2d-2}}{ r^{2d-2} -  r_0^{2d-2} }
    \right),
\end{aligned}
\end{equation}
which upon insertion into the constraints, gives a closed system of ODEs
\begin{align}
    \frac{\omega'(r)}{r^{d-1}} &= \frac{ 2 }{ d-1 } \mathcal{E}(r) + 
    \frac{K_{\phi\phi}^2}{r^{4}}h_1(r)\label{eq:omegaeq},
    \\
    K_{\phi\phi}'(r) & + \frac{ K_{\phi\phi} }{ r }h_2(r) = - \frac{ r^2 }{ d-1
    }\mathcal{J}(r) \label{eq:Keq},
\end{align}
where
\begin{equation}
\begin{aligned}
    h_1(r) = \frac{ (d-2)\left({ r }/{ r_0 }\right)^{2d-2} + d }{ \left({r}/{r_0}\right)^{2d- 2} -
    1 }, \qquad
    h_2(r) = \frac{ (d-3)\left({ r }/{ r_0 } \right)^{2d-2}  + 2 }{ \left({ r }/{ r_0 } \right)^{2d-2} 
    - 1}.
\end{aligned}
\end{equation}
Now, $F(r)$ is
a component of the extrinsic curvature in an orthonormal basis, so it must be
finite at $r_0$. Using this to fix an integration constant,
we find that the solutions of \eqref{eq:omegaeq} and \eqref{eq:Keq} are
\begin{align}
    K_{\phi \phi }(r) &= -\frac{r^2}{(d-1)\sqrt{(r/r_0)^{2
    d-2}-1}}\int_{r_0}^r \dd\rho \mathcal{J}(\rho)\sqrt{(\rho/r_0)^{2 d-2}-1},
    \label{eq:Ksol}\\
    \omega(r) &= \omega(r_0) + \int_{r_0}^{r}\dd \rho
    \left[ 
    \rho^{d-5} h_1(\rho) K_{\phi\phi}(\rho)^2
    + \frac{ 2 \rho^{d-1} }{ d - 1 }\mathcal{E}(\rho) \right], \label{eq:omegasol}
\end{align}
where $K_{\phi\phi}(r_0) =0$ since $\mathcal{J}(\rho)$ must be
bounded.\footnote{For thin-shell spacetimes $\mathcal{J}(\rho)$ can be a delta
function, but we can safely assume this delta function does not have support
exactly at $r=r_0$.}
Inserting \eqref{eq:Ksol} into \eqref{eq:dAdt} and multiplying
by $(4G_N)^{-1}$, we get that
\begin{equation}
    \begin{aligned}\label{eq:dSonHRT}
    \frac{ \dd S_R}{ \dd t }\Big|_{t=0} = 
     \frac{ \text{Area}[\partial R] }{ 4G_N L^{d-2} (d-1) }
        \int_{r_0}^{\infty}\dd r \mathcal{J}(r) \sqrt{r^{2d-2} - r_0^{2d-2}}.
\end{aligned}
\end{equation}
Since $\mathcal{J}>0$ corresponds to a flux of energy density towards decreasing
$r$, we see that matter falling out of the entanglement wedge and deeper into the bulk is directly responsible for the
increase of entanglement. Conversely, outgoing matter is responsible for decrease
in entanglement. We can also rewrite this formula in a covariant way. In
appendix \ref{app:dSKphi} we show that
\begin{equation}\label{eq:dSonHRT2}
\begin{aligned}
    \frac{ \dd S_R}{ \dd t }\Big|_{t=0} = \int_{X} G \mathcal{T}_{ab}n^a t^a
\end{aligned}
\end{equation}
where $n^a$ is the outwards unit normal to $X$ that is tangent to $\Sigma$, and
\begin{equation}\label{eq:Gdef}
\begin{aligned}
    G(r) = \frac{ 2\pi r^{d} }{(d-1) r_0^{d-1}  }.
\end{aligned}
\end{equation}

The formulas \eqref{eq:dAdt}, \eqref{eq:Ksol}--\eqref{eq:Gdef} are the main results of this section. These results together
with the theorems proven in the following section are crucial pieces to our
proven bounds. 
\subsection{Geometric constraints on the HRT surface}\label{sec:HRTconstr}
In this section, we prove the following
\begin{thm}\label{thm:r0bound}
    Let $(\mathcal{M},g_{ab})$ be a regular asymptotically AdS$_{d+1\geq 3}$
    spacetime with planar symmetry satisfying the DEC.  Let $X$ be the HRT surface
    of a strip $R$ of width $\ell$, and let be $r_0$ be the smallest radius
    probed by $X$.  Then
        \begin{equation}\label{eq:r0bnd}
        \frac{ L^2 }{ r_0 } \leq \frac{
            \Gamma\left(\frac{ 1 }{ 2(d-1) }
            \right)}{2\sqrt{\pi}\Gamma\left(\frac{ d }{ 2(d-1) } \right) } \ell.
        \end{equation}
        Furthermore, if $r_{0, \rm vac}$ is the smallest radius probed by the
        HRT surface $X_0$ of a strip of width $\ell$ in pure AdS$_{d+1}$, then
        \begin{equation}\label{eq:r0proof}
        \begin{aligned}
            r_0 \geq r_{0, {\rm vac}}.
        \end{aligned}
        \end{equation}
\end{thm}
\noindent We now give the proof assuming that $\omega(r_0)\geq 0$, and then we will
spend most of the rest of this section proving this assertion. 
\begin{proof}
    By Lemma~\ref{lem:omega0}, proven below, we have that $\omega(r_0) \geq 0$. Furthermore, the DEC implies that
    $\mathcal{E}$ is positive. Hence, \eqref{eq:omegasol} gives that $\omega(r)$ is everywhere positive. 
    But this means that
    \begin{equation}
    \begin{aligned}
        B(r) = \frac{ 1 }{ \frac{ r^2 }{ L^2 } - \frac{ \omega(r) }{ r^{d-2} } }
        \geq \frac{ L^2 }{ r^2 },
    \end{aligned}
    \end{equation}
    which allows us to lower bound the strip width as follows:
    \begin{equation}
    \begin{aligned}
        \ell &= 2L \Phi(\infty) 
        = 2L \int_{r_0}^{\infty}\dd \rho \frac{ \sqrt{B(\rho)} }{ \rho
    \sqrt{(\rho/r_0)^{2d-2} - 1} } \\
        &\geq 2L^2 \int_{r_0}^{\infty}\dd \rho \frac{ 1 }{ \rho^2
    \sqrt{(\rho/r_0)^{2d-2} - 1} } \\
        &= \frac{ 2L^2 }{ r_0 } \frac{ \sqrt{\pi}\Gamma\left(\frac{ d }{ 2(d-1)
        } \right) }{ \Gamma\left(\frac{ 1 }{ 2(d-1) } \right) }.
    \end{aligned}
    \end{equation}
    Finally, if we are in pure AdS, we must have have that the spacetime mass is
    vanishing, implying that $\omega(\infty)=0$, and so by
    $\omega'(r)\geq 0$ and the fact that $\omega(r_0)\geq 0$, we must have $\omega(r)=0$ everywhere. But that means that the above
    inequalities become equalities, giving $\frac{ L^2 }{ r_0 }\geq \frac{ L^2
    }{ r_{0, \text{vac}} }$, which implies \eqref{eq:r0proof}.
\end{proof}

\noindent Now we turn to proving that $\omega(r_0)$ is non-negative. The
crucial tool is a planar-symmetric AdS$_{d+1}$ version of the
Lorentzian Hawking mass \cite{Haw68}, which we define for a planar surface $\sigma$ as 
\begin{equation}
    \begin{aligned}\label{eq:mudef}
    \mu[\sigma] = \frac{ r^{d} }{ L^2 } - \frac{ 2r^{d}\theta_+ \theta_{-} }{ k_+ \cdot
    k_- (d-1)^2},
\end{aligned}
\end{equation}
where $k^{+}$ and $k^{-}$ are the outwards and inwards null vectors orthogonal
to $\sigma$, respectively, and $\theta_{\pm}$ the corresponding null
expansions. In \cite{Fol22}, generalizing the
results of \cite{Hay94} to planar symmetry and AAdS$_{d+1}$ spacetimes, it was shown that the DEC implies that $\mu[\sigma]$ is
monotonically non-decreasing when $\sigma$ is moving in an outwards spacelike
direction, provided we are in a normal region of spacetime, meaning that $\theta_+ \geq
0, \theta_{-}\leq0$ when we take $k_{+}^a$ and $k_-^a$ to be future directed.\footnote{This monotonicity is a planar-symmetric Lorentzian version of the
monotonicity the Riemannian Hawking mass under inverse mean curvature flow,
which has been used to prove Riemannian Penrose inequalities
\cite{Ger73,JanWal77,Wan01,HuiIlm01}. A Lorentzian flow with a monotonic
Lorentzian Hawking mass for compact
surfaces in three dimensions, without any symmetry assumptions, was studied in
\cite{BraHay06}.}

Furthermore, it is useful to rewrite the Riemannian Hawking mass $\omega(r)$ in a different way.
$\omega$ can be thought of as a function of a planar surface $\sigma$ together
with a hypersurface $\Sigma$ containing it, and in \cite{EngFol21} it is shown that we
can write $\omega$ as
\begin{equation}
\begin{aligned}
    \omega[\sigma, \Sigma] = \frac{ r^{d} }{ L^2 } - \frac{ \mathcal{K}[\sigma]^2 }{
        (d-1)^2 },
\end{aligned}
\end{equation}
where $\mathcal{K}$ is the mean curvature of  $\sigma$ in $\Sigma$.
Using \eqref{eq:thetaonslice}, which assumes the normalization $k_{+}\cdot k_{-} = -1$, we see that $2\theta_{+}\theta_{-} =
(K-n^{\alpha}n^{\beta}K_{\alpha\beta})^2 - \mathcal{K}^2$, and so we get the
following relation between the Hawking masses
\begin{equation}\label{eq:muomegacomp}
\begin{aligned}
    \mu[\sigma] = \omega[\sigma, \Sigma] + \frac{ r^{d} }{ (d-1)^2 }(K -
    n^{\alpha}n^{\beta}K_{\alpha\beta})^2.
\end{aligned}
\end{equation}
\noindent With this in hand, we prove the following Lemma.
\begin{figure}
    \centering
    \includegraphics[width=.4\textwidth]{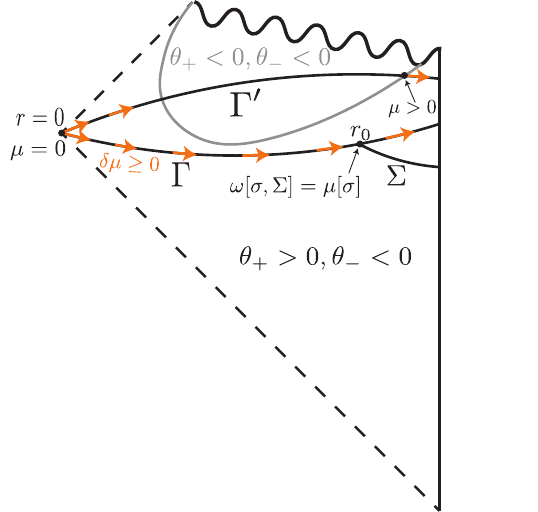}
    \caption{Example of two complete hypersurfaces $\Gamma$ and $\Gamma'$. The
    Lorentzian Hawking mass is vanishing at $r=0$ and positive at marginally trapped
    surfaces, given by the planes contained in the gray line. $\mu$ is
    monotonically non-decreasing along spacelike outwards flows in the untrapped
    region, where $\theta_{+}\geq 0, \theta_{-}\leq0$.
    At the boundary $\sigma$ of the extended homology hypersurface $\Sigma$, 
    the Riemannian Hawking mass $\omega$ with respect to $\Sigma$ agrees with the Lorentzian Hawking mass
    $\mu$.
    }
    \label{fig:muflow}
\end{figure}
\begin{lem}\label{lem:mupos}
    Let $\Gamma$ be a complete planar symmetric hypersurface with one conformal
    boundary. Let $\sigma_r$ be a one-parameter family of planes in $\Gamma$ with
    radius $r$, and with $r\in(0, \epsilon]$ for any $\epsilon>0$. Then
    \begin{equation}
    \begin{aligned}
        \lim_{r \rightarrow 0} \mu[\sigma_r] \geq 0.
    \end{aligned}
    \end{equation}
\end{lem}
\begin{proof}
Let us pick coordinates
\begin{equation}
    \dd s^2|_{\Gamma} = \left[\frac{ r^2 }{ L^2 } - \frac{ \omega(r) }{ r^{d-2}
    }\right]^{-1}\dd r^2 + r^2 \dd \bm{x}^2
\end{equation}
    on $\Gamma$ in a neighborhood of $r=0$. Since $\Gamma$ is complete and we only have one conformal
    boundary, arbitrarily small $r$ must be part of $\Gamma$. Since $\Gamma$ is
    spacelike, we must have $\omega(r) \leq r^{d}/L^2$, which means that
    $\omega(r)\sim \mathcal{O}(r^{d})$ at small $r$. Now, from
    \eqref{eq:muomegacomp} we see that $\mu[\sigma_r] \geq \omega[\sigma_r,
    \Gamma]$ and so
    \begin{equation}
    \begin{aligned}
        \mu[\sigma_r] \geq \mathcal{O}(r^{d}).
    \end{aligned}
    \end{equation}
    Taking $r\rightarrow 0$ proves our assertion.
\end{proof}
\noindent Now we are ready to prove that $\omega(r_0)\geq 0$, together with the
fact that the tip of the HRT surface cannot lie in a trapped region.
\begin{lem}\label{lem:omega0}
    Let $(\mathcal{M}, g_{ab})$ be a planar-symmetric regular asymptotically AdS$_{d+1}$
    spacetime. Let $X$ be the HRT surface of a strip. 
    Then the tip of $X$ lies in an untrapped region of spacetime, meaning the
   future null expansions of the plane $\sigma$ tangent to $X$ at the tip satisfies
   \begin{equation}
   \begin{aligned}
       \theta_+[\sigma] \geq 0, \quad \theta_{-}[\sigma] \leq 0.
   \end{aligned}
   \end{equation}
    Furthermore, if the DEC holds and $(\mathcal{M}, g)$ is regular, 
    the Riemannian Hawking mass of $\sigma$ is non-negative:
   \begin{equation}
   \begin{aligned}
       \omega[\sigma, \Sigma] = \omega(r_0) \geq 0.
   \end{aligned}
   \end{equation}
\end{lem}
\begin{proof}
    Let $\Sigma$ be the unique planar symmetric extended homology hypersurface
    containing $X$. Let $\sigma$ be the boundary of $\Sigma$ in the bulk, having radius $r_0$.
    Its null expansion is
\begin{equation}
\begin{aligned}
    \sqrt{2} \theta_{\pm}[\sigma] = \pm \mathcal{K}[\sigma] + K - r^{\alpha}r^{\beta}K_{\alpha\beta},
\end{aligned}
\end{equation}
where $r^{\alpha} = \frac{ 1 }{ \sqrt{B} }(\partial_r)^{\alpha}$.
 An explicit computation gives
\begin{equation}
    \begin{aligned}\label{eq:Ktr}
    \mathcal{K}[\sigma] &=  \frac{ d-1 }{ r_0 \sqrt{B(r_0)} }, \\
    K &=  \frac{ 1 }{ B }K_{rr} + \frac{ K_{\phi\phi}(d-1) }{ r^2 },
\end{aligned}
\end{equation}
and so we find
\begin{equation}
\begin{aligned}
    \sqrt{2} \theta_{\pm}[\sigma] = \pm \frac{ d-1 }{ r_0\sqrt{B(r_0)} } -
    \frac{ K_{\phi\phi}(r_0)(d-1) }{ r_0^2 }.\label{eq:thetapmform}
\end{aligned}
\end{equation}
From \eqref{eq:Ksol} we have that $K_{\phi\phi}(r_0)=0$, and so we get that
\begin{equation}
\begin{aligned}
\pm \theta_{\pm} \geq 0,
\end{aligned}
\end{equation}
proving the first assertion. 

Next, since $K_{\phi\phi}(r_0)=0$ we see that $2\theta_{+}\theta_{-}|_{\sigma} =
    -\mathcal{K}^2|_{\sigma}$, implying that $\mu[\sigma] =\omega[\sigma,
    \Sigma]$. Now, since our spacetime is AdS-hyperbolic, we can embed $\sigma$ in a
complete hypersurface with planar symmetry $\Gamma$, see Figure \ref{fig:muflow}.
    Since $\sigma$ lies in an untrapped region of
    spacetime, and since $\Gamma$ is spacelike, $\mu[\sigma]$ is monotonically
    non-increasing as we deform $\sigma$ in inwards along $\Gamma$ while
    preserving its planar symmetry. Since the $g_{ab}$ is $C^2$,
    $\theta_{\pm}$ are continuous, and so as we deform $\sigma$ inwards, one of two
    things happen. Either we hit a marginally trapped surface, where
    $\theta_{+}\theta_{-}=0$ and where $\mu$ is manifestly positive, or we approach
    $r=0$, where we again have that $\mu$ is non-negative by
    Lemma~\ref{lem:mupos}. See Figure \ref{fig:muflow}. But since $\mu$ is non-increasing along
    this deformation, and since it ends up somewhere non-negative, we must have
    $\mu[\sigma] \geq 0$. But $\mu[\sigma]=\omega[\sigma, \Sigma]$, completing
    the proof.
\end{proof}

We have illustrated the fact that the tip cannot lie in a trapped region of
spacetime in Figure \ref{fig:planarHRTs} -- the tip cannot lie behind the gray
line. Note that the proof of this fact does not rely on the DEC. 
This result  improves on the findings of \cite{EngFis15} in the special case where we have planar
symmetry. In \cite{EngFis15}, they showed without any symmetry assumptions that the tip of an HRT surface in a
$(2+1)$--dimensional spacetime can never lie in the so-called umbral region,
which is a special subset of the trapped region that lies behind regular
holographic screens \cite{Bou99,EngFis15}. They also showed this result with
planar symmetry in all dimensions.  Here we extend this result to
show that the whole trapped region is forbidden, although our
result is more limited in that it always requires planar symmetry and a strip
(or spherical) boundary region. Note also that this
result does not forbid $X$ to probe inside trapped regions -- it is only the tip
that is forbidden to lie there (see Figure \ref{fig:planarHRTs}). For example, for early times after a quench, the
HRT surface will have portions threading through the trapped region \cite{LiuSuh13a,LiuSuh13b}.

\subsection{Proofs}\label{eq:d2proof}\label{sec:thinshellstrip}
\subsubsection*{Proof of $d=2$ bound}
We are now ready to prove Theorem~\ref{thm:mainthm1}.
Evaluating the Lorentzian Hawking mass on a sphere at large $r$ in a planar
symmetric AAdS$_{d+1}$ spacetime with falloffs \eqref{eq:falloffs}, we get that
\begin{equation}\label{eq:TCFT2}
\begin{aligned}
\left<T_{tt}\right> = \frac{ d-1 }{ 16 \pi G_N L^{d-1}  }\mu(\infty).
\end{aligned}
\end{equation}
This is valid also for $d=2$, except if $\phi$ is periodically identified, we
must replace the left hand side with $\left<T_{tt}\right> -
\left<T_{tt}\right>_{\rm vac}$. It can be seen to be true by evaluating $\mu(\infty)$ near
the boundary in the usual Fefferman-Graham expansion \cite{FefGra85,GraLee91,GraWit99}. Now, from
 \eqref{eq:muomegacomp} and \eqref{eq:Ktr} we have that
\begin{equation}\label{eq:muVSomega}
\begin{aligned}
    \mu(r) = \omega(r) + r^{d-4}K_{\phi\phi}(r)^2.
\end{aligned}
\end{equation}
From \eqref{eq:Ksol}, we see that $K_{\phi\phi}$ has asymptotic falloff
$K_{\phi\phi}\sim \mathcal{O}(r^{3-d})$. 
Thus, we get that for
$d\geq 3$, $\mu(\infty)=\omega(\infty)$, while for $d=2$, we have
\begin{equation}
\begin{aligned}
    \mu(\infty) = \omega(\infty) + \left(\lim_{r
    \rightarrow \infty}r^{-1}K_{\phi\phi}\right)^2.
\end{aligned}
\end{equation}
Since $\omega(\infty) \geq 0$ by the DEC, when $d=2$ we obtain
\begin{equation}\label{eq:d2masterresult}
\begin{aligned}
    \Big|\lim_{r \rightarrow \infty}r^{-1}K_{\phi\phi}\Big| \leq \sqrt{\mu(\infty)} 
    = \sqrt{16\pi G_N L \left<T_{tt}\right> }.
\end{aligned}
\end{equation}
Using that $\text{Area}[\partial R_t]=2$, and combining
 \eqref{eq:d2masterresult} and \eqref{eq:dAdt} then yields
\begin{equation}
\begin{aligned}
    \Big| \frac{ \dd S_R }{ \dd t  } \Big|_{t=0} 
    &\leq \frac{ 1 }{ 2G_N }\sqrt{16\pi G_N L \left<T_{tt}\right>}
    =\sqrt{\frac{8\pi c}{3}\left<T_{tt}\right>},
\end{aligned}
\end{equation}
where we used the known Brown-Henneaux expression for the central charge:
$c = \frac{ 3L }{ 2G_N }$ \cite{BroHen86}. This proves Theorem~\ref{thm:mainthm1}.

\subsubsection*{Proof of bound for small $\ell$}
Now let us consider the result for small subregions, given by
Theorem~\ref{thm:mainthm3}.
The following Lemma is what we need:
\begin{lem}\label{lem:smallLmainlem}
    Let $(\mathcal{M},g_{ab})$ be a regular asymptotically AdS$_{d+1\geq 3}$
    spacetime with planar symmetry satisfying the DEC.  Let $X$ be the HRT surface
    of a strip $R$ of width $\ell$, and let be $r_0$ be the smallest radius
    probed by $X$. Assume that
    \begin{equation}
    \begin{aligned}
    \frac{ \ell^{d}\left<T_{tt}\right> }{ c_{\rm eff} } \ll 1.
    \end{aligned}
    \end{equation}
    Then
    \begin{equation}\label{eq:mainlem1}
    \begin{aligned}
        \Big| \lim_{r \rightarrow \infty} r^{d-3}K_{\phi\phi} \Big| \leq \frac{ L }{ 2r_0 }
        \omega(\infty)\left[1 + \mathcal{O}\left(\frac{
            \ell^{d}\left<T_{tt}\right> }{ c_{\rm eff} }\right)\right].
    \end{aligned}
    \end{equation}
\end{lem}
\begin{proof}
    Let us for convenience define $W=-\lim_{r \rightarrow
    \infty}r^{d-3}K_{\phi\phi}$, and
    assume without loss of generality that $W>0$ (otherwise, just reverse the
    time direction). Using the solutions \eqref{eq:Ksol} and \eqref{eq:omegasol}, we have that
    \begin{equation}
    \begin{aligned}
        \frac{ W }{ \omega(\infty) } &= \frac{ 1 }{ d-1 } \frac{
            \int_{r_0}^{\infty}\dd r \mathcal{J}(r ) \sqrt{r^{2d-2} - r_0^{2d-2}} }{
            \omega(r_0) + \int_{r_0}^{\infty}\dd
            r\left[r^{d-5}h_1(r)K_{\phi\phi}(r)^2 + \frac{ 2r^{d-1} }{ d-1
            }\mathcal{E}(r)\right]  } \\
         &\leq \frac{
            \int_{r_0}^{\infty}\dd r r^{d-1} \mathcal{J}(r )}{
            2 \int_{r_0}^{\infty} \dd r r ^{d-1}\mathcal{E}(\rho)  }.
    \end{aligned}
    \end{equation}
    The DEC requires that
    \begin{equation}
    \begin{aligned}
        0 \leq 8\pi G_N \mathcal{T}_{ab}\left[t^a \pm \frac{ 1 }{ \sqrt{B}
        }(\partial_r)^{a}\right]t^b = \mathcal{E} \pm \frac{ 1 }{ \sqrt{B}
        }\mathcal{J},
    \end{aligned}
    \end{equation}
    and so we have that
    \begin{equation}\label{eq:Ecallower}
    \begin{aligned}
    \mathcal{E} \geq \frac{ 1 }{ \sqrt{B} }|\mathcal{J}|.
    \end{aligned}
    \end{equation}
    Writing $B$ in terms
    of $\omega$, and enforcing the DEC, we get
    \begin{equation}
    \begin{aligned}
        \frac{ W }{ \omega(\infty) }
         &\leq \frac{
            \int_{r_0}^{\infty}\dd r r^{d-1} \mathcal{J}(r )}{
            \frac{ 2 }{ L }\int_{r_0}^{\infty} \dd r r^{d}\sqrt{1 - \frac{
                \omega(r) L^2 }{ r^{d} }} |\mathcal{J}(r)|  }.
    \end{aligned}
    \end{equation}
    Let us now for a moment assume that we are perturbatively close to the vacuum,
    where $\mathcal{\epsilon}$ is a perturbative parameter parametrizing the
    magnitude of $\omega(\infty)$. By monotonicity and positivity of $\omega(r)$, $\omega(r) \sim
    \mathcal{O}(\epsilon)$ as well, and so the $\omega(r)$ appearing in the
    square root gives higher order contributions:
    \begin{equation}
    \begin{aligned}
        \frac{ W }{ \omega(\infty) }
         &\leq \frac{
            \int_{r_0}^{\infty}\dd r r^{d-1} |\mathcal{J}(r )|}{
            \frac{ 2 }{ L }\left[\int_{r_0}^{\infty} \dd r r^{d} |\mathcal{J}(r)| -
            \frac{ L^2 }{ 2 }\int_{r_0}^{\infty} \dd r \omega(r) |\mathcal{J}(r)| +
            \ldots\right]  } \\\
        &=\frac{ L }{ 2 }\frac{
            \int_{r_0}^{\infty}\dd r r^{d-1} |\mathcal{J}(r )|}{
            \int_{r_0}^{\infty} \dd r r^{d} |\mathcal{J}(r)|}\left[1 + \frac{ L^2 }{ 2
        }\frac{ \int_{r_0}^{\infty}\dd r \omega(r) |\mathcal{J}| }{
            \int_{r_0}^{\infty}\dd r r^{d}|\mathcal{J}|} + \ldots \right] \\
        &\leq \frac{ L }{ 2r_0 }\left[1 + \frac{ L^2 }{ r_0^{d} }\omega(\infty)
        + \ldots \right] \\
        &\leq \frac{ L }{ 2r_0 }\left[1 + \frac{ L^{2d} }{ r_0^{d}
        }\frac{\mu(\infty)}{L^{2d-2}}
        + \ldots \right] \\
        &=\frac{ L }{ 2r_0 }\left[1 + 
        \frac{16\pi \eta_d}{d-1}
        \frac{\ell^{d}\left<T_{tt}\right>}{c_{\rm eff}}
        + \ldots \right]
    \end{aligned}
    \end{equation}
    where $\eta_d$ is the $O(1)$ number coming from using \eqref{eq:r0bnd}. We see that the effective expansion 
    parameter is the dimensionless quantity
    $\frac{\ell^{d}\left<T_{tt}\right>}{c_{\rm eff}}$. So the expansion is not
    really in small mass, which is dimensionful, but in small strip width
    relative to the inverse energy density per CFT degree of freedom.  
\end{proof}
\noindent From \eqref{eq:dAdt} and \eqref{eq:r0bnd}, we get, up to the
perturbative corrections,
\begin{equation}
    \begin{aligned}\label{eq:boundsteps}
    \Big| \frac{ \dd S_R }{ \dd t  }\Big| &\leq \frac{ \text{Area}[\partial
    R]}{ 4G_N L^{d-2} } \frac{ L }{ 2r_0 }\omega(\infty) 
    \leq \frac{ \sqrt{\pi} }{ d-1 }\frac{
            \Gamma\left(\frac{ 1 }{ 2(d-1) } \right) }{ 
        \Gamma\left(\frac{ d }{ 2(d-1) } \right)   } \ell \text{Area}[\partial
        R] \left<T_{tt}\right> \\
        &=\text{Vol}[R] \left< T_{tt}\right> \begin{cases}
        2\pi  & d=2 \\
        \frac{ \sqrt{\pi} }{ d-1 } \frac{\Gamma\left(\frac{ 1 }{ 2(d-1) }
            \right) }{ \Gamma\left(\frac{ d }{ 2(d-1) } \right) } & d>2 
        \end{cases},
\end{aligned}
\end{equation}
where we used \eqref{eq:r0bnd} and \eqref{eq:TCFT} in the second inequality. This
also holds for $d=2$, since $\omega(\infty) \leq \mu(\infty)$, provided we
replace $\left<T_{tt}\right> \rightarrow \left<T_{tt}\right> -
\left<T_{tt}\right>_{\rm vac}$ if $\phi$ is compact. 
Also, note that
for $d=2$ we have that $\ell \text{Area}[\partial R] =2 \text{Vol}[R]$.
This completes the proof of Theorem~\ref{thm:mainthm3} for strip regions.

\subsubsection*{Proof of bounds in thin-shell spacetimes}
We now turn our attention to thin-shell spacetimes, where we will be able to establish that
a bound of the form $|\partial_t S| \leq \# \text{Vol}[R]\left<T_{tt}\right>$ holds for any $\ell$.
Furthermore, in this class of spacetimes we will prove our conjectured
generalization of Theorem~\ref{thm:mainthm1} to $d>2$, i.e.
Theorem~\ref{thm:mainthm2}.

\begin{figure}
    \centering
    \includegraphics[width=.7\textwidth]{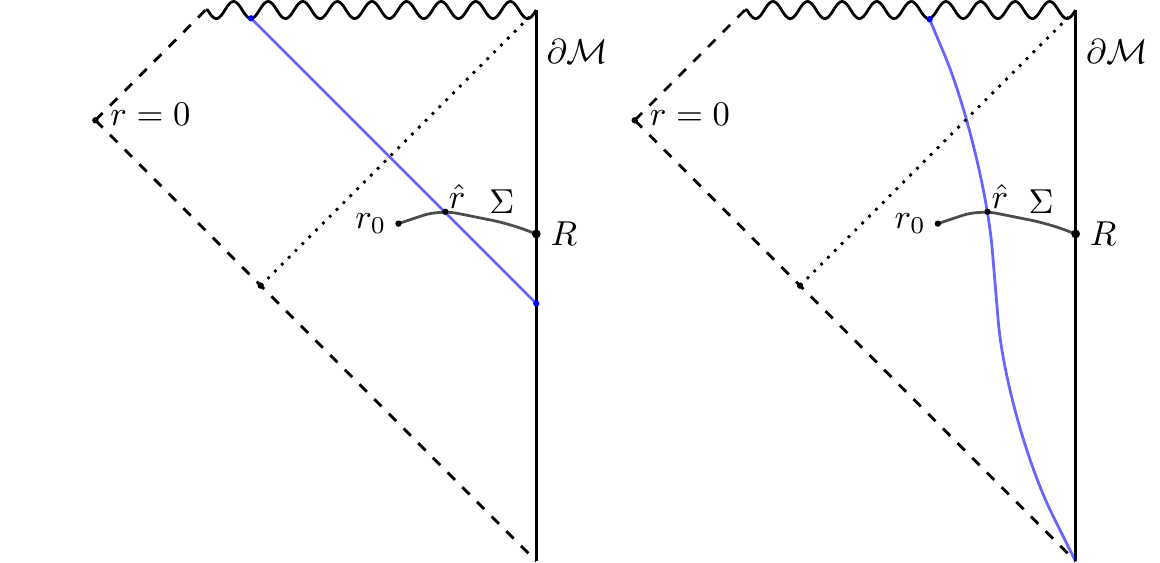}
    \caption{Examples of thin-shell spacetimes, where the blue lines correspond
    to the shells. The left space is dual to a uniform quench, where matter is
    thrown in from the boundary, while the right is a spacetime with a brane in
    the bulk interior. 
    }
    \label{fig:shockspace}
\end{figure}

Consider a spacetime where the matter consists of a single thin shell of
matter that separately satisfies the DEC, together with a possible contribution
from any number of $U(1)$ gauge fields:
\begin{equation}\label{eq:thinshells}
\begin{aligned}
    \mathcal{E} &= \kappa \delta(r-\hat{r}) + \mathcal{E}^{\rm Maxwell}, \\
    \mathcal{J} &= \eta \delta(r-\hat{r}),
\end{aligned}
\end{equation}
for some $\eta, \kappa, \hat{r} > r_0$. See Figure \ref{fig:shockspace}. 
Here we used that in planar symmetry, Maxwell
fields give no contribution to the radial momentum density $\mathcal{J}$ (see for example Sec.~3 of
\cite{EngFol21}). In fact, we can add to the $U(1)$ gauge fields any matter that has
a positive contribution to $\mathcal{E}$ but no contribution to $\mathcal{J}$.

The DEC, through \eqref{eq:Ecallower}, imposes that $\mathcal{J}$ only can have support at
$\hat{r}$. Without loss of generality, we take $\eta>0$.
Let us in this section also use our scaling freedom in $r$ to set $r_0 = L$ and choice of
units to set $L=1$. 

Define again $W = - \lim_{r \rightarrow \infty}r^{d-3}K_{\phi\phi}$.
Plugging \eqref{eq:thinshells} into \eqref{eq:Ksol}, the solution for
$K_{\phi\phi}$ is
\begin{equation}
\begin{aligned}
    K_{\phi \phi }(r) 
    &=  - \frac{ r^2 }{ d-1 }\eta \sqrt{ \frac{ \hat{r}^{2d-2} -
   1 
    }{ r^{2d-2} - 1}  }\theta(r-\hat{r}),
\end{aligned}
\end{equation}
and so
\begin{equation}
\begin{aligned}
    \eta = \frac{ (d-1) W }{ \sqrt{\hat{r}^{2d-2} - 1} },
\end{aligned}
\end{equation}
which gives
\begin{equation}
\begin{aligned}
    K_{\phi\phi}(r) = - \frac{ r^2 W }{ \sqrt{r^{2d-2} -
    1}}\theta(r-\hat{r}) .
\end{aligned}
\end{equation}
Next, let us solve for the contribution to $\omega(r)$ from the squared extrinsic curvature
term in \eqref{eq:omegasol}:
\begin{equation}
\begin{aligned}
    Q(r) &\equiv \int_{1}^{r}\dd \rho \rho^{d-5}K_{\phi\phi}(\rho)^2 h_1(\rho) \\
    &= \theta(r-\hat{r}) W^2 \int_{\hat{r}}^{r}\dd \rho \rho^{d-1} \frac{ h_1(\rho) }{
        \left[\rho^{2d-2} - 1\right] } \\
    &=W^2 \theta(r-\hat{r})\left[ \frac{ \hat{r}^{d} }{ \hat{r}^{2d-2} - 1 } - \frac{
        r^{d} }{ r^{2d-2} - 1 }\right].
\end{aligned}
\end{equation}
To proceed, we need to understand what happens to $\omega$ as we cross the shock.
Restricting attention to a small neighborhood of $\hat{r}$, where we can treat
explicit occurrences of $r$ not appearing in delta functions as constant, the equation for $\omega$ reads
\begin{equation}\label{eq:eqshell}
\begin{aligned}
    (d-1) \frac{ \omega'(r) }{ \hat{r}^{d-1} } = 2\mathcal{E}^{\rm shell} + \ldots,
\end{aligned}
\end{equation}
where the terms indicated with dots will make no contribution to the
discontinuity.
Remembering that the DEC implies that $\sqrt{B}\mathcal{E}\geq |\mathcal{J}|$,
imposing the DEC on the shell means that 
\begin{equation}\label{eq:eshellbnd}
\begin{aligned}
    \mathcal{E}^{\rm shell} \geq \sqrt{\hat{r}^2 - \frac{ \omega(r) }{ \hat{r}^{d-2} }}
    \eta \delta(r-\hat{r}).
\end{aligned}
\end{equation}
Inserting \eqref{eq:eshellbnd} into \eqref{eq:eqshell},
dividing by the prefactor of the delta function, and integrating from
$\hat{r}-\varepsilon$ to $\hat{r}+\varepsilon$ for some small positive
$\varepsilon$, we find
\begin{equation}\label{eq:disconteq}
\begin{aligned}
    \sqrt{\hat{r}^{2} - \frac{
        \omega_{-} }{
        \hat{r}^{d-2} }} - \sqrt{\hat{r}^{2} - \frac{
        \omega_{+} }{ \hat{r}^{d-2} 
    }} \geq \frac{ \hat{r} }{ d-1 }\eta + \mathcal{O}(\varepsilon),
\end{aligned}
\end{equation}
where we defined $\omega_{\pm} = \omega(\hat{r}\pm \varepsilon)$.
We only have a sensible solution when $B(r)$ is real and positive everywhere, which
requires
\begin{equation}
\begin{aligned}
    \frac{ 1 }{ d-1 }\eta \leq\sqrt{1 -
    \frac{ \omega_{-} }{ \hat{r}^{d} }}.
\end{aligned}
\end{equation}
Solving for $\omega_{-}$ from \eqref{eq:disconteq} and inserting our expression for $\eta$, we get that
\begin{equation}
\begin{aligned}
    \omega_{+} 
    &\geq \omega_{-} + \frac{
        \hat{r}^{d} }{ \sqrt{\hat{r}^{2d-2} - 1} }W\left[ \sqrt{1 - \frac{
        \omega_{-} }{ \hat{r}^{d} }} - \frac{ W }{
        \sqrt{\hat{r}^{2d-2} - 1} }
    \right].
\end{aligned}
\end{equation}
Using this and \eqref{eq:muVSomega}, the Lorentzian Hawking mass at infinity has the lower bound
\begin{equation}
\begin{aligned}
    \mu(\infty) &=\omega(\infty) + \delta_{d2} W^2 \\
    &\geq \omega_{-} 
    + \frac{
        \hat{r}^{d} }{ \sqrt{\hat{r}^{2d-2} - 1} }W\left[ \sqrt{1 - \frac{
        \omega_{-} }{ \hat{r}^{d} }} - \frac{ W }{
        \sqrt{\hat{r}^{2d-2} - 1} }
    \right]
    + Q(\infty) + \delta_{d2}W^2 \\
    &= \omega_{-} 
    + \frac{
        \hat{r}^{d} }{ \sqrt{\hat{r}^{2d-2} - 1} }W \sqrt{1 - \frac{
        \omega_{-} }{ \hat{r}^{d} }},
\end{aligned}
\end{equation}
where $\delta_{ij}$ is the Kronecker delta. Thus, for any real $n$, we have that
\begin{equation}
\begin{aligned}
    \frac{ W^{n} }{ \mu(\infty) } \leq \frac{ W^{n} }{  \omega_{-} 
    + \frac{
        \hat{r}^{d} }{ \sqrt{\hat{r}^{2d-2} - 1} }W \sqrt{1 - \frac{
        \omega_{-} }{ \hat{r}^{d} }} } \equiv U_n,
\end{aligned}
\end{equation}
together with the constraints 
\begin{align}
    0 &\leq \omega_{-} \leq \hat{r}^{d}, \label{eq:dom1}\\
    W &\leq \sqrt{\hat{r}^{2d-2} - 1}\sqrt{1 - \frac{ \omega_{-}
    }{ \hat{r}^{d} }}. \label{eq:dom2}
\end{align}
Our goal will now be to upper bound $U_n$ for all legal triplets $\left(W, \hat{r},
\omega_{-}\right)$ for $n=1$ and $n=\frac{ d }{ d-1 }$,
which turns out to be values that will give interesting growth bounds.

Note first that we have
\begin{equation}
\begin{aligned}
   \partial_{\omega_-} ^2 U_n \geq 0,
\end{aligned}
\end{equation}
so any local extremum of $U_n$ with respect to $\omega_-$ is a minimum. Thus, for any given $W$ and $\hat{r}$, $U_n$ is maximized when $\omega_{-}$ is on the boundary of its domain. First, take
$\omega_- = 0$. Then, assuming that $ 1 \leq n \leq \frac{ d }{ d-1 }$,
\begin{equation}\label{eq:Unbnd}
\begin{aligned}
    U_n = \frac{ W^{n-1}\sqrt{\hat{r}^{2d-2} - 1} }{ \hat{r}^{d} }\leq \frac{
        \left[ \hat{r}^{2d-2} - 1\right]^{\frac{ n }{ 2 }} }{ \hat{r}^{d} }
    \leq \frac{ \hat{r}^{n(d-1)} }{ \hat{r}^{d} } \leq 1,
\end{aligned}
\end{equation}
where we used \eqref{eq:dom2} in the second inequality.
For $n=1$, we get the stronger bound
\begin{equation}\label{eq:U1bnd}
\begin{aligned}
    U_{1} \leq  \frac{ \sqrt{\hat{r}^{2d-2} - 1} }{ \hat{r}^{d} } \leq
    \sqrt{\frac{d-1}{d^{\frac{ d }{ d-1 }}}} \equiv \alpha_d,
\end{aligned}
\end{equation}
where the upper bound is found by maximizing with respect to $\hat{r}$.
Next, let us look at the maximal value for $\omega_-$, where we have the equality
\begin{equation}
\begin{aligned}
    W = \sqrt{\hat{r}^{2d-2} - 1}\sqrt{1 - \frac{ \omega_- }{ \hat{r}^{d} }}.
\end{aligned}
\end{equation}
Neglecting the first $\omega_-$ in the denominator of $U_n$ and using $W\leq \sqrt{\hat{r}^{2d-2} - 1}$, we get
\begin{equation}
\begin{aligned}
    U_n \leq  \frac{ \left[\hat{r}^{2d-2} -
    1\right]^{n/2} }{
        \hat{r}^{d} }.
\end{aligned}
\end{equation}
But this is just the expression bounded earlier, and so \eqref{eq:Unbnd} and
\eqref{eq:U1bnd} holds generally. Restoring factors of $L, r_0$, we have the
following true bounds
\begin{align}
    W &\leq L^{\frac{ d-2  }{ d }} \omega(\infty)^{\frac{ d-1 }{ d }},
    \label{eq:thinshell1bnd}\\
    W &\leq \alpha_{d}\frac{ L }{ r_0 }\omega(\infty). \label{eq:thinshell2bnd}
\end{align}

Redoing the steps in \eqref{eq:boundsteps} with the numerical factor
from in \eqref{eq:thinshell2bnd}, we get
Theorem~\ref{thm:mainthm4} for strip regions. Next, inserting \eqref{eq:thinshell1bnd} into
\eqref{eq:dAdt}, we find
\begin{equation}\label{eq:thm2proxy}
\begin{aligned}
    \Big|\frac{ \dd S_{R} }{ \dd t }\Big| &\leq \frac{ \text{Area}[\partial R] }{
        4G_N L^{d-2}}L^{\frac{ d-2 }{ d }}\left[\frac{ 16\pi G_N L^{d-1} }{d-1 
    }\left<T_{tt}\right> \right]^{\frac{ d-1 }{ d }}
    = \frac{ 1 }{ 4 }\text{Area}[\partial R] c_{\rm eff} \left[\frac{ 16\pi }{(d-1)c_{\rm eff}}\left<T_{tt}\right>
    \right]^{\frac{ d-1 }{ d }},
\end{aligned}
\end{equation}
proving Theorem~\ref{thm:mainthm2} for strip regions.

\subsection{Multiple strips and mutual information}
Our results not scaling with $\text{Vol}[R]$ can be generalized to regions $R$
consisting of $n$ disjoint finite strips by simply applying
the same argument to each connected component of the HRT surface separately. For
$d=2$ this results in 
\begin{equation}
\begin{aligned}
    \qty|\frac{ \dd S_R}{ \dd t }| \leq n\sqrt{\frac{ 8\pi c }{ 3
    }\left< T_{tt} \right>}.
\end{aligned}
\end{equation}
It is easy to see that \eqref{eq:thm2proxy} also holds true for $n$ strips. 
No modification is needed, since $\text{Area}[\partial R]$ implicitly contains the 
factor of $n$ present in the $d=2$ case.

\begin{figure}
    \centering
    \includegraphics[width=.6\textwidth]{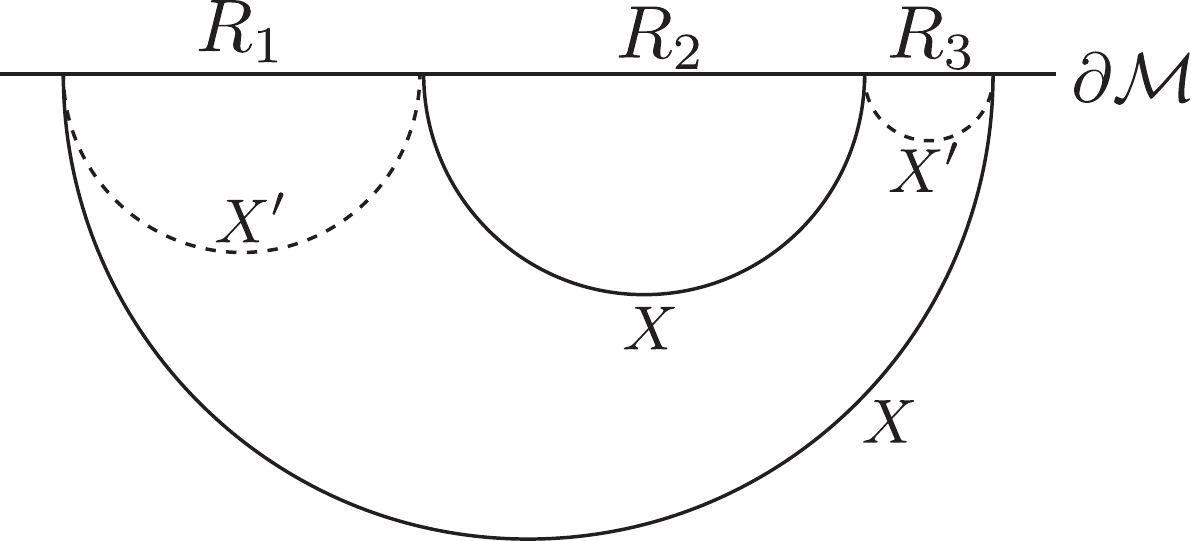}
    \caption{Possible HRT surfaces $X$ and $X'$ of the region $R_1 \cup R_3$, projected onto a
    timeslice.}
    \label{fig:R1R2R3}
\end{figure}

For the bounds scaling like volume, the behavior is different, since
the upper bound depends on the connectivity properties of the entanglement
wedge. Consider for example $d=2$ and the three intervals $R_1, R_2, R_3$ in
Figure \ref{fig:R1R2R3}, and let $R=R_1 \cup R_3$ be the region under consideration. We then see that
\begin{equation}
\begin{aligned}
    \Big|\frac{ \dd S_R }{ \dd t }\Big| \leq \kappa \left<T_{tt}\right>
    \begin{cases}
    \text{Vol}[R] & \text{the entanglement wedge is disconnected}, \\
        \text{Vol}[R]+2\text{Vol}[R_2] & \text{the entanglement wedge is
        connected}, \\
    \end{cases}
\end{aligned}
\end{equation}
where $\kappa$ is the relevant numerical prefactor of either
Theorem~\ref{thm:mainthm3} or \ref{thm:mainthm4}. We get this result by adding
the different volume factors from each connected component of the HRT surface.
Similar games can be played for $n$ strips in $d$ dimensions.

Next, lets us consider $d=2$ and the mutual information between two subsystems $R_1$ and
$R_2$ consisting of $n_1$ and $n_2$ finite intervals, respectively. We then have
\begin{equation}\label{eq:mutualbound}
\begin{aligned}
    \Big|\partial_t I (R_1, R_2)\Big| &= |\partial_t S_{R_1} + \partial_t
    S_{R_2} - \partial_t S_{R_1
    R_2}| 
    \leq |\partial_tS_{R_1}| + |\partial_t S_{R_2}| + |\partial_t S_{R_1
    R_2}| \\
    &\leq \left(2n_1 + 2 n_2 \right)\sqrt{\frac{ 8\pi c }{ 3
    }\left<T_{tt}\right>}.
\end{aligned}
\end{equation}
Using \eqref{eq:shellnonlinearT}, the generalization to higher $d$ is obvious.

\section{Maximal Rates for Balls, Wilson Loops and
Correlators}\label{sec:generaldim}
\subsection{Setup and summary of results}\label{sec:qsummary}
In this section we will consider extremal surfaces $X_t$ of dimension $q+1$
anchored at $q$-dimensional spheres $\partial R_t$ at time $t$ on the conformal
boundary, where extremal means that all the mean curvatures of $X_t$ are
zero.\footnote{This means that the generalized volume of $X_t$ is
stationary under perturbations with compact support. }
We take the spheres to have radius $\mathcal{R}$. For $q=0$, $\partial R_t$ just
consists of two points, and $X_t$ is a
one-parameter family of geodesics. For $q=d-2$, $X_t$ is a one-parameter family
of HRT surfaces anchored at spheres. For $q=1$, $X_t$ are
two-dimensional spacelike worldsheets anchored at circles.

As before we are working with planar symmetric spacetimes, subject to the same
assumptions described in Sec.~\ref{sec:stripsummary}.
The logical steps will be mostly identical to Sec.~\ref{sec:strip}, but with extra
technicalities coming from the curvature of $\partial R_t$. 
Note that since we now have submanifolds of varying dimensions, we will use the
symbol $\norm{\cdot}$ to indicate the measure of the surface in the natural
induced volume form. For quantities on the conformal boundary, $\norm{\cdot}$
means with respect to the induced metric from the Minkowski conformal frame. We
will use $\text{Length}[]$, $\text{Area}[]$ and $\text{Vol}[]$ to refer to the
measure of surfaces of dimension 1, codimension 2, and codimension 1,
respectively.

To describe the relevant subregions in our results, let $\bm{z}$ be Cartesian
coordinates in the direction transverse to the sphere $\partial R_t$. We now
choose coordinates for our Minkowski conformal frame on the boundary to be
\begin{equation}
\begin{aligned}
    \dd s^2 = - \dd t^2 + L^2\left(\dd \phi^2 + \phi^2 \dd \Omega^2_{q} + \dd \bm{z}^2 \right),
\end{aligned}
\end{equation}
with $\dd \Omega^2_{q}$ the metric of a round unit $q$-sphere, and with
the constant$-t$ slices the ones on which one-point functions of local operators
are constant. For $q=0$ there is no
$\dd \Omega^2_q$--term, while for $q=d-2$ there is no $\dd \bm{z}^2$ term.
$\phi$ is a dimensionless radial coordinate on the boundary, and $R_{t'}$ is given by
\begin{equation}
\begin{aligned}
    0 \leq \phi \leq \frac{ \mathcal{R} }{ L }, \qquad t=t', \qquad \bm{z}=0.
    \label{eq:sphereregion}
\end{aligned}
\end{equation}
We now have
\begin{equation}
\begin{aligned}
    \norm{\partial R_t} = \norm{\partial R} = \Omega_{q}\mathcal{R}^{q},
\end{aligned}
\end{equation}
where $\Omega_{q}$ is the volume of a unit $q$-sphere.

Let us now summarize the results proven in this section. 
For entanglement entropy, we will prove the parts Theorems~\ref{thm:mainthm2},
\ref{thm:mainthm3} and \ref{thm:mainthm4} that refer to spherical $\partial R$.
For extremal surfaces of other dimensionalities, the following theorem applies
to the most general class of spacetimes and subregions:
\begin{thm}\label{thm:mainthm5}
    Let $(\mathcal{M},g_{ab})$ be a regular asymptotically AdS$_{d+1}$ spacetime with planar
    symmetry satisfying the DEC. Assume that $d$ is
    even, and let be $X_t$ be an extremal surface of dimension $d/2$, anchored on the conformal boundary at
    the sphere $\partial R_t$. Then
    \begin{equation}\label{eq:mainthm5}
    \begin{aligned}
        \qty|\frac{ \dd }{ \dd t }\norm{X_t}| \leq \norm{\partial
        R}L^{\frac{ d }{ 2 }}\sqrt{\frac{ 16\pi }{ c_{\rm eff}(d-1) }
        \left<T_{tt}\right>}.
    \end{aligned}
    \end{equation}
\end{thm}
\noindent Of course, for $d=2$, this just reduces to Theorem~\ref{thm:mainthm1}.
For $d=4$ this can be converted to the growth bound on circular Wilson loops, given by
 \eqref{eq:Wsym}. Next, for surfaces $X_t$ anchored at small spheres on the
boundary, we get the following:
\begin{thm}\label{thm:mainthm6}
    Let $(\mathcal{M},g_{ab})$ be a regular asymptotically AdS$_{d+1}$ spacetime with planar
    symmetry satisfying the DEC.  Let be $X_t$ be an
    extremal surface of dimension $q+1$, anchored on the conformal boundary at
    the sphere $\partial R_t$ having radius $\mathcal{R}$. Assume
    that
    \begin{equation}
    \begin{aligned}
    q \geq \frac{ d-2 }{ 2 }
    \end{aligned}
    \end{equation}
    and
    \begin{equation}
    \begin{aligned}
        \frac{ \mathcal{R}^{d} \left<T_{tt}\right> }{ c_{\rm eff} } \ll 1.
    \end{aligned}
    \end{equation}
    Then
    \begin{equation}\label{eq:mainthm6}
    \begin{aligned}
        \qty|\frac{ \dd }{ \dd t }\norm{X_t}| \leq \eta_{d,
        q}\norm{\partial
        R} L^{q+1}\mathcal{R}^{d-q-1}\frac{\left<T_{tt}\right>}{c_{\rm eff}}\left[1 +
        \mathcal{O}\left(\frac{ \mathcal{R}^{d} \left<T_{tt}\right> }{ c_{\rm eff} }
        \right) \right],
    \end{aligned}
    \end{equation}
    where
    \begin{equation}
    \begin{aligned}
    \eta_{d, q} = \frac{ 8\pi }{ d-1 }\left[\frac{ \Gamma\left(\frac{ 1 }{ 2(q+1) } \right) }{
        \sqrt{\pi}\Gamma\left(\frac{ q+2 }{ 2(q+1) } \right) }\right]^{d-q-1}.
    \end{aligned}
    \end{equation}
\end{thm}
Using well known dictionary entries, described in Sec.~\ref{eq:d2d4specialbnds},
this converts to growth bounds on the
entanglement of small balls, small circular Wilson loops, and heavy two-point
functions at small separations. Specifically, for the latter two, we get
\begin{equation}\label{eq:OOsmalx}
\begin{aligned}
    \Big|\frac{ \dd }{ \dd t }\log |\left<O(x)
    O(0)\right>_{\rho(t)}|\Big| &\leq \frac{8\pi \Delta }{ c_{\rm eff}
    }|x|\left<T_{tt}\right>\left[1 +
        \mathcal{O}\left(\frac{ |x|^2 \left<T_{tt}\right> }{ c_{\rm eff} }
        \right) \right], \quad d=2,
\end{aligned}
\end{equation}
and
\begin{equation}\label{eq:Wsmallr}
\begin{aligned}
    \Big|\frac{ \dd }{ \dd t  }\log|\left<\mathcal{W}(C)\right>_{\rho(t)}|\Big| &\leq
\frac{8\pi\sqrt{\lambda_{\rm eff}}}{
    (d-1)c_{\rm eff}
}\eta_{d,1}\mathcal{R}^{d-1}\left<T_{tt}\right>\left[1 +
        \mathcal{O}\left(\frac{ \mathcal{R}^{d} \left<T_{tt}\right> }{ c_{\rm eff} }
        \right) \right],\quad  d\in\{3, 4\},
\end{aligned}
\end{equation}
where $\sqrt{\lambda_{\rm eff}} = L^2/\ell_{\rm string}^2$ is the effective 't
Hooft coupling, and $\ell_{\rm string}$ the bulk string length.
Finally, for thin-shell spacetimes, we prove the following:
\begin{thm}\label{thm:mainthm7}
    Let $(\mathcal{M},g_{ab})$ be an asymptotically AdS$_{d+1}$ spacetime with planar
    symmetry satisfying the DEC. Assume that $X_t$ is an extremal surface
    anchored at a boundary sphere of dimension
    \begin{equation}
    \begin{aligned}
    q \geq \frac{ d-2 }{ 2 }.
    \end{aligned}
    \end{equation}
    Next, assume that the bulk matter consists of $U(1)$ gauge fields and a thin shell of matter:
    \begin{equation}
    \begin{aligned}
    \mathcal{T}_{ab} = \mathcal{T}_{ab}^{\rm shell} +
        \mathcal{T}_{ab}^{\rm Maxwell},
    \end{aligned}
    \end{equation}
    where $\mathcal{T}_{ab}^{\rm shell}$ has delta function
    support on a codimension$-1$ worldvolume that is timelike or null, and that separately
    satisfies the DEC. Assume $(\mathcal{M},g)$ is regular, except we do
    not require $g_{ab}$ to be $C^2$ at the shell. Then
    \begin{equation}\label{eq:shellnonlinearTq}
    \begin{aligned}
    \Big|\frac{ \dd }{ \dd t }\norm{X_t}\Big| &\leq \norm{\partial R}
    L^{q+1}\left[\frac{ 16\pi}{ (d-1)c_{\rm eff}
    }\left<T_{tt}\right>\right]^{\frac{ q+1 }{  d}}
    \end{aligned}
    \end{equation}
    and
    \begin{equation}
    \begin{aligned}
        \qty|\frac{ \dd }{ \dd t }\norm{X_t}| \leq \kappa_{d,
        q} \norm{\partial
        R} L^{q+1}\mathcal{R}^{d-q-1}\frac{\left<T_{tt}\right>}{c_{\rm eff}}
    \end{aligned}
    \end{equation}
    with $\kappa_{d,q}$ given by \eqref{eq:kappadq}.
\end{thm}
\noindent The main application of \eqref{eq:shellnonlinearTq} is to bound Wilson loops in
$d=3$, where we get
\begin{equation}\label{eq:d3Wilson}
\begin{aligned}
    \Big|\frac{ \dd }{ \dd t  }\log|\left<\mathcal{W}(C)\right>_{\rho(t)}|\Big| 
    &\leq \frac{ \sqrt{\lambda_{\rm eff}} \text{Length}[C] }{ 2\pi }
    \left[\frac{8\pi}{c_{\rm eff}}\left<T_{tt}\right> \right]^{2/3}.
\end{aligned}
\end{equation}

Let us now turn to the proofs.

\subsection{An implicit solution for the extremal surface location}

As earlier, let $\Sigma$ be the extended planar symmetric homology hypersurface
containing $X$. For the exact same reason as earlier, there
is a unique choice of $\Sigma$. We can now pick coordinates on $\Sigma$ given by
\begin{equation}
\begin{aligned}
    \dd s^2|_{\Sigma} = H_{\mu\nu}\dd y^{\mu}\dd y^{\nu} = B(r) \dd r^2 +
    r^2\left(\dd \phi^2 + \phi^2 \dd \Omega_q^2   + \dd \bm{z}^2 \right).
\end{aligned}
\end{equation}
Again, one such coordinate system covers all of $\Sigma$, as shown in
appendix~\ref{app:nothroat}.

We take our intrinsic coordinates on $X$ to be $(r, \Omega^{i})$, where
$\Omega^i$ are coordinates on the sphere. The embedding coordinates of $X$ in $\Sigma$ reads
\begin{equation}
\begin{aligned}
    X^{\mu} = (r, \phi = \Phi(r), \Omega^{i}, \bm{z}=0),
\end{aligned}
\end{equation}
where the symmetries of the problem dictate $\bm{z}=0$.
The induced metric on $X$ is
\begin{equation}
\begin{aligned}
    \dd s^2|_{X} = \left[B(r) + r^2 \Phi'(r)^2\right]\dd r^2 + r^2 \Phi(r)^2 \dd
    \Omega_q^2.
\end{aligned}
\end{equation}
Now we must implement the condition that $X$ is extremal, which requires us to
compute all its mean curvatures and demand them to be vanishing. To do this, let
$n_a^{I}$ be an orthonormal basis of normal forms to $X$ that are
tangent to $\Sigma$, labeled by $I$.  Let $t_a$ be the future timelike normal orthogonal to $\Sigma$. A complete basis of mean curvatures of $X$ now is
\begin{equation}
\begin{aligned}
    \mathcal{K}^{I} &= h^{ab}\nabla_{a}n_{b}^I, \\
    \mathcal{K}^{0} &= h^{ab}\nabla_{a} t_{b},
\end{aligned}
\end{equation}
where $h^{ab} = g^{ab} + t^a t^b - \delta^{IJ}n^a_I n_J^{b} $. All of these quantities must vanish.
Considering the $\mathcal{K}^{I}$ corresponding the $\bm{z}$ directions, we just get
$0$ by our symmetries. Letting $I=n$ denote the remaining normal direction in
$\Sigma$, we get by direct computation that (see appendix
\eqref{sec:codimqextremal} for some of the required ingredients)
\begin{equation}
\begin{aligned}
   \mathcal{K}^n =
    \frac{ 1 }{ r \sqrt{B}\left[B+r^2 (\Phi')^2\right]^{3/2}}
    \Bigg[&r^2B \Phi'' + (q+1)r^3 (\Phi')^3 \\
   &+ \left((q+2) B - \frac{ 1 }{ 2 }r
    B'\right)r \Phi'   
    - \frac{ q B }{ \Phi }\left( B + r^2 (\Phi')^2 \right)
        \Bigg].
\end{aligned}
\end{equation}
If it was not for the last term, we would reproduce \eqref{eq:planarAeq} by
setting $q=d-2$. The new term is caused by the curvature of $\partial R$. Now, with this last term, we no longer have an explicit analytical
solution (when $q>0$). However, we can find an implicit solution that lets us
proceed. Define 
\begin{equation}
\begin{aligned}
    \chi(r) = \frac{ qB }{ \Phi \Phi' }\left(\frac{ B }{ \Phi'^2 } + r^2 \right),
\end{aligned}
\end{equation}
so that our equation for extremality reads
\begin{equation}\label{eq:qexeq}
\begin{aligned}
    (q+1)r^3 (\Phi')^3 + \left((q+2) B - \frac{ 1 }{ 2 }r B'\right)r \Phi' +
    r^2B \Phi'' - (\Phi')^3 \chi(r) = 0.
\end{aligned}
\end{equation}
Imposing $\Phi(r_0)=0$, where $r_0$ is the tip of the extremal surface, we have
the implicit solution
\begin{equation}
\begin{aligned}
    \Phi(r) = \int_{r_0}^{r} \dd \rho \frac{ \sqrt{B(\rho)} }{ \rho
    \sqrt{\mathcal{C} \rho^{2q+2}h(\rho) - 1  }},
\end{aligned}
\end{equation}
where
\begin{equation}
\begin{aligned}
    h(r) = 1 - \frac{2}{\mathcal{C}} \int_{r_0}^{r}\dd \rho \chi(\rho) \rho^{-6-2q},
\end{aligned}
\end{equation}
for some $\mathcal{C}$ that is fixed by imposing $\Phi'(r_0)=\infty$.
Assuming $h(r_0)$ is finite, we get that $\mathcal{C}=r_0^{-2q-2}$.
This is indeed correct, even though $\chi(r_0)$ looks superficially divergent. Since we are near a minimum of $r(\Phi)$ we have that $r=r_0 +
\mathcal{O}(\Phi^2)$ near $r_0$, and so for $r$ close to $r_0$ we get $\Phi =
\alpha \sqrt{r-r_0}$ for some $\alpha$. Even though $\Phi$ goes to zero at $r_0$
we find that
\begin{equation}
\begin{aligned}
    \chi(r) \sim \mathcal{O}(1) 
\end{aligned}
\end{equation}
near $r_0$, and so $h(r_0) = 1$. Next, reality of $\Phi(r)$ demands that
$h(\rho) \geq (r_0/r)^{2q+2}$, while positivity of $\Phi$ and $\Phi'$ ensures
that $h(r) \leq 1$, and so in total we know that\footnote{By the same kind of
analysis as in appendix~\ref{app:nothroat}, we cannot have additional turning points where
$\Phi'(r)$ diverges, and so we have strict inequality in the lower bound.}
\begin{equation}\label{eq:Phiq}
\begin{aligned}
    \Phi(r) &= \int_{r_0}^{r} \dd \rho \frac{ \sqrt{B(\rho)} }{ \rho
    \sqrt{(\rho/r_0)^{2q+2}h(\rho) - 1  }}, \qquad
    0 < (r_0/r)^{2q+2}  < h(r) \leq 1.
\end{aligned}
\end{equation}

\subsection{The relation between the Hawking masses}
Take $K_{\alpha\beta}$ to be the extrinsic curvature of the extended homology
hypersurface. Like in the case of the strip, we have that
\begin{equation}
    \begin{aligned}\label{eq:muvsomega}
    \mu(r) = \omega(r) + r^{d-4}K_{\phi\phi}(r)^2,
\end{aligned}
\end{equation}
which follows from the same computation as in the previous section, together
with \eqref{eq:Ktrgeneralq} in
appendix~\ref{sec:codimqextremal}. We have that $K_{\phi\phi}\sim \mathcal{O}(1/r^{q-1})$, as becomes clear in the next
section. Thus, at large $r$ we have
\begin{equation}\label{eq:muvsomegaq2}
\begin{aligned}
    \mu(r) = \omega(r) + \mathcal{O}\left(r^{d-2-2q}\right). 
\end{aligned}
\end{equation}
Consequently, $\omega(\infty)$ is proportional to spacetime mass if and only if
\begin{equation}
\begin{aligned}
q > \frac{ d-2 }{ 2 }.
\end{aligned}
\end{equation}
If $2q=d-2$, $\omega(\infty)$ is smaller than $\mu(\infty)$ by some finite
number. For $2q<d-2$, we get $\omega(\infty)=-\infty$ by \eqref{eq:muvsomegaq2} and
the fact that $\mu(\infty)$ is finite and positive. We will see below that this
comes out of the constraint equations, since exactly when $2q<d-2$, $\omega(r)$ is
neither positive nor monotonically increasing. We will not be able to say anything about the case $2q<d-2$.

\subsection{$\partial_t \norm{X_t} \leq $ momentum on $X_t$}\label{sec:matterfluxq}
The time-derivative of the generalized volume satisfies \cite{FisWis16, BaoCao19} 
\begin{equation}\label{eq:dXdteta}
\begin{aligned}
    \frac{ \dd }{ \dd t }\norm{X_t} = \int_{\partial X_t} N^a \eta_a,
\end{aligned}
\end{equation}
where $\eta^a = (\partial_t)^a$  generates the deformation of $\partial X_t$, 
while $N^a$ is the unit vector that is (1) tangent to $X_t$, (2) orthogonal
to $\partial X_t$, and (3) pointing towards the conformal boundary. A
computation in appendix~\ref{app:dSKphi2} shows that \eqref{eq:dXdteta} can be
written as
\begin{equation}\label{eq:dXdt}
\begin{aligned}
    \frac{ \dd }{ \dd t }\norm{X_t}\Big|_{t=0} = \frac{ \norm{\partial R} }{ L^q }\lim_{r
    \rightarrow \infty} r^{q-1}K_{\phi\phi}(r).
\end{aligned}
\end{equation}
Now we again reach the stage where we must write the Einstein constraint
equations as a closed system, which requires us to impose extremality in the timelike direction. 

First, note that from the planar symmetry of $\Sigma$, if $\bm{x}$ are Cartesian
coordinates on the plane containing $R$, then we have that the extrinsic
curvature of $\Sigma$ reads
\begin{equation}
\begin{aligned}
    K_{\mu\nu}\dd y^{\mu} \dd y^{\nu} &= K_{rr}(r)\dd r^2 + K_{\phi\phi}(r) \left(\dd \bm{x}^2 + \dd \bm{z}^2\right) \\
    &= K_{rr}(r)\dd r^2 + K_{\phi\phi}(r) \left(\dd \phi^2 + \phi^2 \dd \Omega^2_{q} +
    \dd \bm{z}^2\right).
\end{aligned}
\end{equation}
Thus, the components of the extrinsic curvature with indices in the sphere
directions reads
\begin{equation}
    \begin{aligned}\label{eq:Kij}
    K_{ij} = K_{\phi\phi}(r) \phi^2 w_{ij},
\end{aligned}
\end{equation}
where  $w_{ij}$ is the unit metric on the $q$-sphere. 
Define again $F(r)$ through the relation $K_{rr}(r) = F(r) B(r)$. Computing
$\mathcal{K}^0 = 0$, using \eqref{eq:Kij}, and solving for $F(r)$ (see
Appendix~\ref{sec:codimqextremal}), 
we get, after substituting our expression for $\Phi(r)$, that
\begin{equation}
    \begin{aligned}\label{eq:Fsolq}
    F(r) &= -\frac{ K_{\phi\phi}(r) }{ r^2 }H(r), 
\end{aligned}
\end{equation}
where we for convenience defined the function 
\begin{equation}
\begin{aligned}
    H(r) = \frac{ q(r/r_0)^{2q+2}h(r) + 1 }{
        (r/r_0)^{2q+2}h(r) - 1 }.
\end{aligned}
\end{equation}
Since $\partial_h H <0$ and $h(r) \leq 1$, we get the lower bound
\begin{equation}
\begin{aligned}
    H(r) \geq \frac{ q(r/r_0)^{2q+2} + 1 }{ (r/r_0)^{2q+2} - 1 } \equiv H_L(r).
\end{aligned}
\end{equation}
The constraint equations \eqref{eq:preConstraint1} and \eqref{eq:preConstraint2} are
unchanged, except now the expression for $F(r)$ is different. Plugging it in we get
\begin{align}
    (d-1) \frac{\omega'(r)}{r^{d-1}} &= 2 \mathcal{E} + \frac{ (d-1) }{ r^4
    }K_{\phi\phi}(r)^2\left[2H(r)  - d + 2\right], \label{eq:omegaQ} \\
    K_{\phi\phi}' + \left[H(r)-1\right] \frac{ K_{\phi\phi} }{ r } &=  - \frac{ r^2 }{ d-1
    }\mathcal{J}(r). \label{eq:KQ}
\end{align}
Now, using the lower bound $H_L(r)$, let us note the following:
\begin{equation}
\begin{aligned}
    2H(r) - d + 2 &\geq \frac{ (2q-d+2)(r/r_0)^{2q+2} + d - 1 }{
        (r/r_0)^{2q+2} - 1}.
\end{aligned}
\end{equation}
This is positive definite for all $r$ only when $q \geq \frac{d-2}{2}$, so for
geodesics ($q=0$), we only have monotonicity of $\omega(r)$ when $d=2$.
But this is just the case studied in the previous section. For ($q=1$), which is the relevant 
case for Wilson loops, we have
monotonicity of $\omega(r)$  only for $d\leq 4$. For an HRT surface we have $q=d-2$, and so
we have monotonicity in all dimensions. It is in fact quite surprising that we
have monotonicity for any $q$ whatsoever, since when looking at the Einstein
constraint equations in covariant form, monotonicity of the Riemannian Hawking mass is only
manifest on hypersurfaces that have vanishing mean curvature. 

Let us assume $2q \geq d-2$ going forward, and
let us bound $K_{\phi\phi}$ and $\mu$ at infinity. 
Fixing an integration constant by demanding that $F(r_0)=\text{finite}$, 
the solution to the momentum constraint is 
\begin{equation}
\begin{aligned}
    K_{\phi\phi}(r) = - \frac{ 1 }{ d-1 }\int_{r_0}^{r}\dd \rho\rho^2 \mathcal{J}(\rho)
    e^{-\int_{\rho}^{r}\dd \frac{ 1 }{ z }\left(H(z) - 1 \right)}.
\end{aligned}
\end{equation}
We have that
\begin{equation}
\begin{aligned}
    |K_{\phi\phi}(r)| 
    &\leq \frac{ 1 }{ d-1 }\int_{r_0}^{r}\dd \rho \rho^2
    |\mathcal{J}(\rho)| e^{-\int_{\rho}^{r}\dd \frac{ 1 }{ z }\left(H_L(z) - 1
    \right)} \\
    &= \frac{ r^2 }{ (d-1)\sqrt{(r/r_0)^{2q+2} -1} }\int_{r_0}^{r}\dd \rho 
    |\mathcal{J}(\rho)| \sqrt{(\rho/r_0)^{2q+2} -1}.\label{eq:Kupper}
\end{aligned}
\end{equation}
We see from this expression that $K_{\phi\phi}\sim \mathcal{O}(1/r^{q-1})$. Also, 
in this last expression, if we replace $|\mathcal{J}|\rightarrow -\mathcal{J}$, we just get the solution of \eqref{eq:KQ} with $H(r)$ replaced by $H_L(r)$. We will use this fact later.

Inserting \eqref{eq:Kupper} in \eqref{eq:dXdt}, we finally get
\begin{equation}
\begin{aligned}
    \Big|\frac{ \dd }{ \dd t }\norm{X_t} \Big| \leq \frac{ (d-1)\norm{\partial R}
    }{ L^q } \int_{r_0}^{\infty}\dd \rho |\mathcal{J}(\rho)| \sqrt{\rho^{2q+2} -
    r_0^{2q+2}}.
\end{aligned}
\end{equation}
Unlike for an HRT surface anchored at a strip, we are here only able to write an
inequality. 

Next, let us turn to the second ingredient: the mass. Rewriting \eqref{eq:omegaQ} in terms of $\mu(r)$, we get
\begin{equation}\label{eq:muEq}
\begin{aligned}
    (d-1) \frac{ \mu'(r) }{ r^{d-1} } = 2\mathcal{E} + \frac{ (d-1) }{ r^{4}
    }K_{\phi\phi}(r)^2\left[2H(r) + d - 6\right] + \frac{ 2(d-1) }{ r^3 }\frac{ \dd }{ \dd r
    }K_{\phi\phi}^2.
\end{aligned}
\end{equation}
After an integration by parts and using $H(r)\geq H_L(r)$, we get that
\begin{equation}
    \begin{aligned}\label{eq:muboundq}
    \mu(\infty) &\geq \mu(r_0) + \frac{ 2 }{ d-1 }\int_{r_0}^{\infty}\dd \rho
    \rho^{d-1}\mathcal{E}(\rho) + \int_{r_0}^{\infty}\dd \rho
    \rho^{d-5}K_{\phi\phi}(r)^2 \left[2H_L(r) + d \right]
\end{aligned}
\end{equation}
where
\begin{equation}
\begin{aligned}
    2H_L(r) + d = \frac{ (d+2q)(r/r_0)^{2q+2} + (d-2)}{ (r/r_0)^{2q+2} - 1 }
    \geq 0.
\end{aligned}
\end{equation}

Possessing now an upper bound on $\partial_t \norm{X_t}$ and a lower bound on
mass, we next need an upper bound on $L^2/r_0$. 

\subsection{Constraints on boundary anchored extremal surfaces}
It turns out that generalizations of Lemmata~\ref{lem:mupos} and \ref{lem:omega0} remain true for the surfaces considered in this section. 
The proof of Lemma~\ref{lem:mupos} is unchanged, while 
 from the discussion in appendix~\ref{app:nothroat} and
 \ref{sec:codimqextremal}, together with the proof of Lemma~\ref{lem:omega0},
 we get the following constraints on the tip of $X$:
\begin{lem}
    Let $(\mathcal{M}, g_{ab})$ be an asymptotically AdS$_{d+1\geq3}$
    spacetime with planar symmetry. Let $X$ be a $(q+1)$-dimensional extremal
    surface anchored at a $q$-sphere on the conformal boundary. Then the tip of
    $X$ lies in an untrapped region. 
    Furthermore, if $(\mathcal{M}, g_{ab})$ is
    regular and satisfies the DEC, then $\omega(r_0)\geq 0$, where $r_0$ is the
    radius of the tip of $X$.
\end{lem}

With this in hand, we readily obtain the spherical dimension--$(q+1)$ version of Theorem~\ref{thm:r0bound}:
\begin{thm}\label{thm:r0bound2}
    Let $(\mathcal{M},g_{ab})$ be a regular planar-symmetric asymptotically AdS$_{d+1\geq 3}$ spacetime
    satisfying the DEC. Let $X$ be a dimension $q+1$
    extremal surface anchored at a sphere of radius
    $\mathcal{R}$. Let be $r_0$ be the
    radius of the plane tangent to the tip of $X$. Then if
    \begin{equation}\label{eq:qbound}
    \begin{aligned}
    q \geq \frac{ d-2 }{ 2 },
    \end{aligned}
    \end{equation}
    we have
    \begin{equation}\label{eq:r0bndq}
    \frac{ L^2 }{ r_0 } \leq \frac{
        \Gamma\left(\frac{ 1 }{ 2(q+1) }
        \right)}{\sqrt{\pi}\Gamma\left(\frac{ q+2 }{ 2(q+1) } \right) }
    \mathcal{R}.
    \end{equation}
\end{thm}
\begin{proof}
    For $2q\geq d-2$, \eqref{eq:omegaQ} implies $\omega'(r)\geq 0$. Combining
    with $\omega(r_0)\geq0$, we get $\omega(r)\geq0$.
    Using now $h(r)<1$ and that $\omega(r)\geq0$ implies $B(r) \geq L/r$, we
    get
\begin{equation}
\begin{aligned}
    \mathcal{R} &= L \Phi(\infty) 
            = L \int_{r_0}^{\infty}\dd r \frac{ \sqrt{B(r)}
            }{ r\sqrt{(r/r_0)^{2q+2}h(r) - 1} } \\
            &\geq L^2  \int_{r_0}^{\infty}\dd r \frac{1 
            }{ r^2 \sqrt{(r/r_0)^{2q+2} - 1} } 
            = \frac{L^2}{r_0} \frac{ \sqrt{\pi} \Gamma\left(\frac{ q+2 }{ 2(q+1) } \right) }{
                \Gamma\left(\frac{ 1 }{ 2(q+1) } \right) }.
\end{aligned}
\end{equation}
\end{proof}

\subsection{Proofs}\label{eq:d2d4specialbnds}
\subsubsection*{Proof of bounds in $d=2$ and $d=4$}
Consider the special dimension 
\begin{equation}
\begin{aligned}
q = \frac{ d-2 }{ 2 },
\end{aligned}
\end{equation}
which can only happen when $d$ is even. As seen previously, we have that
$\omega(r) \geq 0$ in this case. Furthermore,  \eqref{eq:muvsomega} becomes
\begin{equation}
\begin{aligned}
    \mu(\infty) = \omega(\infty) + \left[\lim_{r \rightarrow
    \infty}r^{q-1}K_{\phi\phi} \right]^2,
\end{aligned}
\end{equation}
and so
\begin{equation}
\begin{aligned}
    \Big|\lim_{r \rightarrow \infty} r^{q-1}K_{\phi\phi}(\infty)\Big| \leq \sqrt{\mu(\infty)},
\end{aligned}
\end{equation}
which gives
\begin{equation}\label{eq:specialqbnd}
\begin{aligned}
    \Big|\frac{ \dd }{ \dd t }\norm{X_t}\Big|_{t=0} &\leq \frac{ \norm{\partial R} }{
        L^{q} } \sqrt{\mu(\infty)} 
    &= \norm{\partial R} L^{q+1} \sqrt{\frac{ 16\pi   }{ c_{\rm eff}(d-1)
    }\left< T_{tt} \right>}.
\end{aligned}
\end{equation}
This proves Theorem~\ref{thm:mainthm5}. For $q=0, d=2$, this is just the formula we used to derive \eqref{eq:mainthm1}.

The above result implies a bound on correlators that can be computed using the geodesic
approximation. The geodesic approximation says that the two-point function
of a CFT scalar operator $O$ of large scaling dimension $\Delta \gg
1$ can be computed as
\begin{equation}\label{eq:geoapprox}
\begin{aligned}
    \left<O(\bm{x}) O(0)\right>_{\rho(t)} = \eta
    e^{-\frac{ \Delta }{ L }\norm{X_t}_{\rm reg}},
\end{aligned}
\end{equation}
where $\eta$ is some constant, and $\norm{X_t}_{\rm reg}$ is the
regularized distance of a geodesic  anchored at the points $(t, \bm{x})$ and $(t, 0)$ on the conformal boundary.
We here adopted the Schr\"odinger picture. Combining \eqref{eq:specialqbnd} and
\eqref{eq:geoapprox} and taking $d=2$, we
get\footnote{Of course, we could have derived this in Sec.~\ref{sec:strip}, given that that 
geodesics coincide with the HRT surface when $d=2$.}
\begin{equation}\label{eq:finalCorrbnd}
\begin{aligned}
\Big|\frac{ \dd }{ \dd t }\log \left<O(x)
    O(0)\right>_{\rho(t)}\Big| &\leq \sqrt{\frac{96\pi \Delta^2 }{ c
    }\left<T_{tt}\right>}, \qquad d=2.
\end{aligned}
\end{equation}

Next, for $d=4$, \eqref{eq:specialqbnd} holds for $q=1$, where the $X_t$ are
two-dimensional worldsheets anchored at circles on the boundary. If $\mathcal{W}(C)$ is a
Wilson loop of a circle $C = S^{1}$, we have that \cite{Mal98,Rey98}
\begin{equation}\label{eq:Wformula}
\begin{aligned}
    \left<\mathcal{W}(C)\right>_{\rho(t)} = \eta e^{- \frac{ 1 }{ 2\pi \alpha'
    }\norm{X_t} },
\end{aligned}
\end{equation}
where $\alpha' = \ell_{\rm string}^{2}$ and $\eta$ again some
constant. Combining \eqref{eq:specialqbnd} and \eqref{eq:Wformula} we get
\begin{equation}\label{eq:W4eff}
\begin{aligned}
\Big|\frac{ \dd }{ \dd t  }\log\left<\mathcal{W}(C)\right>_{\rho(t)}\Big| \leq
    \text{Length}[C]\sqrt{\frac{ 4\lambda_{\rm eff} }{ 3\pi c_{\rm eff}
    }\left< T_{tt} \right>}, \qquad d=4,
\end{aligned}
\end{equation}
where $\text{Length}[C]=2\pi \mathcal{R}$. With the precise dictionary for the
duality between type IIB supergravity
on AdS$_5\times S^5$ and $\mathcal{N}=4$ super Yang-Mills with gauge group
$SU(N)$ and 't Hooft coupling $\lambda$, given by
\begin{equation}
\begin{aligned}
    \frac{ G_N }{ R^3 } = \frac{ \pi }{ 2N^2 }, \quad \sqrt{\lambda} = \frac{
        L^2 }{ \alpha' },
\end{aligned}
\end{equation}
\eqref{eq:W4eff} can be written as \eqref{eq:Wsym}.

\subsubsection*{Proof of bounds for small $\mathcal{R}$}
Next we prove bounds that are strong at small radii $\mathcal{R}$. We have:
\begin{lem}
    Consider the same assumptions as in Theorem~\ref{thm:r0bound2}. Assume furthermore that
    $\mathcal{R}^{d}\left<T_{tt}\right>/c_{\rm eff} \ll 1$ and $2q\geq d-2$.
    Then
    \begin{equation}\label{eq:mainlem1q}
    \begin{aligned}
        \Big| \lim_{r \rightarrow \infty} r^{d-3}K_{\phi\phi} \Big| \leq \frac{ L }{ 2r_0^{d-q-1} }
        \omega(\infty)\left[1 + \mathcal{O}\left(\frac{
            \mathcal{R}^{d}\left<T_{tt}\right> }{ c_{\rm eff} }\right)\right].
    \end{aligned}
    \end{equation}
\end{lem}
\begin{proof}
   Define $W=-\lim_{r \rightarrow
    \infty}r^{q-1}K_{\phi\phi}$, and
    assume without loss of generality that $W>0$. Using \eqref{eq:Ksol} and \eqref{eq:omegaQ} and the exact
    same logic as in the proof of Lemma~\ref{lem:smallLmainlem}, we get
    \begin{equation}
    \begin{aligned}
        \frac{ W }{ \omega(\infty) } 
        &\leq  \frac{L\int_{r_0}^{\infty} \dd r r^{q+1}|\mathcal{J}(r)| }{
            2 \int_{r_0}^{\infty}\dd r r^{d}\sqrt{1 - \frac{ \omega(r)L^2 }{ r^{d}
            }}|\mathcal{J}(r)| } \\
        &\leq \frac{L\int_{r_0}^{\infty} \dd r r^{q+1}|\mathcal{J}(r)| }{
            2 \int_{r_0}^{\infty} \dd r r^{d}|\mathcal{J}(r)| }\left[1 +\frac{ L^2
        \omega(\infty) }{ 2r_0^{d} }
         + \ldots \right] \\
         &\leq \frac{ L }{ 2r_0^{d-q-1} }\left[1 + \mathcal{O}\left(\frac{
             \mathcal{R}^{d}\left<T_{tt}\right> }{ c_{\rm eff} } \right)
         \right].
    \end{aligned}
    \end{equation}
\end{proof}
\noindent Inserting now \eqref{eq:mainlem1q} into \eqref{eq:dXdt}, we get
\begin{equation}\label{eq:masterQpert}
\begin{aligned}
    \Big| \frac{ \dd }{ \dd t }\norm{X_t}\Big|_{t=0} 
    &\leq \frac{ 8\pi \norm{\partial R} L^{1+q}}{(d-1)c_{\rm eff}
    }  \frac{ L^{2(d-q-1)} }{ r_0^{d-q-1} }\left<
    T_{tt} \right> \\
    &\leq \frac{ \norm{\partial R} L^{1+q}}{c_{\rm eff}
    } \eta_{d,q} \mathcal{R}^{d-q-1}\left<
    T_{tt} \right> \\
\end{aligned}
\end{equation}
where
\begin{equation}
\begin{aligned}
    \eta_{d, q} = \frac{ 8\pi }{ d-1 }\left[\frac{ \Gamma\left(\frac{ 1 }{ 2(q+1) } \right) }{
        \sqrt{\pi}\Gamma\left(\frac{ q+2 }{ 2(q+1) } \right) }\right]^{d-q-1}.
\end{aligned}
\end{equation}
We can convert this to bounds on two-point functions and circular Wilson
loops. Combining \eqref{eq:masterQpert} with \eqref{eq:geoapprox} and
\eqref{eq:Wformula}, we get the bounds \eqref{eq:OOsmalx} and
\eqref{eq:Wsmallr}.

Finally, with $q=d-2$ and $d>2$, the entanglement entropy of small spheres is bounded as
\begin{equation}
\begin{aligned}
    \Big|\frac{ \dd S_R }{ \dd t } \Big| \leq
     \frac{2\sqrt{\pi} \Gamma\left(\frac{ 1 }{ 2(d-1) } \right)
     }{\Gamma\left(\frac{ d }{ 2(d-1) } \right) }
     \text{Vol}[R]\left<T_{tt}\right> + \ldots
     \label{eq:smallsphereent}
\end{aligned}
\end{equation}
where we used that $\text{Vol}[R]=\text{Area}[\partial R]\mathcal{R}/(d-1)$. The
proves the part of Theorem~\ref{thm:mainthm3} where $\partial R$ is a sphere.

\subsubsection*{Proof of bounds for thin-shell spacetimes}
Finally, let us prove our thin-shell results valid all for $\mathcal{R}$, assuming $2q\geq
d-2$. Since we already have general bounds for two-point
correlators in $d=2$ and Wilson loops in $d=4$, this section is mostly relevant for
entanglement in medium or large balls in general $d$, and for medium
and large Wilson loops in $d=3$. 
We consider the same setup as in
Sec.~\ref{sec:thinshellstrip}, and use the same notation. Again, we choose $r_0=L=1$. 

Now, let us consider the solutions \eqref{eq:omegaQ} and
\eqref{eq:KQ} with the replacement $H(r)\rightarrow H_L(r)$. As discussed in
Sec.~\ref{sec:matterfluxq}, this gives a smaller value for $\mu(\infty)$ and
larger value for $\Big|\lim r^{q-1}K_{\phi\phi}\Big|$ if $\mathcal{J}$ has a fixed sign,
which is the case here. Since we will consider
bounds of the form $\lim r^{q-1}K_{\phi\phi}\leq \# \mu(\infty)^{n}$ for 
$n>0$, the bounds we obtain with this replacement will be valid for the original spacetime.

Now, with a delta function shock, solution for $K_{\phi\phi}$ reads
\begin{equation}
\begin{aligned}
    K_{\phi \phi }(r) 
    &=  - \frac{ r^2 }{ d-1 }\eta \sqrt{ \frac{ \hat{r}^{2q+2} - 1
    }{ r^{2q+2} - 1 }  }\theta(r-\hat{r})
    = - \frac{ r^2 W }{ \sqrt{r^{2q+2} -  1}}\theta(r-\hat{r}).
\end{aligned}
\end{equation}
From \eqref{eq:muboundq} we get that the contribution to $\mu(\infty)$ from the
extrinsic curvature reads
\begin{equation}
\begin{aligned}
    Q(\infty) &\equiv \int_{r_0}^{\infty}\dd \rho \rho^{d-5}K_{\phi\phi}(r)^2
    \frac{ (d+2q)r^{2q+2} + (d-2) }{ r^{2q+2} - 1 } \\
    &= W^2 \frac{ \hat{r}^{d} }{ \hat{r}^{2q+2} - 1 }.
\end{aligned}
\end{equation}
The analysis of how the DEC changes across the shock is unchanged from the strip
case, except for a few exponents, and we find
\begin{equation}
\begin{aligned}
    \omega_+ = \omega_{-} + \frac{ \hat{r}^{d}
    }{\sqrt{\hat{r}^{2q+2}-1 } }W \left[\sqrt{ 1 - \frac{
        \omega_{-} }{ \hat{r}^{d} } }  - \frac{ W }{
        \sqrt{\hat{r}^{2q+2} - 1} }\right].
\end{aligned}
\end{equation}
By the same logic as in Sec.~\ref{sec:thinshellstrip} we get,
\begin{equation}
\begin{aligned}
    \frac{ W^{n} }{ \mu(\infty) } \leq \frac{ W^{n} }{
        \omega_{-} + W\frac{\hat{r}^{d}}{\sqrt{\hat{r}^{2q+2} -
        1}}\sqrt{1 - \frac{ \omega_{-} }{ \hat{r}^{d} }} } \equiv U_n
\end{aligned}
\end{equation}
where 
\begin{align}
    0 &\leq \omega_{-} \leq \hat{r}^{d}, \\
    W &\leq \sqrt{r^{2q+2} - 1}\sqrt{1 - \frac{ \omega_{-} }{
        \hat{r}^{d} }}. \label{eq:Wmaxq}
\end{align}
Again, it now suffices to take $\omega_-$ at the boundary of its allowed domain.
With $\omega_-=0$ and $1 \leq n \leq \frac{ d }{ q+1 }$ we get
\begin{equation}
\begin{aligned}
    U_n = \frac{ W^{n-1} \sqrt{\hat{r}^{2q-2}
    - 1}}{ \hat{r}^{d} } \leq \frac{ [r^{2q+2} - 1]^{n/2}  }{ \hat{r}^{d} } \leq
    1,
\end{aligned}
\end{equation}
and for $n=1$ we get the stronger bound
\begin{equation}
\begin{aligned}
    U_1 \leq \sqrt{\frac{ q+1 }{ d }}\left( \frac{ d }{ d-1-q } \right)^{\frac{
        q + 1 - d}{ 2(q+1) }} \equiv \alpha_{d, q}.
\end{aligned}
\end{equation}
For saturation of \eqref{eq:Wmaxq}, neglecting the first $\omega_-$ in the
denominator of $U_n$ and using that $W\leq  \sqrt{\hat{r}^{2q+q}-1}$, we get
\begin{equation}
\begin{aligned}
    U_n \leq \frac{ [r^{2q+2} - 1]^{n/2} }{
        \hat{r}^{d} } \leq 1.
\end{aligned}
\end{equation}
Restoring factors of $L, r_0$, we have the following general bounds
\begin{align}
    W &\leq \alpha_{q, d}\frac{ L }{ r_0^{d-q-1} }\omega(\infty),
    \label{eq:U1bndq}\\ 
    W &\leq L^{\frac{ 2q+2-d }{ d }} \omega(\infty)^{\frac{ q+1 }{ d }}.\label{eq:Unbndq}
\end{align}

Combining \eqref{eq:U1bndq} with \eqref{eq:dXdt} and \eqref{eq:qbound} now gives that
\begin{equation}
\begin{aligned}
    \Big|\frac{ \dd }{ \dd t }\norm{X_{t}}\Big| \leq \kappa_{d, q}\norm{\partial
    R} L^{1+q}\mathcal{R}^{d-q-1}\frac{ \left<T_{tt}\right> }{ c_{\rm eff} },
\end{aligned}
\end{equation}
where
\begin{equation}\label{eq:kappadq}
\begin{aligned}
    \kappa_{d, q} = 2\alpha_{d,q }\eta_{d,q} = \frac{16\pi}{d-1}\sqrt{\frac{ q+1 }{ d }}\left( \frac{ d }{ d-1-q } \right)^{\frac{
        q + 1 - d}{ 2(q+1) }}\left[\frac{ \Gamma\left(\frac{ 1 }{ 2(q+1) } \right) }{
        \sqrt{\pi}\Gamma\left(\frac{ q+2 }{ 2(q+1) } \right) }\right]^{d-q-1}.
\end{aligned}
\end{equation}
This shows that the type of bounds derived in the small $\mathcal{R}$ limit holds in
thin shell spacetimes for all $\mathcal{R}$, at price of a larger prefactor. 
For the entanglement entropy of balls, we get
\begin{equation}
\begin{aligned}
    \Big|\frac{ \dd S_R}{ \dd t }\Big| \leq 4\sqrt{\frac{ \pi(d-1)  }{
        d^{\frac{ d }{ d-1 }} }} \frac{
        \Gamma\left(\frac{ 1 }{ 2(d-1) } \right) }{
        \Gamma\left(\frac{ d }{ 2(d-1) } \right)
    }\text{Vol}[R]\left<T_{tt}\right>.
\end{aligned}
\end{equation}
For Wilson loops we get
\begin{equation}
\begin{aligned}
    \Big|\frac{ \dd }{ \dd t  }\log|\left< \mathcal{W}(C)\right>_{\rho(t)}|\Big|
    &\leq \frac{\sqrt{\lambda_{\rm eff}}}{c_{\rm eff}}\mathcal{R}^{d-1}\left<T_{tt}\right> \begin{cases}
        \frac{ \sqrt{128\pi}\Gamma\left(1/4\right) }{
            3^{3/4}\Gamma\left(3/4\right) } \approx 26  & d =3 \\
        \frac{ 8\Gamma(1/4)^2 }{ 3\Gamma(3/4)^{2} } \approx 23 & d =4 \\
\end{cases}.
\end{aligned}
\end{equation}

Next, consider \eqref{eq:Unbndq}. This gives us that
\begin{equation}\label{eq:specialqbnd2}
\begin{aligned}
    \Big|\frac{ \dd }{ \dd t }\norm{X_t}\Big|_{t=0} &\leq \norm{\partial R}
    L^{q+1}\left[\frac{ 16\pi}{ (d-1)c_{\rm eff}
    }\left<T_{tt}\right>\right]^{\frac{ q+1 }{  d}}
\end{aligned}
\end{equation}
For $q=d-2$, corresponding to entanglement for
ball subregions, this just gives \eqref{eq:shellnonlinearT},  verifying that it
holds for spheres as well, completing the proof of the part of Theorem~\ref{thm:mainthm4} concerning spherical $\partial
R$.

For $q=0, d=2$ and $q=1, d=4$ we just
reproduce the bounds of Sec.~\ref{eq:d2d4specialbnds}, which are anyway proven
with weaker assumptions there. For $q=1, d=3$ we get a new
bound on Wilson loops, given by \eqref{eq:d3Wilson}.

\section{Bounding Spatial Derivatives}\label{sec:spatialbnds}
The technology we have developed to bound time derivatives also lets us bound
spatial derivatives of extremal surface areas for strips. 

Consider a one parameter family of strips $R_{\ell}$ of variable width $\ell$ at
some fixed boundary time, given by
\eqref{eq:stripregiondef} with $t'$ now held fixed. Let $X_{\ell}$ be the corresponding one-parameter
family of HRT surfaces. A computation in appendix \ref{app:dSKphi} gives that
\begin{equation}\label{eq:dSr0}
\begin{aligned}
    \frac{ \dd }{ \dd \ell }\text{Area}[X_{\ell}]= \frac{ \text{Area}[\partial
    R] }{ L^{d-1} }r_0^{d-1}.
\end{aligned}
\end{equation}
For a strip, we thus see that the depth of the HRT surface tip uniquely determines
$\partial_{\ell} S$. Using our lower bound on $r_0$ given by \eqref{eq:r0bnd}, we now immediately get
the following:
\begin{thm}\label{thm:dSdl}
    Let $(\mathcal{M},g_{ab})$ be a regular asymptotically AdS$_{d+1\geq 3}$ spacetime with planar
    symmetry satisfying the DEC. If $X_{\ell}$ is the HRT surface of
    a strip $R_{\ell}$ of width $\ell$, then
    \begin{equation}\label{eq:dSdx}
        \begin{aligned}
        \frac{ \dd }{ \dd \ell }\left[\frac{ {\rm Area}[X_{\ell}] }{ 4G_N }
        \right] \geq  \frac{ c_{\rm eff} }{ 4\ell^{d-1} }{\rm Area}[\partial
        R]
        \left[\frac{ 2\sqrt{\pi} \Gamma\left(\frac{ d }{2(d-1) } \right) }{
            \Gamma\left(\frac{ 1 }{ 2(d-1) }\right) }\right]^{d-1} .
    \end{aligned}
    \end{equation}
\end{thm}
\noindent The lower bound is equal to
$\partial_{\ell}S_{\rm vacuum}$, and so we get
\begin{equation}\label{eq:dDSdx}
\begin{aligned}
    \frac{ \dd }{ \dd \ell }\Delta S[\rho_{{R_{\ell}}}]\geq 0
\end{aligned}
\end{equation}
where $\Delta S$ is the vacuum subtracted entropy. Since we get the vacuum
entanglement entropy in the limit $\ell \rightarrow 0$, this implies that
\begin{equation}\label{eq:dDSdx2}
\begin{aligned}
\Delta S \geq 0.
\end{aligned}
\end{equation}
It is easy to see that \eqref{eq:dDSdx} and \eqref{eq:dDSdx2} applies to a
subregion $R$ corresponding to a union of any number of finite width strips, with
$\partial_{\ell}$ now interpreted as the derivative with respect to increasing
width of one or more of the connected components.

For $d=2$, we also get a bound on correlators of heavy scalar single trace primaries. Working at a fixed moment of time with a
homogeneous state $\rho$, the combination of \eqref{eq:dSr0}, \eqref{eq:r0bnd} and
\eqref{eq:geoapprox} for $x>0$ gives
\begin{equation}\label{eq:dlnoo}
\begin{aligned}
    \frac{ \dd }{ \dd x }\ln \left<\mathcal{O}(x)\mathcal{O}(0) \right>_{\rho} 
    \leq \frac{ \dd }{ \dd x
    }\ln \left<\mathcal{O}(x) \mathcal{O}(0) \right>_{\rm vacuum} = -\frac{
        2\Delta }{ x}
\end{aligned}
\end{equation}
which means that correlations must die of faster than the vacuum for the states
and operators covered by our assumptions. This in particular implies that
\begin{equation}
\begin{aligned}
    \left<\mathcal{O}(x)\mathcal{O}(0)\right>_{\rho} \leq \left<\mathcal{O}(x)
    \mathcal{O}(0) \right>_{\rm vacuum},
\end{aligned}
\end{equation}
since we just get the vacuum correlator as $x\rightarrow 0$.

\section{Evidence for Broader Validity of Bounds}\label{sec:numerics}
In the previous sections, we have shown that the DEC allows us to prove several
general bounds on the growth of entanglement, correlators and Wilson loops.
However, the proofs crucially relied on the dominant energy condition. 
While the dominant energy condition holds in type IIA, IIB and
eleven-dimensional supergravity (see for example the appendix of \cite{EngFol21b}), it is typically violated after dimensional
reduction \cite{HerHor03}. The prototypical example is a scalar field dual to a relevant
CFT operator. This field has negative mass squared, leading to DEC violation. 

Even though our proofs assumed the DEC, we will now provide strong evidence for
a subset of the bounds that they hold when the DEC is violated 
in reasonable ways. That is, we provide evidence in scalar theories that
violate the DEC, but which have proven positive mass theorems \cite{Bou84,Tow84} (so pure AdS is
stable) and respect the null energy condition (NEC). This is evidence that our bounds are true 
even in CFTs with DEC-violating bulks, since the NEC and a stable vacuum are
both necessary conditions for sensible bulk theories.\footnote{If we choose
completely arbitrary bulk scalar theories, we have no positive mass theorem, so we
can have $\left<T_{tt}\right><0$ in a homogeneous state, and the
bounds are clearly violated. But in these situations there is no sensible
holographic dual. }

In fact, we provide evidence not just that
\begin{equation}
\begin{aligned}
    \mathfrak{R} \leq 1,
\end{aligned}
\end{equation}
but also that when $d>2$,
\begin{equation}
\begin{aligned}
    \mathfrak{R} \leq v_E^{\rm (SAdS)}+\delta v_E < 1
\end{aligned}
\end{equation}
for some small $\delta v_E$ that seems to depend on the scalar potential. Here
\begin{equation}\label{eq:vEsAdS}
\begin{aligned}
    v_{E}^{\rm (SAdS)} = \sqrt{ \frac{ d }{ d-2 } }\left(\frac{ d-2 }{ 2(d-1) }
    \right)^{\frac{ d-1 }{ d }}
\end{aligned}
\end{equation}
is the entanglement velocity computed in a quench in holography, with the final
state being neutral, dual to the AdS-Schwarzschild type black brane \cite{LiuSuh13a, LiuSuh13b}.

The theories we will consider are neutral scalars minimally coupled to gravity,
\begin{equation}
\begin{aligned}
    8\pi G_N\mathcal{L} = \frac{ 1 }{ 2 }R - \frac{ d(d-1) }{ 2L^2 } - \frac{ 1
    }{ 2 }|\dd \phi|^2 - V(\phi),
\end{aligned}
\end{equation}
where $V$ is negative somewhere, leading to violation of the DEC (but not the
NEC). These theories are common in consistent truncations and dimensional reductions of type IIA, IIB, and
eleven-dimensional supergravity \cite{CveDuf99,LuPop99,LuPop99b,CveLu99}. We consider these theories because, for standard forms of minimally coupled bosonic matter, 
neutral scalars appear to pose the biggest risk to our bounds. This is because gauge fields give
no direct contribution to $\mathcal{J}$, and they have a manifestly positive contributions to
the mass (they respect the DEC). 

For free theories where $V=\frac{ 1 }{ 2 }m^2\phi^2$, in order
to maximize the chance of violating our bounds, we choose potentials that are
close to ``maximally negative'', meaning  we pick $m^2$ just slightly above the Breitenlohner-Freedman
\cite{BreFre82,BreFre82b} bound:
\begin{equation}
\begin{aligned}
    m^2 L^2 \geq m_{\rm BF}^2 L^2 \equiv -(d/2)^2.
\end{aligned}
\end{equation}
It is known that if $m^2<m_{\rm BF}^2$, AdS is unstable, and so these theories cannot
be dual to CFTs with a Hamiltonian that is bounded from below. 
Additionally, to have an example of an interacting potential, in $d=3$ we
consider a top down potential that becomes exponentially negative for large
$|\phi|$.

\subsection*{The numerical method}
Let us now explain our procedure. For a given $V(\phi)$ and spacetime
dimension, we will construct an $n$-parameter family of initial data,
parametrized by coefficients $\{f_i\}_{i=1}^{n}$. The data will be provided on
an extended homology hypersurface of some HRT surface.
Then we will define the
function $\mathcal{A}\left(\{f_i\}\right)$ to be equal to the
ratio 
\begin{equation}
\begin{aligned}
\left|\lim_{r \rightarrow \infty}r^{d-3}K_{\phi\phi} \right|/\mu(\infty)^{\frac{ d-1 }{ d }}
\end{aligned}
\end{equation}
in the initial dataset specified by parameters $\{f_i\}$. Different initial
datasets correspond to different moments of time in different spacetimes (with
different sizes of $R$). The value of $\mathcal{A}$
in some particular initial dataset corresponds to the instantaneous entanglement
velocity $\mathfrak{R}$ in that configuration, and we will do a numerical
maximization of $\mathcal{A}$ with respect to the parameters $\{f_i\}$.
If we find that $\mathcal{A}$ is upper bounded, and that the upper bound is $\mathcal{A}_{\rm max}$, we have provided evidence that
\begin{equation}
\begin{aligned}
    \qty|\frac{ \dd  S_R}{ \dd t }| \leq \frac{ 1 }{ 4 } \mathcal{A}_{\rm max} {\rm Area}[\partial 
    R]c_{\rm eff} \left[\frac{ 16\pi }{ (d-1)c_{\rm eff}
    }\left<T_{tt}\right> \right]^{\frac{ d-1 }{ d }}.
\end{aligned}
\end{equation}
If $\mathcal{A}$ is not upper bounded, or if $\mathcal{A}_{\rm max}>1$, we have
a counterexample to $\mathfrak{R}\leq 1$ in the theory under consideration.

We will also evaluate the function $\mathcal{B}(\{f_i\})$, which we define as the value of
\begin{equation}
\begin{aligned}
    \frac{ |\partial_t S_R|}{ \text{Vol}[R]\left<T_{tt}\right> } 
= \frac{ 4\pi }{ d-1 }\frac{ L }{ \ell } \frac{ \qty|\lim r^{d-3}K_{\phi\phi} |
    }{ \mu(\infty)}
    \times
    \begin{cases}
    2 & d=2  \\
    1 & d>2 
    \end{cases}
\end{aligned}
\end{equation}
for any given initial dataset. By
the same logic as earlier, if $\mathcal{B}$ is upper bounded by $\mathcal{B}_{\rm max}$, we
have evidence that
\begin{equation}
\begin{aligned}
 \qty|\frac{ \dd }{ \dd t }S[\rho_R(t)]|\leq \mathcal{B}_{\rm max}
    \text{Vol}[R]\left<T_{tt}\right>.
\end{aligned}
\end{equation}
If $\mathcal{B}$ is not upper bounded, then our volume-type bounds break without
the DEC.

Assuming we work with strips $R$, for a single evaluation of $\mathcal{A}$ and $\mathcal{B}$, we need to numerically solve the ODEs 
given by \eqref{eq:omegaeq} and \eqref{eq:Keq}. For simplicity we will
restrict to strips, since for spheres we cannot solve for $\Phi(r)$
analytically. In this case we would need to solve a set of three coupled equations
instead. 

Let us now specify our family of initial data. An explicit computation gives that
\begin{equation}
\begin{aligned}
    \mathcal{E} &= \frac{ 1 }{ 2 }\dot{\phi}(r)^2 + \frac{ 1 }{ 2 }\left(\frac{ r^2 }{ L^2 } -
    \frac{ \omega(r) }{ r^{d-2} }\right)\phi'(r)^2 + V(\phi), \\
    \mathcal{J} &= \dot{\phi}(r) \phi(r),
\end{aligned}
\end{equation}
where $\dot{\phi} = t^a \nabla_a \phi|_{\Sigma}$. Specifying an initial dataset
now corresponds to specifying the two profiles $\phi(r)$ and $\dot{\phi}(r)$,
together with the initial value of $\omega(r_0)$. Letting 
\begin{equation}
\begin{aligned}
    \Delta = d/2+\sqrt{(d/2)^2 + m^2 L^2}
\end{aligned}
\end{equation}
be the scaling dimension of the CFT operator dual $O$ to $\phi$, the profiles we consider are 
\begin{equation}
\begin{aligned}
    \phi &= f_1 \exp[-\left(\frac{ r-f_2 }{ f_3 }\right)^2] + \frac{ f_4 }{
        r^{\Delta} } +\frac{ f_5 }{ r^{\Delta+2} }\\
    \dot{\phi} &= f_6 \exp[-\left(\frac{ r-f_7 }{ f_8 }\right)^2] + \frac{ f_9 }{
        r^{\Delta+1} } + \frac{ f_{10} }{
        r^{\Delta+3} } \\
\end{aligned}
\end{equation}
which gives a ten-parameter family of initial data. The gaussians give
localized lumps of matter, while the power law falloffs ensures that we can turn on
a VEV of $O$ in the CFT, with $\left<O\right>\propto
f_4$ and $\left<\partial_t O\right>\propto f_8$. Note that the
seemingly unusual $1/r^{\Delta+1}$ falloff in $\dot{\phi}$ is just caused by the fact that the
time derivative is with respect to a unit normal rather than the more standard
global time coordinate near the conformal boundary.

What remains is to pick $\omega(r_0)$. To minimize the CFT energy, we want $\omega(r_0)$ small. 
When the DEC holds, we know that AdS hyperbolicity implies that $\omega(r_0)\geq
0$, as proven in Lemma \ref{lem:omega0}.
However, without the DEC we can have that $\omega(r_0)$ is negative, but not arbitrarily
negative. If we pick $\omega(r_0)$ too negative, it will forbid an embedding of $\Sigma$ in a complete slice.
The difficulty is that how negative $\omega(r_0)$ can be depends on $\phi(r_0)$,
$\dot{\phi}(r_0)$, and $V(\phi)$. We will thus restrict to $\omega(r_0)=0$ and
relegate a more complete study of the future. Even with $\omega(r_0)=0$, it is
far from obvious if our results survive breaking of the DEC, as we can easily
obtain large regions of $\omega < 0$ even with $\omega(r_0)=0$.

Finally, we need to deal with invalid datasets. For a given scalar profile,
it could be that $\omega(r)$ overshoots $r^{d}/L^2$. In this case, the relevant
solution does not correspond to a spacelike hypersurface, and so it must be
discarded. In this case we conventionally define $\mathcal{A}=\mathcal{B}=0$.
Consequently, the functions we are maximizing will have discontinuities. 

We are now ready to proceed to the numerical results.

\subsection*{$d=2$}
We now consider a free massive scalar field with mass
\begin{equation}
\begin{aligned}
m^2 L^2 = 0.9m_{\rm BF}^2 L^2 = -0.9,
\end{aligned}
\end{equation}
dual to a relevant operator with $\Delta \approx 1.32$.
We do not consider saturation of the BF bound, since this requires modification
of the mass formula, and additionally causes $|\partial_t S|$ to be divergent
(for any $d$). Furthermore, we do not want to go too close to the BF bound, since then
$\omega(r)$ converges slowly at large $r$, and so the numerical maximization procedure
becomes prohibitively expensive.

Using now Mathematica's built in \Verb|NMaximize| function, trying all methods for
nonconvex optimization implemented in Mathematica and picking the best result, we find that
\begin{equation}
\begin{aligned}
    \mathcal{A}_{\rm max} \approx 0.999 \leq v_E|_{d=2} = 1.
\end{aligned}
\end{equation}
Thus, in $d=2$ we have evidence that \eqref{eq:mainthm1} holds without the DEC -- at
least in free tachyonic scalar theories.

Next, maximizing $\mathcal{B}$, we find that
\begin{equation}
\begin{aligned}
\mathcal{B}_{\rm max} \approx 3.29 \leq \kappa_{d=2} = 2\pi.
\end{aligned}
\end{equation}
This provides evidence that \eqref{eq:entgrowthbound2} holds when the DEC is
violated, and that the
$\mathcal{O}\left(\ell^{d}\left<T_{tt}\right>/c_{\rm eff}\right)$ corrections
are not needed,
even though we could not prove their absence outside thin
shell spacetimes. In fact, given the large gap between $\mathcal{B}_{\rm max}$
and $\kappa_{d=2}$, the numerical results suggest that our proofs might possibly be
sharpened.

\subsection*{$d=3$}
Now we consider two potentials:
\begin{equation}
\begin{aligned}
    V_{\rm I}(\phi) &= \frac{ 1 }{ 2 }\left(0.9m_{\rm BF}^2 \right)\phi^2, \\
    V_{\rm II}(\phi) &= 1 -\cosh\sqrt{2}\phi,
\end{aligned}
\end{equation}
with $\phi$ dual to operators with scaling dimensions $\Delta_{\rm I}\approx 1.97$ and
$\Delta_{\rm II} = 2$, respectively. 
The potential $V_{\rm II}$ comes from a consistent truncation and dimensional
reduction of
eleven-dimensional SUGRA on AdS$_{4} \times S^7$ \cite{CveDuf99}. We find
\begin{equation}
\begin{aligned}
    \mathcal{A}_{\rm I, max} &\approx 0.693, \\
    \mathcal{A}_{\rm II, max} &\approx 0.702. \\
\end{aligned}
\end{equation}
In both cases $\mathcal{A}_{\rm max}<1$, and so we have evidence that the conjectured bound
\eqref{eq:shellnonlinearT} is true -- even without the DEC and outside thin-shell spacetimes. 

Now, we have that
\begin{equation*}
\begin{aligned}
    v_E^{\rm (SAdS)} = \frac{ \sqrt{3} }{ 2^{4/3} } = 0.687\ldots
\end{aligned}
\end{equation*}
In both cases $\mathcal{A}_{\rm max}$ is close to $v_E^{\rm (SAdS)}$, although
it is slightly larger.  It seems possible that a stronger bound
\begin{equation}
\begin{aligned}
\mathfrak{R}\leq v_E^{\rm (SAdS)}+\delta v_E
\end{aligned}
\end{equation}
is true for some small $\delta v_E$ that potentially depends on the scalar
potential.

For $\mathcal{B}$ we find
\begin{equation}
\begin{aligned}
    \mathcal{B}_{\rm I, max} &\approx 1.71 \leq \kappa_{d=3}\approx 2.62, \\
    \mathcal{B}_{\rm II, max} &\approx 1.72.
\end{aligned}
\end{equation}
Again, there is a significant gap, with the implications being the same as for
$d=2$.

\subsection*{$d=4$}
We now consider
\begin{equation}
\begin{aligned}
    V(\phi) &= \frac{ 1 }{ 2 }\left(0.9m_{\rm BF}^2 \right)\phi^2.
\end{aligned}
\end{equation}
and find
\begin{equation}
\begin{aligned}
    \mathcal{A}_{\rm max} &\approx 0.643.
\end{aligned}
\end{equation}
Again we find evidence that \eqref{eq:introthm2} is true without the DEC
or outside thin-shell spacetimes.
We have
\begin{equation}
\begin{aligned}
    v_E^{\rm (SAdS)} = \frac{ \sqrt{2} }{ 3^{3/4} } = 0.620\ldots,
\end{aligned}
\end{equation}
and so that the instantaneous growth can be above $v_{E}^{(\rm SAdS)}$, but possibly only slightly so.

We also find
\begin{equation}
\begin{aligned}
\mathcal{B}_{\rm max} = 1.91 \leq \kappa_{d=4} \approx 2.43.
\end{aligned}
\end{equation}
Again, there is a significant gap, with the implications being the same as for
$d=2$.

\section{Discussion}\label{sec:discussion}
\begin{table}[t!]
\caption{Proven bounds on entanglement, spatial Wilson loops and equal-time correlators. 
    We suppress $O(1)$ numerical constants in the table. Dots mean corrections scaling as
    $\mathcal{O}(\ell^{d}\left<T_{tt}\right>/c)$ where $\ell$ is the
    relevant characteristic length scale, corresponding to strip width or ball
    radius. We abbreviate the effective central charge and 't Hooft coupling as
    $c$ and $\lambda$, respectively. For proof validity equal to quench+, we mean proofs valid for states dual to spacetimes 
    with thin-shell matter, which includes quenches as a subset. 
    }
    \label{tab:bounds}
\centering
    \begin{tabular}{ |l | l | l | l | l |}
    \hline
        {\boldmath $\partial_t S \leq $} & {\boldmath $d$} & {\boldmath
        \textbf{Region} $R$} &
        \textbf{Proof validity} & \textbf{Eq.}\\
    \hline
        \(\sqrt{\left<T_{tt}\right>/c}\) & $2$ & $n$ intervals & general & \eqref{eq:mainthm1}\\
        \(\text{Vol}[R]\left<T_{tt}\right>+\ldots\) & $\geq 2$ & small strip or ball & general & \eqref{eq:entgrowthbound} \\
        \(\text{Area}[\partial R][\left<T_{tt}\right>/c]^{(d-1)/d}\) & $ \geq  2$ & $n$ strips & quench+ & \eqref{eq:shellnonlinearT} \\
        \(\text{Area}[\partial R] [\left<T_{tt}\right>/c]^{(d-1)/d}\) & $\geq 2$ & ball & quench+ &  \eqref{eq:shellnonlinearT}\\
        \(\text{Vol}[R]\left<T_{tt}\right>\) & $\geq 2$ & strip or ball & quench+ & \eqref{eq:entgrowthbound2}\\
    \hline
        {\boldmath $\partial_{\ell} S \geq $} &  & &  &  \\
    \hline
        $\partial_{\ell}S_{\rm vacuum}$ & $\geq 2$ & $n$ strips & general &  \eqref{eq:dSdx}\\
    \hline
        {\boldmath$\partial_t  \ln\left<\mathcal{W}(C)\right> \leq $} & & \textbf{Length}{\boldmath$[C]$}  & & \\
    \hline
        \(\sqrt{\lambda} \text{Length}[C] \left[\left<T_{tt}\right>/c\right]^{1/2} \) & $4$ & any & general &  \eqref{eq:W4eff} \\
        \(\sqrt{\lambda} \text{Length}[C]^{d-1}\left<T_{tt}\right> + \ldots \) & $3,4 $ & small & general & \eqref{eq:Wsmallr} \\
        \( \sqrt{\lambda} \text{Length}[C]\left[\left<T_{tt} \right>/c\right]^{2/3} \) & $3$ & any & quench+ & \eqref{eq:d3Wilson} \\
    \hline
       {\boldmath $\partial_t
        \ln\left<O(x)O(0)\right> \leq $ } &
         & {\boldmath $|x|$}  & & \\
    \hline
        \( \Delta \sqrt{\left<T_{tt}\right>/c} \) & $2$ & any & general
        & \eqref{eq:finalCorrbnd}\\
        \( \Delta| x|\left<T_{tt}\right>/c+\ldots \) &
        $2$ & small & general & \eqref{eq:OOsmalx} \\
    \hline
       {\boldmath $\partial_{x}
        \ln\left<O(x)O(0)\right> \leq $ } &
         & & & \\
    \hline
       \(\partial_{x} \ln\left<O(x)O(0)\right>_{\rm vacuum}\) & $2$ & any & general
        & \eqref{eq:dlnoo} \\
    \hline
\end{tabular}
\end{table}
In this work we have proven several new upper bounds on the rate of change of
entanglement entropy, spacelike Wilson loops, and equal-time two-point functions
of heavy operators. The proofs apply for spatially homogeneous and isotropic states in strongly coupled CFTs with a holographic dual.
We summarize our bounds in table \ref{tab:bounds}.
We have also provided numerical evidence that the
bounds have broader validity than our proofs. We will now discuss our findings 
and possible future directions.

\textbf{A 2d QWEC:}
The bound \eqref{eq:introthm1} can also be seen as a quantum weak energy
condition (QWEC). Let $S$ be the entropy of be a single interval as a function of one
of the endpoints $p$, so that $\partial_t S$ now refers to
the change of $S$ under the perturbation of this single interval endpoint, rather than both. Then we have
\begin{equation}\label{eq:QWEC}
\begin{aligned}
    \left<T_{tt}\right> \geq \left<T_{tt}\right>_{\rm vac} + \frac{ 3 }{ 2\pi c
    }\left(\partial_t S \right)^2,
\end{aligned}
\end{equation}
while the classical weak energy condition implies that $T_{tt}\geq0$.
Equation~\eqref{eq:QWEC} closely resembles the conformal quantum null energy condition
(QNEC) \cite{Wal11,BouFis15,BouFis15b,AkeKoe16,KoeLei16,BalFau17} in two
dimensions.\footnote{See \cite{MezVir19} for a study on how the QNEC
constrains entanglement growth.}  Consider $2d$ Minkowski space, where $\left<T_{tt}\right>_{\rm
vac}=0$.
Letting $x^{\pm}$ be null coordinates, the conformal QNEC says that \cite{KoeLei16,Wal11}
\begin{equation}\label{eq:QNEC}
\begin{aligned}
    \left<T_{++}\right>|_{p} \geq \frac{ 1 }{ 2\pi }\partial_{+}^2 S + \frac{ 3 }{
        \pi c
    }\left(\partial_{+}S\right)^2.
\end{aligned}
\end{equation}
The structural similarity between \eqref{eq:QWEC} and \eqref{eq:QNEC} is obvious. While \eqref{eq:QWEC} 
does not contain a second derivative, it is in principle possible that \eqref{eq:QWEC} 
could be true also for inhomogeneous states, provided 
we include a term $a \partial_{t}^2 S$ to the right hand side for some fixed
constant $a$. In fact, the conformal QNEC suggests that
$a=(4\pi)^{-1}$, since in the special case of a half-space in a homogeneous state, 
where $\partial_x S=0$, the conformal QNEC and $T\indices{_\mu^{\mu}}=0$ implies
\begin{equation}
\begin{aligned}
    \left<T_{tt}\right>|_{p}\geq \frac{ 1 }{ 4\pi }\partial_t^2 S + \frac{ 3 }{ 2\pi
    c }\left(\partial_t S \right)^2.
\end{aligned}
\end{equation}

\textbf{Why do things fall?}
In \cite{Sus18} it was proposed that the process of gravitational attraction is
dual to the increase of complexity in the CFT. Assuming the complexity=volume
conjecture \cite{StaSus14}, this was given a precise realization in
\cite{BarJos19,BarMar20,BarMar21} (see also \cite{EngFol21}), where it was shown
 that the rate of change of the volume of a maximal volume
slice is given by the momentum integrated on the slice. However, our
formula 
\begin{equation}
\begin{aligned}
\frac{ \dd S_R }{ \dd t } = \int_{X}G t^{a}n^{b}\mathcal{T}_{ab}
\end{aligned}
\end{equation}
shows that change in entanglement can also be seen as directly responsible
for the radial momentum of matter. Thus, at present, ``the increase of entanglement''
seems like an equally good explanation for why things fall. 

\textbf{Relevant scalars, compact dimensions, and DEC breaking:}
Our proofs rely critically on the dominant energy condition -- almost all steps
of the proofs break without it. This rules out having scalars with negative
squared mass, which are dual to relevant operators in the CFT. 
Nevertheless, we found numerical evidence that the bounds hold true without the
DEC, as long as the scalar theories we consider allow a positive mass theorem,
so that AdS is stable and $\left<T_{tt}\right>$ is guaranteed to be positive.

However, there are other reasons to believe that our bounds remain true for these
theories beyond our numerical findings -- at least when working with top-down
theories. Consider working with a theory that is a dimensional reduction and
consistent truncation of type IIA, IIB, or eleven-dimensional SUGRA,
so that any solution can be lifted to solutions on asymptotically AdS$_{d+1}\times K$
spacetimes for some compact manifold $K$. These solutions will typically be
warped products rather than direct products, but there exists significant
evidence \cite{JonMar16} that the entropy computed by the HRT formula in the uplifted spacetime
agrees with the one computed in the dimensionally reduced spacetime -- even when
the product is not direct. But in the uplifted spacetime the DEC holds, since it holds for type II and
eleven-dimensional SUGRA. Thus, if our methods can be generalized to
work for warped compactifications over spherically symmetric AAdS$_{d+1}$ bases, this
appears to be an avenue to prove our bounds even with relevant scalars turned
on. The drawback is that the proofs might have to be carried out separately for
each family of compactifications.

\textbf{Strengthened bounds:}
Our proof that
\begin{equation}\label{eq:outrothm}
\begin{aligned}
    \Big|\frac{ \dd S_R }{ \dd t } \Big| \leq 
    \frac{1}{4}\text{Area}[\partial R] c_{\rm eff}\left[\frac{ 16\pi }{
        (d-1)c_{\rm eff} }\left<T_{tt}\right> \right]^{\frac{ d-1 }{ d }}, 
\end{aligned}
\end{equation}
which implies that $\mathfrak{R}\leq 1$ in neutral states, only applied to thin-shell spacetimes, which are dual to CFT states 
where all dynamics happen at a single energy scale (that evolves with time).
However, we gave numerical evidence that this bound also holds in
general planar symmetric spacetimes with extended matter profiles. A
natural extension of this work is trying to generalize the proof to include
this. This will likely require a better understanding of non-linearities of the
Einstein constraint equations.  

Next, we found that in our numerical maximization of $\mathfrak{R}$ over a ten-parameter family of initial
datasets in $d=3, 4$, that
\begin{equation}
\begin{aligned}
    \mathfrak{R} \leq v_E^{\rm (SAdS) } + \delta v_E,
\end{aligned}
\end{equation}
where 
\begin{equation}
\begin{aligned}
    v_{E}^{\rm (SAdS)} = \sqrt{ \frac{ d }{ d-2 } }\left(\frac{ d-2 }{ 2(d-1) }
    \right)^{\frac{ d-1 }{ d }}
\end{aligned}
\end{equation}
is the entanglement velocity computed in a holographic quench having a neutral
final state, and $\delta v_E$ a
small correction that seemed to depend on the scalar potential, but which was always
small for the theories we studied (less than $0.03$). This hints that it might be
possible to strengthen the prefactor in \eqref{eq:outrothm}.
Similarly, our numerics suggested that the prefactors of \eqref{eq:introthm3}
could be strengthened, and furthermore that this bound is true without
$\mathcal{O}\left(\ell^{d}\left<T_{tt}\right>/c_{\rm eff}\right)$ corrections.

\textbf{$1/N$ corrections:} 
It seems quite likely that $\mathfrak{R}\leq 1$ remains true with perturbative $1/N$
corrections. In fact, the pure
QFT proofs of $\mathfrak{R}\leq 1$ for large subregions \cite{HarJed15,CasLiu15} made no assumption
about large-$N$, so only intermediate and small subregions could be sources
violation. But for small subregions we showed that $\mathfrak{R}\leq
\mathcal{O}(\ell/\beta) \ll 1$
for $\beta$ the effective inverse temperature, which means perturbative $1/N$ corrections are unlikely to
pose a of danger.\footnote{The discovery of quantum extremal surfaces (QES) \cite{EngWal14} far from
classical extremal surfaces \cite{Pen19,AlmEng19} might make this argument
somewhat less convincing, but it seems unlikely that these dominate/exist for small
subregions, whose QES reside close to the conformal boundary.}
For intermediate sized regions things are less clear, but for $d=3,4$ we
numerically did not manage to push $\mathfrak{R}$ close to $1$, hinting that
$1/N$ corrections do not pose a danger in these dimensions.

\textbf{Disentangling the vacuum:} 
Under our assumptions, we have proven for strips that the vacuum-subtracted
entanglement is positive. Thus it is impossible to disentangle the vacuum in a
spatially uniform way without either breaking the existence of a holographic
dual, or turning on operators dual to fields that violate the DEC in the bulk.
In the scenario where classical DEC violation is caused by tachyonic
scalars only, this implies that all uniform states with less entanglement
than the vacuum must have non-zero VEV for some relevant scalar single trace primaries. 

\textbf{Finite coupling:}
Our bounds were proven at strong coupling, but it seems possible that our bounds
survive for arbitrary coupling. In \cite{CasLiu15} it was found that the entanglement
velocity of a free theory (for $d>2$) is strictly smaller than the holographic
strong coupling result, suggesting that dialing up the coupling increases the
capability of generating entanglement. 

\textbf{Primaries close to the unitarity bound:}
In order to turn on bulk fields dual to relevant CFT operators with scaling
dimensions $\Delta$ in the window
\begin{equation}
\begin{aligned}
\frac{ d-2 }{ 2 } \leq \Delta < \frac{ d }{ 2 },
\end{aligned}
\end{equation}
we must consider scalars with masses
\begin{equation}
\begin{aligned}
m_{\rm BF}^2 \leq m^2 < m_{\rm BF}^2 + 1/L^2,
\end{aligned}
\end{equation}
and turn on the slow falloffs rather than fast falloffs (see for example
\cite{HerHor04,HenMar04,HenMar06,HerMae14}). This leads to violation of the falloff assumptions \eqref{eq:falloffs}, and causes 
the ordinary definition of the spacetime mass to be divergent. Then neither of 
the Hawking masses reduce to the CFT energy at conformal infinity. 
Consequently, significant modifications of our proofs would be required. 
The same holds if we turn on sources that perturb us away from
a CFT. Things can get even more challenging, given that for some falloffs
$\partial_t S$ itself might become divergent \cite{MarWal16}. In this case we should only try to
bound finite quantities, like the mutual information or the renormalized entanglement
entropy \cite{LiuMez12,LiuMez13}, where these divergences cancel.

\textbf{Hartman-Maldacena surfaces and half-spaces:}
Our bounds are restricted to subregions that lie in a single
component of the conformal boundary. Our formalism can however readily be
modified to deal with HRT surfaces
corresponding to a CFT region $R$ with connected components in different CFTs,
so that the corresponding HRT surface threads a wormhole, like those studied in
\cite{Har13}.\footnote{See \cite{LiYan22} for a discussion of speed limits for
these surfaces in static black holes.} It would be interesting to see if the bounds remain unchanged. We also only considered strips of finite width, but it is clear
that our formalism can also handle half-spaces, where $\ell=\infty$. Since
\eqref{eq:introthm1} and \eqref{eq:introthm2} do not depend on $\ell$, they
almost certainly remain true for half-spaces.

\textbf{End-of-the-world branes:} Our bounds are proven with the assumption that
the Cauchy slices are complete manifolds, which rules out spacetimes with
end-of-the-world (EOW) branes. However, for HRT surfaces that do not end on the
EOW brane, our proofs survive as long as we assume that the Hawking mass of the brane is
positive, which ensures that Lemma~\ref{lem:omega0} remains true. In the case where the
surfaces end on the EOW brane, we expect the analysis to look similar to the
analysis for HRT surfaces threading a wormhole. 

\textbf{Vaidya has optimal entanglement growth: } There is a sense in
which Vaidya spacetimes are maximizing the growth of entanglement for some given
distribution of bulk energy density $\mathcal{E}(r)$, within the class DEC respecting planar spacetimes. In \eqref{eq:Ecallower}, we saw that the DEC imposed that the energy
density at a point was lower bounded by the momentum-density the same
point (up to a factor). As is clear from our proofs, excess energy density above the lower bound contributes to increasing
the CFT energy without increasing $\partial_t S$, and so we see that
spacetimes saturating the DEC have the greatest $\partial_t S$ compatible with
their distribution of energy density. But a direct computation shows that all 
Vaidya spacetimes saturate the DEC. This makes sense, since we have shown
that infalling matter leads to entanglement growth and outgoing matter
gives entanglement decrease, and in Vaidya no matter is wasted on being outgoing. While
this is no proof, it suggests that speed limits that hold in all Vaidya
spacetimes also hold in other planar-symmetric spacetimes respecting
the DEC.

\textbf{Inhomogeneous states:}
It would be interesting to prove bounds for inhomogeneous states,
but this appears to be a difficult task, given how crucial the bulk planar
symmetry was for our proofs. Furthermore, our bounds imply a strengthening of 
the positive mass theorem (PMT) for planar symmetry, and if generalizations of our results exist
without the uniformity assumption, it appears likely that these result will
strengthen the general PMT for AAdS spacetimes. Given how challenging it was to prove the
PMT \cite{SchYau81,Wit81}, and how the Penrose inequality, which is the most
famous strengthening of the PMT, still does not have a
general proof, this is a daunting task. 

\textbf{Bounds with charge:} 
It is a persistent finding that $U(1)$ gauge fields tend to slow down the growth 
of extremal surfaces of various dimensions \cite{LiuSuh13a,LiuSuh13b,Mez16,CarCha17}.
It thus seems plausible that our bounds can be strengthened by taking into
account nonzero charges in the CFT. 
It is suggestive that, in spherical symmetry, $U(1)$ gauge fields
contribute energy density, but they have no pure contribution to the momentum
density -- that is -- they only contribute to $\mathcal{J}$ through gauge covariant
derivatives acting on other matter fields. 

\textbf{Approximate bounds for large regions:}
To derive our bounds scaling with the volume of the region $R$, we used the
inequality
\begin{equation}\label{eq:r0uboundappendix}
\begin{aligned}
\frac{ L^2 }{ r_0 } \leq \eta_d \ell,
\end{aligned}
\end{equation}
where $\eta_d$ is a numerical constant. Assume now that we consider a state that has
volume law scaling for the entropy when $\ell$ is roughly above some length
scale $\beta$. For $\ell \gg \beta$, as we increase $\ell$, we expect $r_0$ to saturate at some
radius $r_{\rm barrier}$, and so \eqref{eq:r0uboundappendix} becomes a very poor bound. 
It seems likely that we can get stronger approximate bounds in this limit, by
relating $r_{\rm barrier}$ to an effective inverse temperature $\beta_{\rm eff}$, so that
we effectively get a bound like $|\partial_t S|\leq \# \text{Area}[\partial
R]\beta_{\rm eff} \left<T_{tt}\right>$. This bound might or might not reduce to \eqref{eq:outrothm}. 

\textbf{A numerical laboratory:} Our formalism makes it very easy to numerically
compute properties HRT surfaces in planar symmetric spacetimes and test various hypothesis about their
behavior without constructing full spacetimes and having to deal with the
evolution of the Einstein equations. We imagine this can be used as a laboratory
to learn more about the properties of HRT surfaces. 

\textbf{Other boundary geometries:} Except for $d=2$, our proofs always assumed
Minkowski space on the boundary. However, our bounds survive if we
compactify on a torus, so the boundary geometry is $\mathbb{R}\times T^{d-1}$,
provided we make a few additional assumptions. For the bounds of the type $|\partial_t S|
\leq \kappa \text{Vol}[R]\left<T_{tt}\right>+\ldots$ (and the similar bounds for
Wilson loops), we should always consider regions less than half the system size --
otherwise $\text{Vol}[R]$ should be replaced with the volume of the complement. For
bounds with multiple strips, if we have a torus, we need to make sure that the
entangling surfaces are all parallel, which happens automatically in Minkowski
due to the parallel postulate. 

Next, our proofs for single regions should imply
growth bounds for CFTs on the static cylinder $\mathbb{R}\times S^{d-1}$, as long as we take 
the regions to be very small compared to the curvature radius of the boundary sphere.
The results for balls will translate to results for small caps,
while the result for strips will translate to results for thin belts around the
equator.

\section*{Acknowledgments}
It is a pleasure to thank Sean Colin-Ellerin, Netta Engelhardt, Gary Gibbons,
Matt Headrick, Hong Liu, Dan Roberts, Jon Sorce, and Brian Swingle for
discussions. A.D. was supported by the University of Minnesota Doctoral
Dissertation Fellowship and by the National Science Foundation Graduate Research
Fellowship under Grant No. 00039202. The research of \AA{}.F. is
supported in part by the John Templeton Foundation via the 
Black Hole Initiative, NSF grant no. PHY-2011905, and an Aker
Scholarship. The research of \AA{}.F. was also supported in part by the
Heising-Simons Foundation, the Simons Foundation, and NSF grant no. PHY-1748958.

\appendix

\section{Appendix}\label{sec:appendix}
\subsection{The mean curvature of $X$ in $\Sigma$}\label{app:meancurv}
Let now $A, B, \ldots$ be indices for tensors on $X$, and $\alpha, \beta,
\ldots$ be indices for tensors on $\Sigma$, and consider intrinsic coordinates on $\Sigma$ and $X$
from Sec.~\ref{sec:Xlocsol}. The induced metrics are
\begin{equation}
\begin{aligned}
    \dd s^2|_{\Sigma} &= H_{\mu\nu}\dd x^{\mu}\dd x^{\nu} = B(r)\dd r^2 + r^2\left[\dd\phi^2 + \delta_{ij}\dd
    x^i \dd x^{j}\right] \\
    \dd s^2|_{X} &= \gamma_{AB}\dd y^A \dd y^B = \left[B(r)+r^2 \Phi'(r)^2\right]\dd r^2 + r^2 \delta_{ij}\dd
    x^i \dd x^{j}.
\end{aligned}
\end{equation}
A basis of tangent vectors to $X$ in $\Sigma$ is $\{e^{\alpha}_A\}$, with coordinate expressions
\begin{equation}\label{eq:tangents}
\begin{aligned}
    e^{\mu}_r &= (1, \Phi(r), 0), \\
    e^{\mu}_{i} &= (0, 0, \delta^{\alpha}_{i}).
\end{aligned}
\end{equation}
The normal to $X$ inside $\Sigma$ reads
\begin{equation}\label{eq:normalExpr}
\begin{aligned}
    n_{\mu} = \sqrt{ \frac{ B r^2 }{ B + r^2 (\Phi')^2 } }\left(\Phi'(r), -1,
    0\right)
\end{aligned}
\end{equation}
With this in hand, we can compute the mean curvature of $X$ in $\Sigma$:
\begin{equation}
\begin{aligned}
    \mathcal{K} &= \gamma^{AB}e^{\mu}_A e^{\nu}_B \nabla_{\mu}
    n_{\nu} \\
    &= \frac{ 1 }{ r \sqrt{B}\left[B+r^2 (\Phi')^2\right]^{3/2}}
    \left[ (d-1)r^3 \Phi'(r)^3 + \left(d B - \frac{ 1 }{ 2 }r B'\right)r \Phi' +
    r^2B\Phi''\right].
\end{aligned}
\end{equation}

\subsection{Explicit form of $K-n^{\alpha}n^{\beta}K_{\alpha\beta}$}\label{app:tExtremality}
Noting that
\begin{equation}
\begin{aligned}
    n^{\mu} = \sqrt{ \frac{ B r^2 }{ B + r^2 (\Phi')^2 }
    }\left(\frac{\Phi'(r)}{B(r)}, -\frac{ 1 }{ r^2 },
    0\right)
\end{aligned}
\end{equation}
we get
\begin{equation}
\begin{aligned}
    K &= H^{\alpha\beta}K_{\alpha\beta} = \frac{ 1 }{ B }K_{rr} + \frac{ d-1 }{
        r^2
    }K_{\phi\phi}, \\
    K_{\alpha\beta}n^{\alpha}n^{\beta} &= \frac{ Br^2 }{ B+r^2\Phi'(r)^2 }\left[
        \frac{ \Phi'(r)^2 }{ B^2 }K_{rr} +
    \frac{ 1 }{ r^4 }K_{\phi\phi}\right].
\end{aligned}
\end{equation}
Inserting $K_{rr}(r) = B(r) F(r)$ and doing some algebra,
$K-n^{\alpha}n^{\beta}K_{\alpha\beta}=0$ becomes \eqref{eq:extremalitycond}.

\subsection{Deriving formulas for $\partial_t S$ and $\partial_{\ell}S$ }\label{app:dSKphi}
\subsubsection*{A formula for $\partial_t S$}
In this section, we show that the entropy growth is proportional to the infalling matter flux. We will first need to prove that
\begin{equation}
    \begin{aligned}
        \dv{\text{Area}[X_t]}{t}\eval_{t=0} &= \int_{\partial X}N^a\eta_a = -\frac{\text{Area}[\partial R]}{L^{d-2}}\lim_{r\to\infty}r^{d-3}K_{\phi\phi}
    \end{aligned}
\end{equation}
For the calculation, we will construct the vector $\eta^a$, which is
tangent to the boundary $\partial \mathcal{M}$, and $N^a$, the outwards unit
normal to $\partial\Sigma=\partial\mathcal{M}\cap\Sigma$ in $\Sigma$, in a
coordinate system. To do so, introduce the ADM coordinates adapted to the
extended homology hypersurface
\begin{equation}\label{eq:AdMcoords}
    \dd s^2|_{\mathcal{M}} =  -r^2\dd \tau^2 +
    H_{\mu\nu}(\tau,x)\dd
    x^\nu\dd x^\nu,
\end{equation}
where we took the shift to be vanishing, and the lapse to be $r$. $x^{\mu}=(r,
\bm{x})$ are the coordinates on $\Sigma$, and $H_{\mu\nu}(\tau=0, x)$ its induced
metric, given by \eqref{eq:cancoords2}. The extinsic curvare of $\Sigma$ reads
\begin{equation}
\begin{aligned}
K_{\alpha\beta} = \frac{ 1 }{ 2r }\partial_{\tau}H_{\alpha\beta}|_{\tau=0}.
\end{aligned}
\end{equation}

Imagine now we have the coordinates $z^i=(t,{\bm x})$ on $\partial \mathcal{M}$
and take $\partial\Sigma$ to be located at $(r=r_c, t=\tau=0)$ with a temporary
cutoff $r=r_c$. We want to find embedding coordinates $(r(t),\tau(t))$ for $\partial \mathcal{M}$ such that the induced metric reads
\begin{align}
    \dd s^2|_{\partial \mathcal{M}} = h_{ij}\dd z^i \dd z^j =
    \frac{r_c^2}{L^2}\left[-\dd t^2 + L^2\dd \phi^2 + L^2\dd{\bm x}^2\right]
\end{align}
The components of the induced metric then satisfy 
\begin{align}
    h_{tt} &= g_{rr}\dot{r}^2 - r^2 \dot{\tau}^2 = -\frac{r_c^2}{L^2},\\
    h_{\phi\phi} &= g_{\phi\phi} = r_{c}^2.
\end{align}
Taking the derivative of the second equation, and then setting $t=0$ gives a
system of equations that is easily solved to give (see the appendix of \cite{EngFol21})
\begin{align}
    \dot{\tau}(0) &= \frac{1}{L\sqrt{1-B K_{\phi\phi}^2 r^{-2}}}\eval_{r=r_c},\\
    \dot{r}(0) &= -\frac{K_{\phi\phi}}{L\sqrt{1-B K_{\phi\phi}^2
    r^{-2}}}\eval_{r=r_c},
\end{align}
where we have chosen the branch $\dot{\tau}>0$. Thus, $\eta^a$ in our
ADM coordinate system reads
\begin{equation}
    \eta^a = (\partial_t)^a = \dot{\tau}(0) (\partial_t)^a + \dot{r}(0)(\partial_r)^a.
\end{equation}
Now, the tangents to $\partial X$ are tangent to $\partial R$, and the sole
remaining tangent vector $e_r^\alpha$ in \eqref{eq:tangents} is then the normal
to $\partial\Sigma$. Hence, up to a normalization $C$, 
\begin{equation}
\begin{aligned}
    N^{\mu} = C e^{\mu}_r = C(1, \Phi', 0),
\end{aligned}
\end{equation}
which can be unit normalized and pushed
forward to a spacetime vector yielding (in the coordinates \eqref{eq:AdMcoords}),
\begin{equation}\label{eq:Naexplicit}
\begin{aligned}
    N^{a}= \frac{ 1 }{ \sqrt{B(r) + r_c^2 \Phi'(r)^2} }(0, 1,
    \Phi'(r), 0)|_{r=r_c}.
\end{aligned}
\end{equation}
We can now compute the integral on the cutoff regulated $\partial X$:
\begin{equation}\label{eq:etaNcomp}
    \begin{aligned}
    \int_{\partial X}\eta^a N_a 
    &= - \frac{ K_{\phi\phi} }{ L\sqrt{1 - B(r_c) K_{\phi\phi}^2/r_c^2} } \times
        \frac{B(r_c)}{\sqrt{B(r_c) + r_c^2 (\partial_r \Phi)^2} } \times
        r_c^{d-2} \int \dd^{d-2}\bm{x} \\
    &= -\frac{ K_{\phi\phi}\sqrt{B(r_c)} }{ L\sqrt{1 - B(r_c)
        K_{\phi\phi}^2/r_c^2} } \times   \sqrt{1-(r_0/r_c)^2} \times r_c^{d-2}
        \frac{\text{Area}[\partial R]}{L^{d-2}},
    \end{aligned}
\end{equation}
where we have used the differential equation for the embedding function $\Phi(r)$. In the large $r$ limit, the asymptotic behaviors are
\begin{equation}
    B(r) \sim \mathcal{O}(r^{-2}), \hspace{.2in} K_{\phi\phi}\sim\mathcal{O}(r^{-(d-3)}).
\end{equation}
Taking the cutoff to the boundary, one finds the area growth to be given by
\eqref{eq:dAdt}.

\subsubsection*{A covariant formula for $\partial_t S$} 
Now let us write the intergral formula for $\partial_t S$ in a covariant way.
Letting $t^a$ be the future unit normal to $\Sigma$ and $n^a$ the outwards
normal to $X$ tangent to $\Sigma$, we have for some function $G$ on $X$ that only depends on $r$:
\begin{equation}
\begin{aligned}
    8\pi G_N \int_{X} G(r) t^{a}n^{b}\mathcal{T}_{ab} &= \int \dd^{d-2}\bm{x} \int_{r_0}^{\infty}\dd r
    r^{d-2}\sqrt{B + r^2 (\Phi')^2} G(r) n^{r} \mathcal{J} \\
    &= \frac{ \text{Area}[\partial R] }{ L^{d-2} }\int\dd rr^{d-2}\frac{
        \sqrt{B}r }{ B + r^2 (\Phi')^2 }\Phi' G(r) \mathcal{J} \\
    &= \frac{ \text{Area}[\partial R] }{ L^{d-2} }\int\dd rr^{d-2} \left(\frac{
        r_0}{ r } \right)^{2d-2} G(r) \mathcal{J}\sqrt{\left(r/r_0\right)^{2d-2} -1} \\
    &= \frac{ \text{Area}[\partial R] }{ L^{d-2} }\int\dd r \frac{ r_0^{d-1} }{
        r^{d}} G(r) \mathcal{J}\sqrt{r^{2d-2} -r_0^{2d-2}},
\end{aligned}
\end{equation}
where we used \eqref{eq:normalExpr} and \eqref{eq:phisol}. Letting
\begin{equation}
\begin{aligned}
G(r) = \frac{2\pi r^{d}}{(d-1)r_0^{d-1}},
\end{aligned}
\end{equation}
we get a covariant formula for the entropy growth
\begin{equation}
    \dv{S_R}{t} = \int_X G t^an^b\mathcal{T}_{ab}
\end{equation}

\subsubsection*{A formula for $\partial_{\ell}S$}
Consider a one-parameter family of HRT surfaces $X_{\ell}$
anchored at the strip region region $R_{\ell}$, given by \eqref{eq:stripregiondef}, but
letting now $\ell$ vary, holding $t$ fixed. Taking the vector field $\eta^a$
generating the flow of $\partial X_{\ell}$ to be $\eta^a|_{\partial
\mathcal{M}} = \frac{ 1 }{ L }(\partial_{\phi})^a$, \eqref{eq:Neta} becomes
\begin{equation}
\begin{aligned}
    \frac{ \dd }{ \dd \ell }\text{Area}[X_{\ell}] = \frac{ \text{Area}[\partial R] }{
    L^{d-1}} \lim_{r \rightarrow \infty} r^{d-2} g_{\phi\phi}N^{\phi}.
\end{aligned}
\end{equation}
Using \eqref{eq:Naexplicit}, \eqref{eq:AdMcoords}, and \eqref{eq:phisol}, this
evaluates to
\begin{equation}
\begin{aligned}
    \frac{ \dd }{ \dd \ell }\text{Area}[X_{\ell}] &= \frac{ \text{Area}[\partial R] }{
        L^{d-1}} \lim_{r \rightarrow \infty } r^{d} \frac{ \Phi' }{ \sqrt{B + r^2 (\Phi')^2}  }
    &= \frac{ \text{Area}[\partial R] }{
        L^{d-1}}r_0^{d-1}.
\end{aligned}
\end{equation}

\subsection{Expression for $\partial_t \norm{X_t}$}\label{app:dSKphi2}
The derivation of \eqref{eq:dXdt} is almost identical to the derivation in
Sec.~\ref{app:dSKphi}. Let us just highlight what must be changed. First, we do not
have an explicit formula for $\Phi'(r)$, but this does not matter, since
everything we need is its rate of falloff, which we can read off to be
$\Phi'(r) \sim \mathcal{O}\left(1/r^{q+2}\right)$ from \eqref{eq:Phiq}. Next, in
\eqref{eq:etaNcomp}, it is sufficient to replace $r_c^{d-2} \rightarrow
r_c^{q}$. After doing that, and taking into account that $K_{\phi\phi}$ now has
falloff $\mathcal{O}(1/r^{q-1})$, the $r_c \rightarrow \infty$ limit of
\eqref{eq:etaNcomp} with these modifications gives \eqref{eq:dXdt}.

\subsection{Geometric properties of $X$ and $\Sigma$}\label{app:nothroat}
Let us now prove various properties of extended homology hypersurfaces. We give
the proof for HRT surfaces of strips and comment how the proofs are modified for
$(q+1)$-dimension extremal surfaces anchored at spheres.

\begin{lem}\label{lem:appnothroat}
    The extended homology hypersurface $\Sigma$ of a strip region in a spacetime
    with planar symmetry cannot have a throat in its interior, i.e., a radius
    where $B(r)$
    diverges. 
\end{lem}
\begin{proof}
Assume for contradiction that $\Sigma$ has a throat $T$ in its interior  -- if there are multiple, take $T$ to be the outermost one. 
Since $\Sigma$ by definition terminates at the plane tangent to the tip of $X$, this means that $X$ must pass beyond the throat. 
But if $X$ crosses the throat, there must be a point on $X \cap T$ where $X$ is not tangent to $T$. 
Let now $U \subset \Sigma$ be the region outside $T$, which we
can always cover with a coordinate system
\begin{equation}
\begin{aligned}
    \dd s^2|_{U} = B(r)\dd r^2 + r^2 \dd \bm{x}^2, \quad r \in [r_T, \infty), \quad B(r_T)=\infty, \quad B(r>r_T) < \infty.
\end{aligned}
\end{equation}
where the throat is at $r=r_T$. The fact that $X$ is not tangent to $T$ means that $|\Phi'(r_T)|<\infty$. 
    But the solution for the extremal surface reads
\begin{equation}
\begin{aligned}
    \Phi'(r)^2 = \frac{ B(r) }{ r^2(r^{2d-2} c - 1) }
\end{aligned}
\end{equation}
for some constant $c$, and so $|\Phi'(r_T)|=\infty$, which is a contradiction. Hence
$T$ cannot exist in the interior of $\Sigma$. 
\end{proof}
This proof goes through the case of $X$ anchored at a dimension $q$
sphere, as discussed in Sec.~\ref{sec:generaldim}, simply taking $\Phi'$ to be computed from \eqref{eq:Phiq}.

Next, we have the following:
\begin{lem}
The HRT surface $X$ of a strip region $R$ can only have one turning point.
    \end{lem}
\begin{proof}
Assume for contradiction that $X$ has multiple turning points. Then there must be at least one turning
point of $X$ in the interior of $\Sigma$ that is a local maximum of the
    embedding function $r(\phi)$. Let $\phi=\phi_t$ where
this turning point occurs, and let us restrict our attention to a neighbourhood $\phi \in \mathcal{O}_{\epsilon} = (\phi_t, \phi_t +
    \epsilon)$, where we can invert $r(\phi)$ to get $\Phi(r)$, describing the
    embedding in $\mathcal{O}_{\epsilon}$. We see that $\Phi'(r)<0$ and $\Phi''(r)<0$ 
    in this neighbourhood. Now, the equation for the HRT surface in the neighbourhood $\mathcal{O}_{\epsilon}$ is
    \begin{equation}
        (d-1) r^3 \Phi'^3 + r^2 B \Phi ''+r \Phi ' \left(d B-\frac{r}{2} B'\right) = 0.
    \end{equation}
    Since we are in the interior of $\Sigma$, $B$ is bounded on
    $\mathcal{O}_{\epsilon}$ by Lemma~\ref{lem:appnothroat}. Since $\Phi'$ diverges at the
    turning point, the equation for $\Phi'$ near the turning point reads
    \begin{equation}\label{eq:turningMastereq}
    \begin{aligned}
        r^2 B \Phi'' = -(d-1) r^3 \Phi'^3\left[1 + \mathcal{O}\left( \frac{ 1 }{ (\Phi')^2 }\right)\right],
    \end{aligned}
    \end{equation}
    where the correction can be neglected to arbitrarily good precision since $B$ is bounded. But this implies that 
    $\Phi''$ and $\Phi'$ must have opposite signs in $\mathcal{O}_{\epsilon}$ for sufficiently small $\epsilon$, which
    is a contradiction. Hence $X$ can only have one turning point.
\end{proof}
We did not consider the case where $r'(\phi_t)=r''(\phi_t)=0$, but this was shown to
be ruled out by \cite{EngFis15}. Also, note that this proof survives the case of
spherical boundary anchoring and a $(q+1)$-dimensional extremal surface, since the new term
$(\Phi')^3 \chi(r)$ in the extremality equation \eqref{eq:qexeq} is subleading
at the prospective turning point, since it scales like $\chi \sim 1/\Phi'$.
Thus, \eqref{eq:turningMastereq} remains true (up to a numerical factor and an
$\mathcal{O}(1/\Phi')$ correction).

\subsection{General extremality conditions}\label{sec:codimqextremal}
Let
\begin{equation}
\begin{aligned}
    H_{ab} &= g_{ab} + t_{a}t_{b}
\end{aligned}
\end{equation}
be the induced metric on $\Sigma$. Then we have $h^{ab} = H^{ab}-
\delta^{IJ}n_{I}^a n_{J}^{b}$, and so 
\begin{equation}
\begin{aligned}
    \mathcal{K}^{0} &= (H^{ab} - \delta^{IJ}n_{I}^a n_{J}^{b}) \nabla_a t_{b} \\
    &= K  - \delta^{IJ}n_{I}^a n_{J}^{b} \nabla_{(a} t_{b)} \\
    &= K  - \delta^{IJ}n_{I}^a n_{J}^{b} K_{ab} \\
\end{aligned}
\end{equation}
where we in the last line used that $n^a_{I} n^b_{J}$ is tangent to $\Sigma$, which projects out the difference
$\nabla_{(a}t_{b)} - K_{ab}$. Now, a collection of tangents $e^{\mu}_I$ to $X$ are
\begin{equation}
\begin{aligned}
    e\indices{^{\mu}_r} &= (1, \Phi'(r), 0, 0), \\
    e\indices{^{\mu}_{i}}&= (0, 0, \delta^{\mu}_{i}, 0),
\end{aligned}
\end{equation}
where $i$ runs over sphere directions. The second to last slot here runs over
the sphere directions, with the last slot runs over the $\bm{z}$-directions.
The unit normals to $X$ (in $\Sigma$) are
\begin{equation}
\begin{aligned}
    n\indices{^{r}_{\mu}} &= (\alpha, \beta, 0, 0),  \\
    n\indices{^{i}_{\mu}} &= (0,0,0, r\delta_{\mu}^{y^{i}}),
\end{aligned}
\end{equation}
for some $\alpha, \beta$ we now work out, and where the coordinates are with
respect to the index $\mu$ on $\Sigma$. $r, i$ should be view as indices in the orthonormal
tangent basis labeled by $I$ on $X$.
Now
\begin{equation}
\begin{aligned}
    0 &= n\indices{^{r}_{\mu}}e\indices{^{\mu}_r} = \alpha + \beta \Phi'(r), \\
    1 &= H^{\mu\nu}n\indices{^{r}_{\mu}}n\indices{^{r}_{\nu}} =  \frac{ \alpha^2 }{ B} + \frac{
        \beta^2 }{ r^2 }.
\end{aligned}
\end{equation}
Solving for $\alpha, \beta$, we get
\begin{equation}\label{eq:explicitNormals}
    \begin{aligned}
        n\indices{^{r}_{\mu}} = \sqrt{\frac{ Br^2 }{
        B + r^2 (\Phi')^2 }}\left(\Phi', -1, 0, 0\right), \quad 
    n\indices{^{r\mu}} = \sqrt{\frac{ Br^2 }{
        B + r^2 (\Phi')^2 }}\left(\frac{ 1 }{ B }\Phi', -\frac{ 1 }{ r^2 }, 0,
    0\right).
\end{aligned}
\end{equation}
Now we want to impose 
\begin{equation}\label{eq:extCond}
\begin{aligned}
 \mathcal{K}^{0}= K - \delta^{IJ}n_I^a n_J^b K_{ab} = 0.
\end{aligned}
\end{equation}
Explicitly we have
\begin{equation}
\begin{aligned}
    K &= H^{\mu\nu} K_{\mu\nu} = \frac{ 1 }{ B }K_{rr} + \frac{ 1 }{ r^2
    }K_{\phi\phi} + \frac{ w^{ij} }{ r^2 \phi^2 } \times \underbrace{\phi^2
    w_{ij}  K_{\phi\phi}}_{K_{ij}} +
    \frac{d-2-q}{r^2} \underbrace{K_{\phi\phi}}_{K_{z z}}, \\
    &= \frac{ 1 }{ B }K_{rr} +  \frac{ d-1 }{ r^2 }K_{\phi\phi},
    \label{eq:Ktrgeneralq}
\end{aligned}
\end{equation}
and
\begin{equation}
\begin{aligned}
    \delta^{IJ}n^{a}_I n^{b}_J K_{ab} &= (n\indices{^{rr}})^2K_{rr} +
    (n\indices{^{r\phi}})^2
    K_{\phi\phi} + (d-2-q)\frac{ 1 }{ r^2 }K_{\phi\phi}.
\end{aligned}
\end{equation}
Thus, with $K_{rr}=BF$, condtion \eqref{eq:extCond} reads
\begin{equation}
\begin{aligned}
    \left[1 - B (n^r_r)^2\right]F + \left[\frac{ q+1 }{ r^2 } -
    (n^{\phi}_r)^2 \right]K_{\phi\phi} = 0
\end{aligned}
\end{equation}
Using the explicit formula for $n^r_r, n^{\phi}_r$ in
\eqref{eq:explicitNormals}, inserting the explicit formula for $\Phi(r)$
from \eqref{eq:Phiq}, and solving for $F(r)$, we find \eqref{eq:Fsolq}.

Note also that thanks to \eqref{eq:Ktrgeneralq}, \eqref{eq:thetapmform} is
unchanged, and so the proof that $\pm\theta_{\pm}[\partial \Sigma] \geq 0$ for
strips survive for dimension--$q+1$ surfaces anchored at $q$-spheres.

\bibliographystyle{jhep}
\bibliography{all}

\end{document}